\documentclass{llncs}

\usepackage{makeidx}  
\usepackage{amsmath}
\usepackage{amsfonts}
\usepackage{graphicx}
\usepackage{mathrsfs}
\usepackage{txfonts}
\usepackage{amssymb}
\begin{document}

\title{Complementary cooperation, minimal winning coalitions, and power indices\thanks{This research is supported by the  973 Program(NO.2010CB731405) and NNSF of China (NO.71101140).}
}

\author{Zhigang Cao,  Xiaoguang Yang
}


\institute{ Key Laboratory of Management, Decision \& Information Systems,\\ Academy of Mathematics
and Systems Science, Chinese Academy of Sciences,
Beijing, 100190, China.\\
zhigangcao@amss.ac.cn, xgyang@iss.ac.cn
}
\maketitle

\begin{abstract}We introduce a new simple game, which is referred to as the complementary
weighted multiple majority game (C-WMMG for short). C-WMMG models a basic cooperation rule, the complementary cooperation rule, and can be taken as a sister model of the famous weighted  majority game (WMG for short). In C-WMMG, each player is characterized by a nonnegative vector with a fixed dimension, and players in the same coalition cooperate by producing a characteristic vector for this coalition (each dimension of this vector equals the maximum of the corresponding dimensions of its members). The value of a coalition is 1 if and only if the sum of its characteristic vector is larger than that of its complementary coalition, in which case the coalition is called winning. Otherwise, the coalitional value is 0.
In this paper, we concentrate on the two dimensional C-WMMG. An interesting property of this case is that there are at most  $n+1$ minimal winning coalitions (MWC for short), and they can be enumerated in time $O(n\log n)$, where $n$ is the number of players. This property guarantees that the two dimensional C-WMMG is more handleable than WMG. In particular, we prove that the main power indices, i.e. the Shapley-Shubik index, the Penrose-Banzhaf index, the Holler-Packel
index, and the Deegan-Packel index, are all polynomially computable.  To make a comparison with WMG, we know that it may have exponentially many MWCs, and none of the four power indices is polynomially computable (unless P=NP). Still for the two dimensional case, we show that
local monotonicity holds for all of the four power indices. In WMG, this property is possessed by the Shapley-Shubik index and the Penrose-Banzhaf index, but not by the Holler-Packel index or the Deegan-Packel index.
Since our model fits very well the cooperation and competition in team  sports, we hope that it can be potentially applied in measuring the values of players in team sports, say help people give more objective ranking of NBA players and select MVPs,  and consequently bring new insights into  contest theory and the more general field of sports economics. It may also provide some interesting enlightenments into the design of non-additive voting mechanisms. Last but not least, the threshold version of C-WMMG is a generalization of WMG, and natural variants of it are closely related with the famous airport game and the stable marriage/roommates problem.

\end{abstract}

\section{Introduction}

Cooperation makes human being flourish. For other animals, competition plays an absolutely dominant role, and cooperation is never so deep or broad \cite{f03}. Game theoretical models, both the coalitional ones and the strategic ones, are basic tools for studying cooperation. Generally speaking, strategic models are suitable for scenarios where cooperation is necessary but seems impossible in one-shot games, that is, Nash equilibrium is not Pareto optimal. So the main research focus in strategic form games on cooperation is {\it whether} cooperation is attainable in sequential games (either repeatedly or not), and how to promote cooperation (cf. \cite{ah81,dk04,JBT10,KMW82,n85,n06}).  Coalitional models, in contrast,  are suitable for scenarios where the attainability of cooperation is not a problem (because it is usually assumed that there is a powerful arbitrator), but {\it how} to cooperate is, i.e. how to share the benefit  gained by cooperation {\it fairly} among players \cite{ps07}.

However, the concept of fairness is so ambiguous
and debatable that there could be numerous impossibility theorems in coalitional
game theory. This fact might have been neglected by most researchers in coalitional game theory. There does exist one common spirit behind fairness: the more you contribute the more you should get. How to
measure the contributions of players, nonetheless, is not an easy job.  And in fact, it is as tricky as how to interpret fairness. So this common spirit helps us very little, if any. Another issue is that in coalitional game theory,
It is often very difficult to distinguish whether a study is normative or positive, and fair
allocation is also frequently interpreted as possible allocation that will occur among ¡°rational¡±
players, and this allocation is determined by bargaining powers of players. Consequently, benefit allocating, contribution measuring, and power measuring are usually triune in the study of coalitional games.

There are mainly two reasons why it is so hard to have a completely
satisfactory allocation rule. The first comes from the
self-contradictions of the concept of fairness, which  are not easy to be noticed by intuition. And at the same time people are trying to design rules that are consistent with every aspect of  it.
The second is due to that in coalitional game theory people usually try to design a
general allocation rule suitable for all possible coalitional games (with very weak
restrictions, say super-additivity), which is a work that will never be accomplished. And this is why there can be numerous impossibility theorems in the coalitional game theory.

That ``the cooperative side of game theory has been dominated by the noncooperative side, at least judging from their respective influence on mainstream economics"  is  an undeniable fact (Maskin, 2003, \cite{m03}). Further discussion of this problem is beyond the task of this paper (and beyond the ability of the authors),  but based on the arguments in the above two paragraphs, one of the right ways we think is to give up the dream of finding a universally perfect solution, forget the concept of fairness, only discuss various properties (e.g. various consistencies, monotonicities, cf. \cite{ps07}), and leave the choosing right of which solution is more suitable to the game theory users and the real players.  And what's more, pay more attention to more concrete models. We are happy to see that this is also what many researcher are doing today in this field.

Perhaps the most extensively studied and most successfully used concrete model in the coalitional game theory is the weighted majority game (WMG for short). The main contribution of this paper is to propose a new coalitional game, which can be taken as a sister of WMG.
\subsection{Additive cooperation VS. complementary cooperation\label{s1.1}}
One of the main reasons why WMG  receives so much attention and finds successful applications may come from the fact that it is both simple and fundamental. Simpleness implies richness, because only simple rules can exist behind wide range of situations, and only simple rules can be understood and used well by human being. Simpleness of this model is obvious. It is also very fundamental, because it models a very basic cooperation rule.

Suppose party A has 100 votes, and party B has 30 votes, then cooperation between the two parties means that they possess totally 130 votes.  And with 130 votes, they may beat party C with 120 votes, which seems no possible at all if they do not cooperate.   You have a strength of 50 Kg, I have 40 Kg, working together we can move a stone which is too heavy for any single of us. You have 1000 dollars, I have 1500, pooling them we may buy a car that neither of us can afford separately. No other cooperation rule, as far as we can imagine, could be simpler than this. Since the process of this kind of cooperation can be nicely characterized by the mathematical operation of  addition, we name it the {\it additive cooperation rule} \footnote {Another choice is to name it as the ``substitute cooperation",  as we shall name its counterpart as the ``complementary cooperation". However, we finally abandoned this more fashionable term, because ``substitute" and ``complementary" have very rich implications in economics that are quite popular but significantly different from what we intend to express.}.

There is another cooperation rule, which is also very simple and fundamental. Suppose country A can produce 10 units cattle or 3 units wheat in one year, while in one year country B can produce 4 units cattle or 9 units wheat. Then country A has absolute advantage in producing cattle and country B in wheat. Therefore,  their maximum overall productivity is $\max\{10,4\}=10$ units cattle  and $\max\{3,9\}=9$ units wheat, which can be achieved by letting each country producing which she has advantage. It is absurd to think that the two countries' total productivity is $10+4=14$ units cattle and $9+3=12$ units wheat.

In basketball, there are five positions: point guard, shooting guard, small forward, power forward, and center. Each basketball player can play any of the five positions with varied levels. For a basketball team, its level in each of the five positions is determined by the highest level, rather than the sum, of all its players. And the overall level of a team is determined by certain aggregation (say, taking average, or equivalently sum) of its five separate levels. For football, things are similar, except that there are eleven positions instead of five (some of the positions are identical).

Cooperations in the above two examples have obvious differences with the additive cooperation rule. (i) It has multiple rather than single roles.  (ii) The most efficient cooperation strategy  is to let each role be played by the one who can play it best. We call this kind of cooperation  the {\it complementary cooperation rule}.

The complementary cooperation rule prevails not only in economy, team sports (not including tug of war),  but also in many other fields of the society. In fact, it embodies very well the division of labor, and thus plays a fundamental role in civilization. The functioning of the society depends severely on numerous complementary cooperations. Even marriage can be taken as a kind of complementary cooperation. And so is sexual reproduction for any kind of species that is using it. Another striking example of complementary cooperation that is observed frequently in biology is symbiosis \footnote {This concept can be found in any textbook about biology. For a recent theoretical study of symbiosis in bacteria,  see Katsuyama et. al.  \cite{knt09}. The term of ``complementary cooperation" was coined by them, independently with us. And this is why we noticed that paper.}.

To a great extent, additive cooperation is quantitative change, and complementary cooperation is structural and qualitative change. In this sense, we can safely argue that complementary cooperation is at least as fundamental as additive cooperation.

\subsection{Virtual integration and the story of Dell}
Let's see a real case of Dell in this subsection. This case is quite famous in the field of management science (for more detailed information, see \cite{m98}), and  embodies the essence of complementary cooperation very well.

During the inception of the computer industry, each company had to produce all the computer
components: it had to manufacture the disk drives, memory chips, monitors, application softwares,
and even its own graphics chips. And of course it had to do its own marketing. When Dell entered
this industry, it took quite a different strategy: Dell did not create any component itself, but
bought them from its partners. All they did was to concentrate on a brand-new marketing mode, the
nowadays well known direct mode, which made it so close not only to the customers but also to the
suppliers that its founder Michael Dell called Dell and its partners as ``virtually integrated". This
new kind of cooperation was remarkably close to such an extent that virtually integrated companies
could be seemed in a great sense as a big whole company. For a simple instance, the cooperation
with the monitor supplier Sony was roughly like this: Sony put the brand of Dell on each monitor,
and Dell told the third-party logistics company UPS to pick up the exact number of monitors they
needed every day from the plant of Sony and to deliver them directly to the customers. No
warehouse, no inventory, no inefficient transportation, and even no quality testing. All liked in
the same company. As we have seen,  virtual integration made a great success, helping Dell grow
into a \$12 billion company in just 13 years (1984-1997), and also bringing its partners large
profit.

From the aspect of division of labor, what Dell and its partners did is quite easy to understand:
Intel is expert in making CPUs, Sony in Monitors, and Dell in marketing. Working together, they
provide excellent computers for customers and beat their competitors. This kind of cooperation, i.e. virtual integration, is in fact a specialization that penetrates the whole economy of our society and plays an essential role
in the civilization of the human being, and it
can be well modeled by the complementary cooperation rule that we have discussed in the last subsection.

\subsection{Profit allocation, players evaluation, and benchwarmers' contribution\label{s1.3}}
Profit allocation (or cost sharing) is the key topic in cooperative game theory, and is also one of the critical issues in complementary cooperation. In the example of producing two goods, what is the fairest way to allocate the final goods? In virtual integration,   how to share the total profit? Although this issue is usually rather complicated in the real world, and perhaps people have already found a way out (e.g. in the example discussed in the previous subsection, Dell played a dominant role, and hence the real story is that it negotiated one by one with its potential partners), and the word ``fairness" is really ambiguous, theoretical research is still meaningful.

For the example of basketball, profit allocation turns out to be players evaluation. This is a meaningful problem in both theory and practice. For an NBA final champion team, it is interesting to know who contributes most and if every player is worth his salary.  Different people may have quite different ideas. Is there any more {\it objective} way to do this? Can we give a convincing way to rank all the players according to their values? Is it possible to give some {\it scientific} suggestions on how to select the MVP or how to trade players more wisely? All these problems fall into the topic of players evaluation, and in some sense they are equivalent to the problem of profit allocation, the focus of the cooperative game theory (in WMG, as well as C-WMMG, the model proposed by this paper, this is equivalent to players' power measuring. The famous formulas for measuring players' power are called power indices, which will be introduced very soon in Subsection \ref{pi}). Of course, we should not expect that one simple theory can completely solve all the real world complicated problems, but theoretical discussion is the necessary first step and it may play a crucial role in the final solution.

One interesting property of the example of basketball is that the evaluation of benchwarmers is not that obvious to study, compared with the key players frequently appearing on the court. As mentioned in Subsection \ref{s1.1}, at any fixed time there can be at most five players on the court for any basketball team. However, we all know that usually a team has far more than five players in total. The immediate reason for keeping benchwarmers is that  key players may suffer from injuries frequently, and even without injury no player can always play 48 minutes every game. Another reason that is not so obvious is that even if a benchwarmer is not good enough to play any position at all in team A, s/he may be able to help team B to beat team A. Just consider the situation that the center of team B really sucks that s/he becomes a bottleneck. Although team B is stronger than team A in all the other positions, still it can not beat team A. With the joining of that benchwarmer from team A, who is not really very good but much better than the current center in team B, however, it is very likely for team B to beat team A. In this situation, it is a very natural strategy for team A to keep that benchwarmer by giving him/her a high enough salary. And s/he deserves it, probably without doing anything at all except for applauding.

It is a pity that we failed to find a nice case study in NBA, or FIFA, that illustrates the importance of a benchwarmer in so extreme case that we discussed. Although there are examples where very lucky benchwarmers got the champion ring, their importance in the competing teams is not obviously high either. In general market competition, there is a well-known aggressive and effective strategy: first hire more than necessary employees (usually top talents)  from your competitors with, say double or triple salaries, then,  after beating them, sack all the redundant employees, i.e. benchwarmers. Although this strategy is quite controversial and likely to have legal problems, it demonstrates clearly that the importance of benchwarmers has been realized very well by business people.

Therefore, to evaluate the accurate contributions of employees/players and  offer corresponding salaries, not only what they have done should be calculated, but also what they did not do, i.e. their abilities of potential harm to the organization by hopping to the competitors, should be carefully considered. As we shall see soon, this factor is fully recognized by the main power indices in the study of game theory.

\subsection{Pure complementary cooperation}
 Although they  are both complementary cooperations, the examples of production and basketball still have several significant differences.

 (i) Let's consider the situation where the number of players is greater than that of roles. In basketball, there can be at most five players cooperating for each team on the court at any fixed time. For the example of producing two products, the cooperation of three players is meaningful, because it is possible for two of them to produce the same kind of production and the overall productivity of this product is the sum of the two players' productivity.

 (ii) For the basketball, it is possible for any player to play multiple roles on the court (consider the case where one player fouls out). This is not the case, however, for the production example. Suppose country A is more productive than country B in both products. Then country A should still focus on producing one product and country B on the other one in which it has {\it comparative advantage}, and therefore cooperation is still possible. The critical reason is that each country has limited time for production.
This theory, first described by D. Ricardo in 1817 and usually called the {\it law of comparative advantage}, is one of  the most fundamental principals of economics (cf. \cite{m06}).

Based on the above discussions, the complementary cooperation in basketball games seems much cleaner than that in production games. It is not satisfactory to use complementary cooperation alone to model the  production example. In fact, as argued in point (i) above, the cooperation in the production game is a mixture of additive cooperation and complementary cooperation. The cooperation in the basketball, however, can be perfectly described by the complementary cooperation rule alone. This kind of {\it pure} complementary cooperation is the focus of this paper.  Of course it can be argued that the positions of the basketball game are actually not completely independent, what players really do is usually not that clear-cut,  and perhaps more importantly, whenever a player has to play double positions (when one of his partners fouls out), his actual levels should be reduced in each of these positions. We have to say that all these considerations are correct. However, as a theoretical study, we need to do some necessary assumptions to make the problem clean enough. We believe that our assumptions are mild and acceptable.

We shall begin our study with a model that is as simple as possible as long as it embodies the essence of (pure) complementary cooperation, and thus concentrate on {\it simple games} \cite{tz99}, i.e. coalitional games where each coalitional value is either 1 or 0 (and with several plain restrictions). Although this assumption seems at the first glance too simple and even artificial, it models the nature of cooperation: by cooperation people can do something that may never be done by any single person. Back to the field of team sports, our assumption is that the unique objective of all teams is the final champion.  Of course, {\it externality} in this model will be inevitable, because otherwise benchwarmers will be useless. That is to say, our model, strictly speaking,  is not of the characteristic function form, but of the more general partition function form \cite{lt63}. However, this partition function is a quite trivial one. In fact, it can be  reduced to a characteristic function (as we shall see in the subsection below).  And therefore  we shall take the terms of characteristic function form coalitional game theory. To study more general models (say runner-up is also valuable, besides champion), terms of partition function form coalitional game theory may be inevitable.
\subsection{Model description and basic notations}
Formally, we are given a set of players $N=\{p_1,
p_2, \cdots, p_n\}$. Each player $p_j ~(1\leq j\leq n)$ is characterized by a $d_0$ dimensional non-negative
 vector $(p_{j}^1, p_{j}^2,\cdots, p_{j}^{d_0})$, which is referred to as her {\it characteristic vector}.  For any coalition $C$, i.e. a subset of $N$, we denote by $$(q_1(C), \cdots, q_{d_0}(C))=\left(\max\{p_j^1:p_j\in C\},\cdots,\max\{p_j^{d_0}:p_j\in C\}\right)$$ its
characteristic vector, and $$q(C)=\sum\limits_{i=1}^{d_0}q_i(C)$$ its {\it competitive power}.

We say that a coalition $C$ is a {\it winning coalition} (also abbreviated as WC) iff $q(C)>q(C^{-})$, where $C^-=N\setminus C$. The value of a coalition is 1 if it is a winning coalition, and 0 otherwise.

We shall refer to this model as  {\it the complementary weighted multiple majority game}, because it has a symmetric relation with a special case of the well-known weighted multiple majority game (to be introduced in the next subsection). It is straightforward that C-WMMG is super-additive, though its {\it core} is usually empty.

A winning coalition  is called a {\it minimal winning coalition} (MWC for short) if it does  not have any proper subset that is also winning.
We also use ${\mathcal W}{\mathcal C}$ to denote the set of WCs, and ${\mathcal M}{\mathcal W}{\mathcal C}$ the set of MWCs.

 The following two concepts are from the classical literature of WMG. For any winning coalition
$C$, $p_j\in C$ is called a $decisive$ player (w.r.t. $C$) if $C\setminus \{p_j\}$ is losing. If
$C\setminus \{p_j\}$ is still winning, $p_j$ is called a {\it null} player (w.r.t. $C$). We shall also interchangeably use the term {\it indecisive}. Trivially,
all the players in an MWC are decisive.

 In this paper, we shall concentrate on the two dimensional C-WMMG.

The next two concepts are new. For any coalition $C$ and $k\in \{1, 2\}$, denote by $A_k(C)=\{p_j\in C: p_j^k=q_k(C)\}$ the set of players whose $k$-th dimensions are the largest among players in $C$,  and by
$B(C)=A_1(C)\cup A_2(C)$ the set of players who have at least one dimension that is largest among players in $C$.

We call players in $B(C)$ $\it busy$ players (w.r.t. $C$) and players in
$C\setminus B(C)$ {\it benchwarmers} (w.r.t. $C$).

Notice that it is possible that $C$ has more
than two busy players, and it is also possible that $C$ has only one busy player.

We also let
$A(C)=A_1(C)\times A_2(C)=\{(p_{j_1}, p_{j_2}): p_{j_1}\in A_1(C), p_{j_2}\in A_2(C)\}$ be the set of {\it best pairs}. Members of $A_1(C)$ and $A_2(C)$ will also be referred  to as {\it busy-1} players and {\it busy-2} players (w.r.t. $C$), respectively.

All the above terms and notations will be used throughout this paper.

\subsection{Mathematical relation with the weighted multiple majority game\label{s1.4}}
The weighted majority game (a.k.a. weighted voting game, weighted simple game, weighted threshold game), first formulated by von Neumann (1944, \cite{nm44}), is one of the most intensively studied cooperative game models.

A WMG is usually
represented by $G=(q; a_1, a_2, \cdots, a_n)$, where $q$ is called the quota and $a_j$ the weight
of player $p_j$. The characteristic function is defined as: $v(C)=1$ iff $\sum_{p_j\in C}a_j\geq
q$; $v(C)=0$ otherwise. The special case of $q=\lfloor (\sum_{p_j\in
N}a_j)/2\rfloor +1$, i.e. $v(C)=1$ is equivalent to $\sum_{p_j\in C}a_j>\sum_{p_j\in C^{-}}a_j$, is  called the {\it simple weighted majority game}.

Let $G_1=(q_1; p^1_1, p^1_2, \cdots, p^1_n)$ and $G_2=(q_2; p^2_1, p^2_2, \cdots, p^2_n)$ be two
WMGs with identical player set, whose characteristic functions are $v_1$ and $v_2$, respectively. $G=G_1+G_2$, the {\it sum} of $G_1$ and $G_2$,  is a simple game with characteristic function $v$, defined as $v(C)=1$ iff at least one of $v_1(C)$ and $v_2(C)$ equals 1.
 The sum of $k$ WMGs is usually called a $k$ dimensional
WMMG, which can be defined similarly.

Suppose that $G=G_1+G_2+\cdots +G_{d_0}$ is the sum of $d_0$ simple WMGs, i.e. $q_i=\lfloor (\sum_{p_j}p^i_j)/2\rfloor +1$ for all $1\leq i\leq d_0$. It is obvious that $v(C)=1$ is equivalent to \begin{equation}
\max\limits_{1\leq i\leq d_0}\left\{\sum_{p_j\in C}p^i_j-\sum_{p_j\in
C^-}p^i_j\right\}>0.\end{equation} While in the model of this paper, $v(C)=1$ is equivalent to
\begin{equation}\sum_{1\leq i\leq d_0}\left(\max\limits_{p_j\in C} \{p^i_j\}-\max\limits_{p_j\in C^-} \{p^i_j\}\right)>0.\end{equation}

It can be observed that there is some symmetry between the two models, and this is one of the reasons why we named our model
as C-WMMG. To have more knowledge about WMG and WMMG, please refer \cite{tz99}.

\subsection{Main power indices\label{pi}}
How to measure powers  of players in WMG is a very interesting  topic that has rich applications in politics, because the direct measurement by weights is not reasonable at all.
Ever since the seminal work of  Shapley and Shubik \cite{ss54}, this topic keeps attracting scientists from political science, law, game theory, and computer science. As discussed in Subsection \ref{s1.3}, this is also the main focus of this paper. There are typically four widely accepted approaches, namely the  Shapley-Shubik index, the Penrose-Banzhaf index, the Holler-Packel index, and the
Deegen-Packel index.

Recall that  ${\mathcal W}{\mathcal C}$ is the set of WCs, and ${\mathcal M}{\mathcal W}{\mathcal C}$ the set of MWCs. The Shapley-Shubik index \cite{ss54} is a direct application of the celebrated Shapley value \cite{s53} on WMGs. It argues that the bargaining power of each player
is equal to her expected marginal contribution, assuming that players arrive one by one in a random order.
For all $p_j\in N$, let  \begin{equation}{\mathcal W}{\mathcal C}(j)=\{C\in {\mathcal W}{\mathcal C}: C\ni
p_j, C\setminus \{p_j\}\notin {\mathcal W}{\mathcal C}\},\end{equation} i.e. the set of winning coalitions where $p_j$ is {\it decisive}. The exact definition of the Shapley-Shubik index is as follows:
\begin{equation}ss_j=\sum_{C\in{\mathcal W}{\mathcal C}(j)
}\frac{(|C|-1)!(n-|C|)!}{n!},\label{ss}\end{equation}
where the $|\cdot|$ is the cardinality of any set and will be used throughout this paper.

As unearthed by Felsenthal and Machover \cite{fm98,fm01,fm05}, the first theoretical study of power measurement in WMG should be attributed to Penrose \cite{p46}. Please refer Felsenthal and Machover for more details where they discuss measuring power
as a way of measuring influence, or as a way of measuring the rewards of winning.

Banzhaf \cite{b65}, many years after Penrose first published his
work, proposed an equivalent voting power index.  The Penrose-Banzhaf index assumes that all coalitions are equally likely to form. The exact definition is as follows:

\begin{equation}pb_j=\frac{1}{2^{n-1}}|{\mathcal W}{\mathcal C}(j)|.\end{equation}

Unlike the above two indices determined by the chance of being decisive, the Holler-Packel index (or Public Good Index, see \cite{h82,hp83}) and the
Deegen-Packel index \cite{dp79} argue that only MWCs will possibly form. This argument is initially known as the Riker's {\it size principle}, see
\cite{fb96,R62}. Let \begin{equation}{\mathcal M}{\mathcal W}{\mathcal C}(j)=\{C\in {\mathcal M}{\mathcal
W}{\mathcal C}: C\ni p_j\},\end{equation} the definitions of the (un-normalized) Holler-Packel index and the (un-normalized) Deegen-Packel index are as follows:
\begin{equation}hp_j=|{\mathcal M}{\mathcal W}{\mathcal C}(j)|,\end{equation}

\begin{equation}dp_j=\sum_{C\in {\mathcal M}{\mathcal W}{\mathcal C}(j)}\frac{1}{|C|}.\end{equation}

\subsection{Contribution and organization of this paper}
Because there has already been several very famous power indices for measuring values of players in WMG (as shown in Subsection \ref{pi}), our objective of this paper is not to propose  any new measurement especially for C-WMMG, but naturally to study the properties of the old ones on C-WMMG. This can serve at least as a first step of the study of C-WMMG. In particular, much effort of this paper is devoted to the computational issue, as this has been always  one of the core problems for the model of WMG and has been attracting a large amount of attention ever since the beginning of this field. The computational issue is also one of the core focuses for the booming relatively new research field {\it algorithmic game theory} \cite{nrtv07}. As one of the main power indices is simply an application of the famous Shapley value, which is one of the few key concepts in coalitional game theory, widely used in many fields, but hard to compute in general, one of our findings that the Shapley value can be efficiently computed for the two dimensional C-WMMG has its own theoretical interest to the field of algorithmic game theory. In fact, as far as we know, the two dimensional C-WMMG  is the only nontrivial special case where the Shapley value can be efficiently computed, besides the classical model of {\it the airport game} \cite{lo73,l74}.

Our main findings are summarized as follows (we remind the reader again that we restrict our discussion in this paper to the two dimensional C-WMMG, unless explicitly stated otherwise.):

(i) The structure of the set of MWCs is quite simple. To be precise,  the total number of MWCs in the two dimensional C-WMMG can be  upper-bounded by $n+1$ (Theorem \ref{t1}, Section \ref{s3}), and the whole set of MWCs  can be computed in time $O(n\log n)$ (Theorem \ref{t2}, Appendix \ref{s4}).

It is well known that the set of MWCs  plays a fundamental rule for any simple game, because it determines the complete structure of this game. In fact, the most general form of a simple game is defined as the set of its MWCs \cite{tz99}. The structure of the set of MWCs for WMG is generally very complicated \cite{fb96}. Even in simple WMGs, and hence in higher dimensional cases, the number of MWCs can be exponentially large.  This can be verified easily by a trivial example with an odd number of players and all players having identical weight. In fact, any coalition with $(n+1)/2$ members is an MWC, and there are $C^{(n+1)/2}_n\approx 2^n \sqrt{\frac{2n}{\pi(n^2-1)}}$ such coalitions (by Stirling's formula).
For simple WMGs, it is even NP-hard to compute the minimum MWC, i.e. the one with the smallest total weight, which can be easily proved by reduction to the partition problem \cite{j73}.

 Although  the two dimensional C-WMMG has a symmetric relation with the sum of two simple WMGs as shown in Subsection \ref{s1.4}, it has a much simpler and thus theoretically more handleable structure. It still remains open whether it is still so for the three or higher dimensional cases.

(ii) Each of the four power indices can be computed in polynomial time.

Due to the simple structure of MWCs (Appendix A), this result is not surprising for the Holler-Packel index or the Deegan-Packel index, because their definitions are based on MWCs (Appendix B). The Penrose-Banzhaf index and the Shapley-Shubik index, however, are both based on winning coalitions, which can be exponentially many. We show that the structure of the set of winning coalitions is not complicated either. In fact, they can be determined by critical ones that are not so many and thus can be represented implicitly in polynomial time (Appendix C). Consequently, the Penrose-Banzhaf index and the Shapley-Shubik index can both be efficiently computed (Appendix D).
 In contrast, none of the four power indices is polynomially computable for WMG, unless P=NP.

(iii) Local monotonicity holds for each of the four power indices (Section \ref{s5}).

 There are many widely discussed monotonicity concepts for power indices of WMG, e.g. {\it the local monotonicity, the re-distribution monotonicity, the new member monotonicity, the bloc monotonicity}, etc.. They are also known as paradoxes and postulates, see \cite{fm98,hn04} for extensive study. The local monotonicity says that players with
larger weights should have more power than the ones with smaller weights. We focus above all on local monotonicity for C-WMMG, rather than any other monotonicity concept, because we think that this concept is the most natural one among all the similar concepts, and thus is the first property that we {\it expect} any power index to possess. In WMG, we know that  local monotonicity is  possessed by the Shapley-Shubik index and the Penrose-Banzhaf index, but not by the Holler-Packel index or the Deegan-Packel index \cite{hn04}.

 For the Shapley-Shubik index and the Penrose-Banzhaf index, local monotonicity holds trivially for any dimensional C-WMMG.  We show that local monotonicity  holds for the Holler-Packel index and the Deegan-Packel index in the two dimensional C-WMMG (Theorem \ref{t3}), but not in the three or higher dimensional cases (Example \ref{eg5}).

We remind the reader that we didn't say that possessing the property of local monotonicity, or any other monotonicity, is a necessary or good thing for any power index and failing to possess it is bad. We agree with Holler and  Naple \cite{hn04} that the violation of local monotonicity is not a fatal drawback in power measuring at all, as taken for granted by many researchers.   And perhaps it might be this violation that reflects the strategic powers of seemingly weak players.

The rest of this paper is organized as follows. The next section gives  a very short literature review. Section \ref{s3} is for the structure of the set of MWCs and the upper bound of its cardinality. Local monotonicity is discussed in Section \ref{s5}. All the computational technical details, which are theoretically not very hard but really complicated, ugly, and rather tedious for the reader, are moved to appendices. To be specific, Appendix A is for computing the MWCs, Appendix B for the Holler-Packel index and the Deegan-Packel index, Appendix C for the structure of winning coalitions, and Appendix D for the Penrose-Banzhaf index and the Shapley-Shubik index. In the discussions of Appendix A to Appendix D, quite a few notations are used. To help the reader not get lost, we provide a list of notations in Appendix E. Section 5 concludes this paper by pointing out several interesting topics for further research.

\section{Literature Review\label{s2}}

The classical and traditional literature on WMG is too vast for us to give a complete review, we refer the reader to the standard book of Taylor and Zwicker \cite{tz99}. We shall mainly concentrate on the computational side of this field.

Ever since Owen's seminal work on computing the Shaply-Shubik index through multilinear extension of a game \cite{o72}, the computational issue of power indices for WMGs  has been attracting a lot of attention of scholars from game theory. Literature following Owen's idea can be found in Alonso-Meijide et. al. \cite{achf08}. The generating function method is another important, but relatively new, way to computing power indices, see Algaba et. al. \cite{ab03}. Fatima et. al. \cite{fwj08} provided another algorithm, which runs in linear time, as in Owen's multilinear extension algorithm, but has a smaller approximation error (on average).

In the last decade, this problem also began to attract the attention of scholars from theoretical computer science and operations research. Below is a short review.

 Matsui and Matsui \cite{mm00,mm01} proved that the problems of computing the Shapley-Shubik
index, the Penrose-Banzhaf index, and the Deegan-Packel index in WMG, are all NP-hard, and there are
pseudo-polynomial time dynamic programming algorithms for them. Deng and Papadimitriou
\cite{dp94}, who pioneered in the study of computational complexities for cooperative solution
concepts, also proved that it is \#P-complete to compute the Shapley-Shubik index.  Matsui and
Matsui \cite{mm01} also observed that Deng and Papadimitriou's proof can be easily carried over
to the problem of computing the Penrose-Banzhaf index. For  the special case of simple WMG, it can be observed from their proofs
that all the above complexity results hold. It is still not hard to prove that computing the
Holler-Packel index in WMG, even in the  special case of simple WMG, is  NP-hard, and the algorithm of
Matsui and Matsui \cite{mm01} for computing the Deegan-Packel index can trivially be modified to
compute the Holler-Packel index, with an even lower time complexity \cite{cy09}.

Other work on the computation of power indices includes
\cite{achf08,l03}. For the computational issues of more solution concepts on WMG, see Elkind et. al. \cite{eggw08,ep08,gmps11}. For the recent research on the structure of MWCs in WMG, see \cite{ar10}. The Penrose-Banzhaf index and the Shapley-Shubik index can also be calculated using MWCs, see \cite{kl10}.  Studying quota manipulating problems in WMGs is also a promising new direction \cite{zfbe12}.

\section{Structure of the Set of MWCs and the Upper Bound of its Cardinality\label{s3}}

Please recall all the terms and notations in Subsection 1.3. We alert the reader that almost all the notations, except the small ones such as $i,j,k,t$ that are indicating player identities or dimensional names, are global instead of local. They are valid throughout  this paper right after it is defined (appendices included).

First of all, we present a warmup example, which
shows that, very interestingly,  a busy player in a winning coalition $C$ may not be decisive. At the
same time, however, $C$ may have a decisive player who is a benchwarmer.

\begin{example} There are four players in total: $p_1=(3, 3), p_2=(4, 0), p_3=(0, 2), p_4=(5,
0)$. Then $C=\{p_1, p_2, p_3\}$ is a winning coalition. Obviously, $p_2$ is busy in $C$ but not
decisive, while $p_3$, a benchwarmer, is decisive in $C$.
\end{example}

\begin{lemma}In C-WMMG, each winning coalition contains at least one $MWC$.\label{l1}\end{lemma}

\begin{proof} This property is straightforward as for each winning
coalition we can dump the null players, if any, one by one until all of the remaining players are
decisive. \qed
\end{proof}

\begin{lemma}In the two dimensional C-WMMG, if $C$ is a winning coalition, then either $A_1(N)\subseteq C$ or
$A_2(N)\subseteq C$.\label{l2}\end{lemma}
\begin{proof}Suppose on the contrary that neither relation holds. Then $C^-\cap A_1(N)\neq \emptyset$ and $C^-\cap A_2(N)\neq \emptyset$, which imply that $q_1(C^-)=q_1(N)$ and $q_2(C^-)=q_2(N)$, and hence $q(C^-)=q(N)\geq q(C)$, which contradicts the fact that $C$ is a winning coalition. \qed\end{proof}

\subsection{Structure in the simple cases}

We discuss the simplest case first.
\begin{lemma}In the two dimensional C-WMMG, suppose $A_1(N)\cap A_2(N)\neq\emptyset$.

(a) If  $A_1(N)=A_2(N)$, then ${\mathcal M}{\mathcal W}{\mathcal C}=\{A_1(N)\}=\{A_2(N)\}$;

(b1) If  $A_1(N)\subset A_2(N)$, then ${\mathcal M}{\mathcal W}{\mathcal C}=\{A_1(N)\}$;

(b2) If  $A_2(N)\subset A_1(N)$, then ${\mathcal M}{\mathcal W}{\mathcal C}=\{A_2(N)\}$;

(c) If  $A_1(N)\cap A_2(N)\neq\emptyset$ but none of the conditions in (a) (b1) (b2) is true, then ${\mathcal M}{\mathcal W}{\mathcal C}=\{A_1(N),A_2(N)\}$.
\label{l03}\end{lemma}
\begin{proof}(a) Since $q(A_1(N))=q(N)$ and $q(N\setminus A_1(N))<q(N)$, we know that $A_1(N)$ is a winning coalition. Obviously, it is also minimal winning, because for any player $p_j\in A_1(N)$, we have $q(\{p_j\})=q(N)$, and hence $q(A_1(N)\setminus \{p_j\})\leq q((N\setminus A_1(N))\cup\{p_j\})$.
Due to Lemma \ref{l2}, it is the only MWC.

(b1) First of all, $A_1(N)$ is a winning coalition, because $q(A_1(N))=q(N)$ while $q(N\setminus A_1(N))<q(N)$  (note that $q_1(N\setminus A_1(N))<q_1(N)$). Since for each $p_j\in A_1(N)$, we have $q(\{p_j\}\cup (N\setminus A_1(N))\})$=$q(\{p_j\}\cup \{A_2(N)\setminus A_1(N)\})=q(N)\geq q(A_1(N)\setminus \{p_j\})$, $A_1(N)$ is also minimal winning.
Due to Lemma \ref{l2}, it is the only MWC.

(b2) is symmetric to (b1).

(c) $A_1(N)$ is a winning coalition, because $q(A_1(N))\geq q(A_1(N)\cap A_2(N))=q(N)>q(N\setminus A_1(N))$. That the conditions in (a) and (b1) are not true implies that $A_2(N)\setminus A_1(N)\neq \emptyset$, and hence for each $p_j\in A_1(N)$ we have $q(\{p_j\}\cup (N\setminus A_1(N)))=q(N)\geq q(A_1(N)\setminus \{p_j\})$. Due to the above argument, we know that $A_1(N)$ is a minimal winning coalition. Using the same argument, we know that $A_2(N)$ is also a minimal winning coalition. Due to Lemma \ref{l2}, they are the only MWCs, and hence the lemma. \qed\end{proof}


\begin{definition}We denote by $m_1$ and $m_2$ the number of busy-1 players and the number of busy-2 players of the grand coalition $N$, respectively, i.e. \begin{equation}m_1=|A_1(N)|,\end{equation} \begin{equation}m_2=|A_2(N)|.\end{equation}\end{definition}

\begin{lemma}In the two dimensional C-WMMG, suppose $A_1(N)\cap A_2(N)=\emptyset$.

 (a1) When $A_1(N)$ is a winning coalition:  \begin{equation}{\mathcal M}{\mathcal W}{\mathcal C}=\left\{\begin{array}{ll} \Big\{A_1(N)\Big\}&if~m_1=1\\
 \Big\{A_1(N)\Big\}\cup \Big\{A_2(N)\cup \{p_i\}: p_i\in A_1(N)\Big\}&if~m_1>1\end{array}\right.;\end{equation}

 (a2) When $A_2(N)$ is a winning coalition:  \begin{equation}{\mathcal M}{\mathcal W}{\mathcal C}=\left\{\begin{array}{ll} \Big\{A_2(N)\Big\}&if~m_2=1\\
 \Big\{A_2(N)\Big\}\cup \Big\{A_1(N)\cup \{p_i\}: p_i\in A_2(N)\Big\}&if~m_2>1\end{array}\right..\end{equation}
\label{l04}\end{lemma}

\begin{proof}Since the two parts of this lemma are symmetric, we only prove part (a1). If $m_1=1$, then any coalition containing the only player of $A_1(N)$ is winning, and any coalition not containing her is losing, therefore $A_1(N)$ is the only MWC. We suppose in the rest of this proof that $m_1>1$.  Because $A_1(N)$ is a winning coalition, and  $A_1(N)\cap A_2(N)=\emptyset$, we know that it must be minimal winning, and therefore $A_1(N)\in {\mathcal M}{\mathcal W}{\mathcal C}$, and it is the only MWC that contains $A_1(N)$. Let $C$ be an arbitrary MWC that does not contain $A_1(N)$, i.e. $A_1(N)\nsubseteq C$. By Lemma \ref{l2}, we know that $A_2(N)\subseteq C$. The hypothesis that $A_1(N)$ is a winning coalition implies that $C\cap A_1(N)\neq \emptyset$, because otherwise $C$ will not be winning. We claim that $C$, if exists, contains exactly one player from $A_1(N)$. This is true because if it contains more than one players from $A_1(N)$, any of these players will not be decisive. Thus we have now $\Big\{A_2(N)\cup \{p_i\}: p_i\in A_1(N)\Big\}$ is the set of all potential MWCs that does not contain $A_1(N)$. Using the hypothesis that $m_1>1$, it is easy to check that each coalition in $\Big\{A_2(N)\cup \{p_i\}: p_i\in A_1(N)\Big\}$ is indeed an MWC, and hence the lemma.\qed\end{proof}

\subsection{Structure in the complex case}
Let's now discuss the more general case.

\begin{definition}Based on Lemma \ref{l2}, we divide ${\mathcal M}{\mathcal W}{\mathcal C}$ into three sub-collections:
\begin{equation}{\mathcal M}{\mathcal W}{\mathcal C}1=\Big\{C\in
{\mathcal M}{\mathcal W}{\mathcal C}: A_1(N)\subseteq C, A_2(N)\cap C=\emptyset\Big\},\end{equation}
\begin{equation}{\mathcal M}{\mathcal W}{\mathcal C}2=\Big\{C\in {\mathcal M}{\mathcal W}{\mathcal C}: A_2(N)\subseteq C, A_1(N)\cap C=\emptyset\Big\}, \end{equation}
\begin{equation}{\mathcal M}{\mathcal W}{\mathcal C}3=\Big\{C\in {\mathcal M}{\mathcal W}{\mathcal C}: A_1(N)\cap C\neq \emptyset, A_2(N)\cap C\neq
\emptyset\Big\}.\end{equation}\end{definition}

The structure of ${\mathcal M}{\mathcal W}{\mathcal
C}3$ in the general case is simple.
\begin{lemma}In the two dimensional C-WMMG, suppose $A_1(N)\cap A_2(N)=\emptyset$, and neither $A_1(N)$ nor $A_2(N)$ is a winning coalition, then
\begin{equation}{\mathcal M}{\mathcal W}{\mathcal
C}3=
\left\{\begin{array}{ll}\Big\{A_1(N)\cup \{p_i\}: p_i\in A_2(N)\Big\}& if~m_1=1\\
\Big\{A_2(N)\cup \{p_i\}: p_i\in A_1(N)\Big\}&if~m_2=1\\
\bigcup\limits_{p_i\in A_1(N)} \Big\{A_2(N)\cup \{p_i\}\Big\}\cup \bigcup\limits_{p_i\in A_2(N)}\Big\{A_1(N)\cup \{p_i\}\Big\}&if~m_1>1,m_2>1\end{array}\right..\label{mwc3}\end{equation}

Consequently, \begin{equation}|{\mathcal M}{\mathcal W}{\mathcal
C}3|=
\left\{\begin{array}{ll}m_2& if~m_1=1\\
m_1&if~m_2=1\\
m_1+m_2&if~m_1>1,m_2>1\end{array}\right..\label{b012}\end{equation}\label{l05}\end{lemma}
\begin{proof}Notice first,  by Lemma \ref{l2} and the hypothesis that neither $A_1(N)$ nor $A_2(N)$ is a winning coalition, that in any of the three cases, it holds that \begin{equation}{\mathcal M}{\mathcal W}{\mathcal
C}3\subseteq \bigcup\limits_{p_i\in A_1(N)} \Big\{A_2(N)\cup \{p_i\}\Big\}\cup \bigcup\limits_{p_i\in A_2(N)}\Big\{A_1(N)\cup \{p_i\}\Big\},\end{equation}
and when $m_1>1,m_2>1$, the above inclusion relation holds as equality.

When $m_1=1, m_2>1$, $\bigcup\limits_{p_i\in A_1(N)}\Big\{A_2(N)\cup \{p_i\}\Big\}$ has only one member, $A_1(N)\cup A_2(N)$, which is not minimal winning, because no player in $A_2(N)$ is decisive (remember the hypothesis that $A_1(N)\cap A_2(N)=\emptyset$). Every coalition in $\Big\{A_1(N)\cup \{p_i\}: p_i\in A_2(N)\Big\}$, however, is obviously an MWC. Therefore in this case ${\mathcal M}{\mathcal W}{\mathcal
C}3=\Big\{A_1(N)\cup \{p_i\}: p_i\in A_2(N)\Big\}$, which still holds when $m_1=1$ and  $m_2=1$. Hence the case $m_1=1$ is valid, and the case $m_2=1$ is also true by symmetry. (\ref{b012}) is straightforward.
\qed\end{proof}

Due to Lemma \ref{l03}, Lemma \ref{l04} and Lemma \ref{l05},  the following upper bound for the cardinality of ${\mathcal M}{\mathcal W}{\mathcal C}3$ is immediate.
\begin{lemma}In the two dimensional C-WMMG, $|{\mathcal M}{\mathcal W}{\mathcal C}3|\leq m_1+m_2$. \qed\label{l06}\end{lemma}

To analyze ${\mathcal M}{\mathcal W}{\mathcal
C}1$ and ${\mathcal M}{\mathcal W}{\mathcal
C}2$, we need more notations. 

\begin{definition}We denote by $M$ the set of benchwarmers in the grand coalition $N$, i.e. \begin{equation}M=N\setminus
B(N).\end{equation}\end{definition}

\begin{definition}$\forall p_i\in M$, let ${\mathcal M}{\mathcal W}{\mathcal C}1_{i}$ be the collection of MWCs where the set of busy-1 players are $A_1(N)$ and player $p_i$ is a busy-2 player, i.e.  \begin{equation}{\mathcal M}{\mathcal W}{\mathcal C}1_{i}=\{C\in {\mathcal M}{\mathcal W}{\mathcal C}:
A_1(C)=A_1(N), p_i\in A_2(C)\}.\label{mwc1i}\end{equation}\end{definition}

Obviously, when $A_1(N)$ is not winning, we have \begin{equation}{\mathcal M}{\mathcal W}{\mathcal C}1=\bigcup_{p_i\in M}{\mathcal M}{\mathcal W}{\mathcal C}1_{i}.\end{equation} 

We claim that \begin{equation}|{\mathcal M}{\mathcal
W}{\mathcal C}1_{i}|\leq 1.\end{equation}

 In fact, ${\mathcal M}{\mathcal
W}{\mathcal C}1_{i}$ has only one potential member that is defined below.
 
 \begin{definition}We define $C_{1i}$ as the only potential member of ${\mathcal M}{\mathcal
W}{\mathcal C}1_{i}$:
 \begin{equation}C_{1i}=A_1(N)\cup
\{p_i\}\cup D_{1i},\label{c1i}\end{equation}where $D_{1i}$ is the set of ``blocking" players in $M$ whose first dimensions are too big that $C_{1i}$ cannot afford to exclude, i.e.\begin{equation}D_{1i}=\{p_j\in M: p_j^1+q_2(N)\geq p_i^2+q_1(N)\}.\label{d}\end{equation}\end{definition}

 Notice it may be true that $p_i\in D_{1i}$. The lemma below show that whether $C_{1i}\in {\mathcal M}{\mathcal W}{\mathcal C}1_i$ can be checked very easily.
\begin{lemma}In the two dimensional C-WMMG, $\forall p_i\in M$, $C_{1i}\in {\mathcal M}{\mathcal W}{\mathcal C}1_i$ iff $C_{1i}$ is a winning coalition and $p_i$ is both decisive and busy in $C_{1i}$.\label{l3}\end{lemma}
\begin{proof}The necessary part is obvious, so we only need to show the sufficient part. As $C_{1i}$ is winning and all its members  are decisive, we have $C_{1i}\in {\mathcal M}{\mathcal W}{\mathcal C}$.
Together with $p_i$ is busy in $C_{1i}$, we complete the proof.\qed
\end{proof}

 $\forall p_i\in M$, we can similarly define ${\mathcal M}{\mathcal W}{\mathcal C}2_{i}=\{C\in {\mathcal M}{\mathcal W}{\mathcal C}:
A_2(C)=A_2(N), p_i\in A_1(C)\}$, $D_{2i}=\{p_j\in M: p_j^2+q_1(N)\geq
p_i^1+q_2(N)\}$,  $C_{2i}=A_2(N)\cup \{p_i\}\cup D_{2i}$, and have parallel results.

\subsection{The upper-bound}
Now we are ready to prove the main result of this section.
\begin{theorem} In the two dimensional C-WMMG, $|\mathcal{MWC}|\leq n+1$, and this bound is tight.\label{t1}
\end{theorem}
\begin{proof}
In the case that $A_1(N)\cap A_2(N)\neq \emptyset$, Lemma \ref{l03} implies the theorem. So we only need to consider the case $A_1(N)\cap
A_2(N)=\emptyset$.

Let \begin{equation}T=\{p_i\in M: |{\mathcal M}{\mathcal W}{\mathcal C}1_{i}|=|{\mathcal
M}{\mathcal W}{\mathcal C}2_{i}|=1\}.\end{equation}

 For arbitrary $p_i, p_j\in T$, we will show that $p_i^1\neq p_j^1$ implies
$p_i^2=p_j^2$. W.l.o.g. we assume first that \begin{equation}p_i^1>p_j^1.\label{asu1}\end{equation}

Since $p_j\in T$, we know by definition that ${\mathcal
M}{\mathcal W}{\mathcal C}2_{j}\neq \emptyset$, which further tells us that its only possible member $C_{2j}$, as similarly defined in (\ref{c1i}), is an MWC. Therefore, if $q_1(N)+p_i^2\geq q_2(N)+p_j^1$, then we would have $p_i\in D_{2j}$, and consequently the fact that $p_j$ is a busy-1 player of $C_{2j}$ would imply that $p_i^1\leq p_j^1$, contradicting assumption (\ref{asu1}). Hence we can only have
\begin{equation}q_1(N)+p_i^2<q_2(N)+p_j^1.\label{a12}\end{equation}

By a similar argument, $p_i\in T$ says that $C_{1i}$ is an MWC, containing $p_i$ as a busy-2 player.  Inequality (\ref{a12}) tells us that
$p_j\in D_{1i}\subset C_{1i}$, and therefore \begin{equation}p_j^2\leq p_i^2.\label{a11}\end{equation}

 Combining inequalities (\ref{asu1})  (\ref{a12}) (\ref{a11}) we have
\begin{equation}q_1(N)+p_j^2<q_2(N)+p_i^1.\label{a3}
\end{equation}

Again, $p_j\in T$ also implies that $C_{1j}$ is an MWC containing $p_j$ as a busy-2 player. Inequality (\ref{a3}) tells us
that $p_i\in D_{1j}\subset C_{1j}$, and hence $p_i^2\leq p_j^2$, which,  together with inequality (\ref{a11}), gives $p_i^2=p_j^2$.

By symmetry we know that $p_i^2\neq p_j^2$ implies $p_i^1=p_j^1$. It is not hard to check that this can only happen when all the members in $T$ have one identical dimension. W.l.o.g. we assume that their second dimensions are the
same.

For arbitrary $p_i, p_j\in T$, suppose that ${\mathcal M}{\mathcal W}{\mathcal C}1_{i}\neq
{\mathcal M}{\mathcal W}{\mathcal C}1_{j}$, next we will show that ${\mathcal M}{\mathcal
W}{\mathcal C}2_{i}= {\mathcal M}{\mathcal W}{\mathcal C}2_{j}$. Remember that \begin{equation}p_i^2=p_j^2.\label{a16}\end{equation}

First of all, since ${\mathcal M}{\mathcal W}{\mathcal C}1_{i}=\{C_{1i}\}, {\mathcal M}{\mathcal W}{\mathcal C}1_{j}=\{C_{1j}\}$,
${\mathcal M}{\mathcal W}{\mathcal C}1_{i}\neq {\mathcal M}{\mathcal W}{\mathcal C}1_{j}$ means that \begin{equation}C_{1i}\neq C_{1j}.\label{a13}\end{equation}

By definition (\ref{d}), $p_i^2=p_j^2$ implies that \begin{equation}D_{1i}=D_{1j}.\label{a14}\end{equation}

We claim that $p_j\notin D_{1i}$, i.e.
\begin{equation} p_j^1+q_2(N)<p_i^2+q_1(N).\label{a15}\end{equation} In fact,  otherwise we would have $C_{1j}\subseteq C_{1i}$, which would further imply $C_{1j}=C_{1i}$, because they are both MWCs. This contradicts (\ref{a13}).

From inequality (\ref{a15}), we immediate have $p_i\in D_{2j}\subset C_{2j}$ and further \begin{equation}p_i^1\leq p_j^1,\label{a17}\end{equation} because $p_i\in T$. Combining inequalities (\ref{a16})(\ref{a15})(\ref{a17}), we have
\begin{equation} p_i^1+q_2(N)<p_j^2+q_1(N).\label{a4}\end{equation}

 Again, inequality (\ref{a4}) tells us that \begin{equation}p_j\in D_{2i}\subset C_{2i},\label{a18}\end{equation} and hence
$p_j^1\leq p_i^1$, which,  together with (\ref{a17}), gives \begin{equation}p_i^1=p_j^1.\label{a19}\end{equation}

By definition, (\ref{a19}) implies that $D_{2i}=D_{2j}$, which, together with (\ref{a18}), further implies that $C_{2j}\subseteq C_{2i}$. Since $C_{2i}$ and $C_{2j}$ are both MWCs, we get eventually  $C_{2i}=C_{2j}$.

According to the above discussion, either ${\mathcal M}{\mathcal W}{\mathcal C}_{1i}={\mathcal M}{\mathcal W}{\mathcal
C}_{1j}$ holds for all $p_i, p_j\in T$, or ${\mathcal M}{\mathcal W}{\mathcal C}_{2i}={\mathcal
M}{\mathcal W}{\mathcal C}_{2j}$ holds for all $p_i, p_j\in T$. Therefore,
\begin{equation}\left|\bigcup _{p_j\in T}({\mathcal M}{\mathcal W}{\mathcal C}1_{j}\cup {\mathcal M}{\mathcal W}{\mathcal C}2_{j})\right|\leq |T|+1.\label{a20}\end{equation}

For the case $A_1(N)$ is a winning coalition and the case $A_2(N)$ is a  winning coalition, the theorem is trivially true (Lemma \ref{l04}). So we
assume that $q(A_1(N))<q(N\setminus A_1(N))$ and $q(A_2(N))<q(N\setminus A_2(N))$. Therefore,
${\mathcal M}{\mathcal W}{\mathcal C}1=\bigcup _{p_j\in M}{\mathcal M}{\mathcal W}{\mathcal
C}1_{j}$ and ${\mathcal M}{\mathcal W}{\mathcal C}2=\bigcup _{p_j\in M}{\mathcal M}{\mathcal
W}{\mathcal C}2_{j}$. Finally we have:
\begin{eqnarray*}
 &&|{\mathcal M}{\mathcal W}{\mathcal C}|\\
&=&|{\mathcal M}{\mathcal W}{\mathcal C}1\cup {\mathcal M}{\mathcal W}{\mathcal C}2\cup {\mathcal M}{\mathcal W}{\mathcal C}3|\\
 &=&\left| \bigcup _{p_j\in M}({\mathcal M}{\mathcal W}{\mathcal C}1_{j}\cup {\mathcal M}{\mathcal W}{\mathcal C}2_{j})\right|+|{\mathcal M}{\mathcal W}{\mathcal C}3|\\
 &\leq&\left|\bigcup _{p_j\in M\setminus T}({\mathcal M}{\mathcal W}{\mathcal C}1_{j}\cup {\mathcal M}{\mathcal W}{\mathcal C}2_{j})\right|+\left|\bigcup _{p_j\in T}({\mathcal M}{\mathcal W}{\mathcal C}1_{j}\cup {\mathcal M}{\mathcal W}{\mathcal C}2_{j})\right|+m_1+m_2\\
 &\leq& (|M|-|T|)+(|T|+1)+m_1+m_2\\
 &=&n+1,
\end{eqnarray*}
where the second last inequality is from Lemma \ref{l06}, and the last inequality is from (\ref{a20}).

The following example shows that the upper bound $n+1$ is tight.

There are 4 kinds of players: (1) 4 {\it huge} players: $p_1=p_2=(n^2, 0), p_3=p_4=(0, n^2)$; (2)
$t-1$ {\it left} players: $x_t=(t, 0), x_{t-1}=(t-1, 0), \cdots , x_2=(2, 0)$; (3) $t-1$ {\it
right} players: $y_t=(0, t), y_{t-1}=(0, t-1), \cdots, y_2=(0, 2)$; (4) 1 {\it versatile} player
$z=(1, 1)$.

It can be easily checked that $A_1(N)=\{p_1,p_2\}$, $A_2(N)=\{p_3, p_4\}$; $|{\mathcal M}{\mathcal W}{\mathcal
C}3|=2+2=4$; For each of the left player $x_j, 2\leq j\leq t$, ${\mathcal M}{\mathcal W}{\mathcal
C}1_{x_j}=\emptyset$ and $C_{2x_j}=\{p_3, p_4, x_j, y_t, y_{t-1},\cdots,  y_j\}\in {\mathcal
M}{\mathcal W}{\mathcal C}2$; For each of the
 right players $y_j, 2\leq j\leq t$, ${\mathcal M}{\mathcal W}{\mathcal C}2_{y_j}=\emptyset$ and
 $C_{1y_j}=\{p_1, p_2, y_j, x_t, x_{t-1}, \cdots, x_j\}\in {\mathcal M}{\mathcal W}{\mathcal C}1$;
 For the versatile player $z$, $C_{1z}=\{z, p_1, p_2, x_t, x_{t-1},$ $\cdots, x_2\}\in {\mathcal M}{\mathcal W}{\mathcal C}1$ and
 $C_{2z}=\{z, p_3, p_4, y_t, y_{t-1}, \cdots, y_{2}\}$ $\in {\mathcal M}{\mathcal W}{\mathcal C}2.$

 Therefore, there are $n=2t+3$ players and $|{\mathcal M}{\mathcal W}{\mathcal C}|=4+2(t-1)+2=n+1$.\qed\end{proof}

 \section{Local Monotonicity\label{s5}}

In this section, we shall prove that local monotonicity, which says that players with larger weights
have no less power than the ones with smaller weights, holds for all the four power indices in the two
dimensional C-WMMG.

First of all, it is trivial that this property holds for the Penrose-Banzhaf index and the Shapley-Shubik index
(for arbitrary dimension, in fact). To verify that it also holds for the Holler-Packel index and
the Deegan-Packel index, we need two lemmas.
\begin{lemma}In the two dimensional C-WMMG, for all $C\in {\mathcal M}{\mathcal W}{\mathcal C}$, the following properties hold:\\
(a) $\forall p_j\in C$, either $p_j$ is busy in $C$ or $p_j$ is busy in $C^-\cup \{p_j\}$. That is, $p_j\in B(C)\cup B(C^-\cup \{p_j\})$;\\
(b) If $C$ has a benchwarmer, then either $q_1(C)<q_1(C^-)$ or
$q_2(C)<q_2(C^-)$,  and they can't hold simultaneously;\\
(c) If $|A_1(C)|\geq 2$ and $A_1(C)\nsubseteq A_{2}(C)$, then either $q_{1}(C)\leq q_{1}(C^-)$ or
$q_{2}(C)\leq q_{2}(C^-)$, and they can't hold simultaneously.\label{l6}
\end{lemma}
\begin{proof}
(a) If $p_j$ is a busy player of $C$, the proof is finished. Otherwise, $C\in {\mathcal M}{\mathcal
W}{\mathcal C}$ implies $q(C)>q(C^-)$ and $q(C)=q(C\setminus \{p_j\})\leq q(C^-\cup \{p_j\})$.
Therefore, $q(C^-\cup \{p_j\})>q(C^-)$, which means $p_j\in B(C^-\cup \{p_j\})$.

(b) Suppose that $p_j$ is a benchwarmer of $C$, then $q(C)\leq q(C^-\cup \{p_j\})$ and $p_j\in
B(C^-\cup \{p_j\})$. Obviously, $(p_j,p_j)\notin A(C^-\cup \{p_j\})$. Thus, either $q(C)\leq
q_{1}(C^-)+p_j^2$ or $q(C)\leq q_{2}(C^-)+p_j^1$. Since $q(C)=q_1(C)+q_2(C)$ and
$q_{1}(C)>p_j^1$, $q_{2}(C)>p_j^2$, we have either $q_1(C)<q_1(C^-)$ or
$q_2(C)<q_2(C^-)$. It is trivial that they can't hold simultaneously, since $C$ is a winning
coalition.

(c) Since $|A_1(C)|\geq 2$ and $A_1(C)\nsubseteq A_{2}(C)$, we can
 take some $p_j\in A_1(C)\setminus A_{2}(C)$ such that $q(C\setminus \{p_j\})=q(C)$.
As $C\in {\mathcal M}{\mathcal W}{\mathcal C}$, we have $q(C)\leq q(C^-\cup \{p_j\})$ and
$p_j\in B(C^-\cup \{p_j\})$. Moreover, $p_j\in A_1(C)\setminus A_{2}(C)$ tells us that
$(p_j,p_j)\notin A(C^-\cup \{p_j\})$. Thus, either $q_1(C)+q_2(C)\leq q_{1}(N\setminus
C)+p_j^2$ or $q_1(C)+q_2(C)\leq q_{2}(C^-)+p_j^1)$. $q_{1}(C)\geq p_j^1$ and
$q_{2}(C)\geq p_j^2$ give that either  $q_{1}(C)\leq q_{1}(C^-)$ or $q_{2}(C)\leq q_{2}(C^-)$.
They can't hold simultaneously, because $C$ is a winning coalition.\qed
\end{proof}
\begin{lemma}In the two dimensional C-WMMG, take $C_1, C_2\in {\mathcal M}{\mathcal W}{\mathcal C}$. Suppose $(p_{i_0},p_{j_0})\in A(C_1)\cap A(C_2)$, then $C_1=C_2$ if one of the following conditions holds:\\
(a) $|A(C_t)|=1$ for some $t\in \{1,2\}$, i.e. $A_1(C_t)=\{p_{i_0}\}$ and $A_2(C_t)=\{p_{j_0}\}$;\\
(b) $|A_t(C_1)|=|A_t(C_2)|=1$ for some $t\in \{1,2\}$.\label{l7}
\end{lemma}
\begin{proof}
(a) Without loss of generality, we assume that $t=1$. Suppose on the contrary that $C_1\neq C_2$, then $C_1\setminus
C_2\neq \emptyset$ (note that no set can be a proper subset of the other because they are both MWCs). Take some player $p_{k_0}\in C_1\setminus C_2$. Then $p_{k_0}$ is a benchwarmer of $C_1$ because $p_{i_0}$ and $p_{j_0}$ are the all busy players of $C_1$. From Lemma
\ref{l6}(b), we know that
\begin{equation}q_k(C_2)=q_k(C_1)<q_k(C^-_1) ~ for ~some ~k\in \{1, 2\}.\label{a2}\end{equation}

$C_1$ is an MWC means that, for all $p_j\in A_{k}(C^-_1)$: \begin{equation}q(\{p_j,
p_{k_0}\})=q_k(C^-_1)+p^{3-k}_{k_0}=q(C^-_1\cup \{p_{k_0}\})\geq q(C_1\setminus
\{p_{k_0}\})=q(C_2).\end{equation} As $p_{k_0}\notin C_2$ and $C_2$ is a winning coalition, we have
$p_j\in C_2$. Therefore $A_{k}(C_1^-)\subseteq C_2$, which further gives $q_{k}(C^-_1)\leq
q_k(C_2)$ and contradicts statement (\ref{a2}).

(b) Without loss of generality, we assume that $t=1$. Therefore \begin{equation}A_1(C_1)=A_1(C_2)=\{p_{i_0}\}.\label{f3}\end{equation}



Suppose on the contrary that $C_1\neq C_2$. Then $C_1\setminus C_{2}\neq \emptyset$. Take any player \begin{equation}p_{k_0}\in
C_1\setminus C_{2}.\label{f6}\end{equation}

If $p_{k_0}$ is a benchwarmer of $C_1$, the proof is done by using the same argument as in (a). So
we assume that $p_{k_0}$ is a busy player of $C_1$.



$C_1\in {\mathcal M}{\mathcal
W}{\mathcal C}$ tells us that: \begin{equation}q(C^-_1\cup \{p_{k_0}\})\geq q(C_1\setminus
\{p_{k_0}\})=q(C_1)=q(C_2).\end{equation}

If $(p_{k_0}, p_{k_0})\in A(C^-_1\cup \{p_{k_0}\})$, we
would have $q(\{p_{k_0}\})\geq q(C_1)$, and further $p_{k_0}\in A_1(C_1)=\{p_{i_0}\}$, which contradicts
(\ref{f6}). Therefore, there must exist some \begin{equation}w\in C^-_1\label{f8}\end{equation} such that
\begin{equation}q(\{w,p_{k_0}\})=q(C_1^-\cup \{p_{k_0}\})\geq q(C_2)=q(C_1\cup C_2).\label{f7}\end{equation}

As $p_{k_0}\notin C_2$, from (\ref{f7}) we have \begin{equation}w\in C_2,\label{f2}\end{equation} from the fact that $C_2$
is a winning coalition.

 Due to (\ref{f6}) (\ref{f2}), we have $\{w,p_{k_0}\}\subseteq C_1\cup C_2$. (\ref{f7}) implies further that $q(\{w,p_{k_0}\})=q(C_1\cup C_2)$, hence either $(w,p_{k_0})\in A(C_1\cup C_2)$ or $(w,p_{k_0})\in A(C_1\cup C_2)$, and consequently either $w=p_{i_0}$ or $p_{k_0}=p_{i_0}$.
Clearly, $w\neq p_{i_0}$ (due to  (\ref{f3}) and (\ref{f8})), and $p_{k_0}\neq p_{i_0}$ (\ref{f3}) and (\ref{f6})). A contradiction.\qed\end{proof}

\begin{theorem}In the two dimensional C-WMMG, local monotonicity holds for the Holler-Packel index and the Deegan-Packel index. That is, for all  $p_{i_0}, p_{j_0}\in N$, such that $p_{i_0}\leq p_{j_0}$, i.e. $p_{i_0}^1\leq p_{j_0}^1$ and $p_{i_0}^2\leq p_{j_0}^2$,
we have $\theta_{i_0}\leq \theta_{j_0}$ for any $\theta\in\{hp, dp\}$. In particular, $p_{i_0}=
p_{j_0}$ implies $\theta_{i_0}=\theta_{j_0}$.\label{t3}
\end{theorem}
\begin{proof}We shall show that there exists at most one $C\in {\mathcal M}{\mathcal W}{\mathcal C}$ that
consists of $p_{i_0}$ but not $p_{j_0}$, and if there does exist
such $C$, we have  at least one $C^*\in {\mathcal M}{\mathcal W}{\mathcal C}$ that consists of
$p_{j_0}$ but not $p_{i_0}$.

Suppose that $C\in {\mathcal M}{\mathcal W}{\mathcal C}$, $p_{i_0}\in C$ and $p_{j_0}\notin C$. By
the definition of winning coalitions, it is obvious that $(C\cup \{p_{j_0}\})\setminus \{p_{i_0}\}$
is winning. By Lemma \ref{l1}, there exists at least one $C^*\in {\mathcal M}{\mathcal W}{\mathcal C}$
such that $C^*\subseteq (C\cup \{p_{j_0}\})$. It is straightforward that $p_{j_0}\in C^*$,
$p_{i_0}\notin C^*$ and $1/|C^*|\geq 1/|C|$. It suffices to show that there is at most one such $C$.

Suppose that there is still another $C^{'}\in {\mathcal M}{\mathcal W}{\mathcal C}$ such
$p_{i_0}\in C^{'}$ and $p_{j_0}\notin C^{'}$. $p_{i_0}\leq p_{j_0}$ implies that $p_{i_0}$ is a
busy player in both $C$ and $C^{'}$, that is $p_{i_0}\in B(C)\cap B(C^{'})$. Without loss of
generality, we assume that $p_{i_0}\in A_1(C)$. Since $p_{i_0}\leq p_{j_0}$ and $p_{j_0}\notin
 C$, we  have further that $A_1(C)=\{p_{i_0}\}$, because otherwise $p_{i_0}$ would be indecisive to $C$.

  We claim that $p_{i_0}\in A_1(C^{'})$ and
 thus $A_1(C^{'})=\{p_{i_0}\}$. Otherwise, we will have $p_{i_0}\in A_2(C^{'})$ and
 thus $A_2(C^{'})=\{p_{i_0}\}$. Since $p_{i_0}\leq p_{j_0}$ and $p_{j_0}\notin C$, it holds that $(p_{i_0}, p_{i_0})\notin A(C)$, there must exist
 some $p_{k_0}\in C$ such that $(p_{i_0}, p_{k_0})\in A(c)$. Therefore $q(C)=q(\{p_{i_0},
 p_{k_0}\})$. Similarly, there exists some $p_{k_0}^{'}$ such that $(p_{k_0}^{'}, p_{i_0})\in A(C^{'}))$ and $q(C^{'})=q(\{p_{i_0},
 p_{k_0}^{'}\})$. As $p_{k_0}^2>p_{i_0}^2$ and $p_{k_0}^{'1}>p_{i_0}^1$, we have
 $p_{k_0}^{'}\notin C$ and $p_{k_0}\notin C^{'}$. $C$ is a winning coalition
 tells us that \begin{equation}q(\{p_{i_0}, p_{k_0}\})=q(C)>q(C^-)\geq q(\{p_{j_0},
 p_{k_0}^{'}\})\geq q(\{p_{i_0}, p_{k_0}^{'}\}).\end{equation} The last inequality holds because $p_{i_0}\leq
 p_{j_0}$. As $C^{'}$ is also a winning coalition and $q(C^{'})=q(\{p_{i_0},
 p_{k_0}^{'}\})$, we have \begin{equation}q(\{p_{i_0}, p_{k_0}^{'})>q(C^{'-})\geq q(\{p_{j_0},
 p_{k_0}\}),\end{equation} which gives $q(\{p_{i_0}, p_{k_0}\})>q(\{p_{j_0}, p_{k_0}\})$, a contradiction with
 $p_{i_0}\leq p_{j_0}$.

 As $C$ is a winning coalition, $p_{i_0}\in A_1(C)$ and $p_{i_0}\leq p_{j_0}$ imply
 that $A_2(C)=A_2(N)$, because otherwise $C$ would not be winning. Since $p_{i_0}\in A_1(C^{'})$, for the same reason we have
 $A_2(C^{'})=A_2(N)$. Therefore, $A_2(C)=A_2(C^{'})$. Together with
 $A_1(C)=A_1(C^{'})=\{p_{i_0}\}$, we finally get $C=C^{'}$ by Lemma \ref{l7}(b).

 As for the special case $p_{i_0}=p_{j_0}$, it is straightforward that
 $\theta_{i_0}=\theta_{j_0}$ since we have both $\theta_{i_0}\leq \theta_{j_0}$
 and $\theta_{j_0}\leq \theta_{i_0}$.\qed
\end{proof}

The following example shows that even if $p_{i_0}<p_{j_0}$, that is $p_{i_0}^1<p_{j_0}^1$ and
$p_{i_0}^2<p_{j_0}^2$, it is still possible that $\theta(p_{i_0})=\theta(p_{j_0})$ for all $\theta\in \{hp, dp,
pb, ss\}$.

\begin{example}There are four players in total: $p_1=(10, 2), p_2=(2, 10), p_3=(1, 3), p_4=(2,
4)$. $p_3$ has only one $MWC$: $\{p_1,p_3\}$,  $p_4$ has only one $MWC$: $\{p_1, p_4\}$. Therefore
$hp_3=hp_4=1$, and $dp_3=dp_4=1/2$;

$p_3$ has only one winning coalition in which she is decisive: $\{p_1,p_3\}$,  $p_4$ also has only one
winning coalition in which she is decisive: $\{p_1, p_4\}$. Therefore $pb_3=pb_4=1/8$ and
$ss_3=ss_4=1/12$.
\end{example}

The following example shows that a player who is busy in the grand coalition $N$, that is a player in
$B(N)$, may have smaller power than that of a player who is not, if powers are measured by Holler-Packel
index or Deegan-Packel index.

\begin{example} There are six players in total: $p_1=(10, 0), p_2=(0, 10), p_3=p_4=p_5=(0, 3), p_6=(3,
0)$. $p_2$ has 2 $MWC$s: $\{p_1, p_2\}$ and $\{p_2, p_3, p_4, p_5, p_6\}$; $p_6$ has 4 $MWC$s:
$\{p_1, p_3, p_6\}$, $\{p_1, p_4, p_6\}$, $\{p_1, p_5, p_6\}$ and $\{p_2, p_3, p_4, p_5, p_6\}$.
Therefore $hp_2=2<hp_6=4$, and $dp_2=7/10<dp_6=6/5$.\label{eg4}
\end{example}

 The example below shows that Theorem \ref{t3} doesn't
hold for the three or higher dimensional cases.

\begin{example}There are five players in total: $p_1=(5, 2, 1), p_2=(4, 0, 0), p_3=p_4=(0, 2, 0), p_5=(0, 0,
4)$. $p_1$ has 2 $MWC$s: $\{p_1, p_2\}$ and $\{p_1, p_5\}$; $p_2$ has 3 $MWC$s: $\{p_1, p_2\}$,
$\{p_2, p_3, p_5\}$ and $\{p_2, p_4, p_5\}$. Therefore $hp_1=2<hp_2=3$, and $dp_1=1<dp_2=7/6$,
while $p_1>p_2$.\label{eg5}
\end{example}

\section{Concluding Remarks\label{s7}}
C-WMMG is introduced in this paper. It models the complementary cooperation and can be taken as a sister model of the
famous weighted majority game.
It is well known that in continuous math, analysis that involves the operation of max is usually much more complicated than that of addition because it is not differentiable. In combinatorial optimization, the max form objective is also harder than that of sum form. In the field of coalitional game theory, however, we see in this paper the opposite relation between WMG and the two dimensional C-WMMG.

Since C-WMMG is a brand-new model, it is not surprising that, for further research, there are lots of promising open problems. We only list a few that are to the most interest of us.

(a) An obvious direction is to discuss the higher dimensional cases. It is
meaningful to analyze various other monotonicities. In particular, counter examples can easily be
constructed to show that new member monotonicity, which says that when a new member enters into the
game, normalized powers (original powers divided by the total power) of initial players will not
increase, is violated by all the four power indices. We conjecture that the paradox of redistribution, which says that
a player's (normalized) power may decrease when its weight increases, will not occur in any of the
four power indices. Whether an  example similar to Example \ref{eg4}  exists for the Shapley-Shubik
index and Penrose-Banzhaf
index is also an open question.

(b) The threshold variant of C-WMMG (TC-WMMG for short), where there is a lower bound $L$ such that $v(C)=1$ if and only if $q(C)\geq L$, is a more natural analogue of WMMG and thus worth further study. In fact, it can be taken as a generalization of C-WMMG, though it is difficult to give a clear value of the threshold such that TC-WMMG with this threshold is exactly C-WMMG (recall the relation between WMG and simple WMG as discussed in the second paragraph of Subsection \ref{s1.4}). What's more, it is also a generalization of WMG. To be precise, the special case where the dimension number $d_0$ equals the number of players $n$ and players have distinct nonzero dimensions is exactly WMG. The simplicity of a $d_0$ dimensional TC-CWMMG is that each MWC has at most $d_0$ members (there are no benchwarmers in MWC). Hence when $d_0$ is a constant, the total number of MWCs can be bounded by $O(n^{d_0})$. When $d_0=2$, this bound is tight, being reached when half of the players are characterized by $(L/2,0)$ and half by $(0,L/2)$. As in WMG and C-WMMG, local monotonicity holds trivially for the Shapley-Shubik
index and Penrose-Banzhaf index. This is not true, however, for the Holler-Packel index and the Deegan-Packel index, because giant players who can win by themselves may have smaller Holler-Packel and Deegan-Packel indices than the {\it versatile} ones that have a great number of potential partners. For {\it regular} TC-CWMMG where each MWC consists of exactly $d_0$ players, we are happy to observe that  local monotonicity always holds for the Holler-Packel index and the Deegan-Packel index. Based on the upper-bound of ${\mathcal M}{\mathcal W}{\mathcal C}$, efficient algorithms for computing all the Holler-Packel and Deegan-Packel indices exist trivially. Computing  Shapley-Shubik
and Penrose-Banzhaf indices, we guess, can also be done efficiently, though careful analysis is still needed. Based on the above discussions, it seems that almost all of the positive results for the two dimensional C-WMMG, presented in this paper, are still valid for TC-WMMG with a constant dimension.

(c) Another interesting problem for TC-WMMG, which is not addressed for C-WMMG (but already for WMG, see \cite{cy11}), is to study the selfish behavior of players. To be precise, assume that the value of each winning coalition is divided among its members proportional to their contributions, study the price of anarchy under various equilibrium or stability concepts. This problem, very interestingly, can be taken as a more concrete model of the famous stable marriage problem (or more precisely, the stable roommates problem), which is still actively studied today, see \cite{gs62,ss71,ir85,ir02} for standard references, and \cite{bmm12,c12,f12,fim11,hw11,k11,kk12,p12,t12} for recent works. A prominent  feature of the new model is that it is {\it cardinal} rather than {\it ordinal}, that is, players may get different payoffs by cooperating with different partners. There is not only a preference list of each player for all her potential partners or partner sets, but also the exact payoffs. We hope that this new ingredient may deepen the previous results of the stable marriage problem as well as of the stable roommates problems,  and bring brand-new interesting properties and problems.

(d) The airport game is a classical model in cooperative game theory \cite{b00,lo73,l74}, and is still attracting some new attention \cite{ail09,hy12,hty12,hty12b}. It has been well studied in real cases of cost allocation of airport building. A prominent advantage of the airport game is that the formidable Shapley value can be efficiently computed, as in the two dimensional C-WMMG. In fact, it is the only nontrivial model, other than the two dimensional C-WMMG and TC-WMMG (as discussed above in (b)), that possesses this property, as far as known to the authors. Interestingly, if we simply define the cost function as $c(C)=q(C)=\sum_{i=1}^{d_0}\max\{p_j^i:p_j\in C\}$, we can find that the cost sharing cooperative game model $(c(\cdot), N)$ is exactly the high dimensional version of the airport game. We hope that this new model can find good applications just like the airport game.

(d) An interesting variant of C-WMMG is the model where in any coalitional structure there are more than one coalitions that can get a nonzero payoff, say besides the winner, there is also a runner-up (we can either assume that the two coalitions get identical payoffs, or assume that the winner gets more than the runner-up). Coalitional games with externalities, also known as the partition function form as compared with the characteristic form \cite{lt63}, are very meaningful models that are attracting more and more attention both from the field of game theory, and from the field of supply chain management. To represent a general partition function form coalitional game, it is even more difficult than that of a characteristic form one, where there are $2^n$ number of coalitional values (to be exact, the number of coalitional structures in a general partition function form coalitional game, for each of which a certain number of values should be assigned,  is called {\it  Stirling number of the second kind}. Needless to say, it is extremely huge). This is one of the main barriers for in-depth study of this model. And in contrast to the classical characteristic form, where there have been lots of very famous, and well studied and well applied,  concrete models, very few is known to the partition function form. The new model we just proposed, we hope, may be helpful at least in part in solving the above two problems.

(f) Cooperation rule that is not purely complementary as displayed in the example of producing two goods in Subsection \ref{s1.1}, i.e. each player is restricted to play at most one role, is also of great interest. In fact, there exists quite a few natural and interesting cooperation rules, see \cite{cy10} for a tentative study.

(g) Computing other cooperative solutions for C-WMMG, say the nucleolus, the kernel, the least core, and various bargaining sets, are also meaningful. We remind the reader that these problems have already been considered for WMG \cite{eggw08,ep08,gmps11}.

(h) The inverse power index problems for WMG,  i.e. given a certain power index and a distribution  on players,  compute a WMG instance that has a distribution of power indices that is as close as possible to the given distribution, began to attract researchers' attention recently \cite{a10,dds12,k11,kkz10}. This is, in some sense, the ultimate goal for the research of power index measuring. Finding that  players' real bargaining powers in voting may not be proportional to their direct voting weights and conceiving more scientific measuring ways is only a first step. The second, and  probably the final,  step is to design better voting mechanisms or find better vote allocating ways. This branch of research, doubtlessly, is likely to be even more controversial than power measuring, because it is extremely hard for people to agree on what fairness is. Regardlessly of this, we believe that the inverse problem is a promising new direction that is worth more serious attention.   Various inverse problems for C-WMMG, of course,  are also meaningful. Since the two dimensional C-WMMG is much more handleable than WMG, we expect that its inverse problems will also be easier. Due to this feature, we can construct a new voting mechanism where each vote is characterized by a vector and pooling more than one votes is not adding them together, but to take a maximum for each dimension. Although this new voting mechanism may seem rather weird at the first glance, and honestly, we have completely no idea whether taxpayers might accept it one day or another, we hope that it may be useful in some really special scenarios.

 (i) Although complementary cooperation is an extremely popular phenomenon, as argued in the introduction part of this paper, the first possible applications and case studies of C-WMMG are likely to come from team sports, because the cooperation of players in team sports is largely  pure complementary, and evaluation of players (and benchwarmers in particular) is also meaningful in both theory and practice. This topic falls into the fields of  contest theory, and more generally sports economics \cite{dbsb06,hh08,s03}. For instance, our theory, along with more future in-depth  studies, may help people rank the NBA players, evaluate the benchwarmers, and select the MVP more objectively. We know that NBA player trading  is a huge market, our theory may also be helpful in advising team managers on how to trade players more efficiently.

{\bf Acknowledgements.} The authors would like to thank Prof. Zhi Jin for hosting a wonderful
seminar, in which the authors got the initial ideas of this paper. They thank Dr. Jian Tang
for pointing out a critical mistake in a first version of this paper. Thanks also go to Prof.  Xiuli Chao, Xin Chen, Chengzhong Qin, for helpful discussions and suggestions. They are also grateful to an anonymous referee for really valuable and detailed helps on significantly improving the presentation of this paper (in particular, reference \cite{ab03} was informed by him/her).

\begin{appendix}
\section{Computing MWCs\label{s4}}
Due to the analysis in Section \ref{s3}, an $O(n^2)$ algorithm for computing all the MWCs can be simply designed by brute-force enumeration. We show in this section that the time complexity can be reduced to $O(n\log n)$.  We need to deal with the data structure more carefully. All the notations in Section \ref{s3} are still valid in this section.

As shown in Lemma \ref{l03} and Lemma \ref{l04},  the case $A_1(N)\cap A_2(N)\neq \emptyset$ and the cases $A_1(N)$ is winning or $A_2(N)$ is winning are very simple, so we concentrate on the case where  $A_1(N)\cap A_2(N)=\emptyset$ and neither $A_1(N)$ nor $A_2(N)$ is winning.

Remember that $M$ is the set of benchwarmers of the grand coalition $N$. Let $m$ be the cardinality of $M$, i.e. \begin{equation}m=|M|,\end{equation} and re-index all the players in $M$ as $p_1, p_2, \cdots, p_{m}$ in non-increasing order of their first dimensions, i.e.
\begin{equation}p_1^1\geq p_2^1\geq\cdots \geq p_{m}^1.\label{r1}\end{equation}

\begin{definition}For each player $p_i\in M$,  we give her a second index $l(i), 1\leq l(i)\leq m$. We also assume that the second indices are in non-increasing order of their second dimensions, i.e.
\begin{equation}p^2_{l^{-1}(1)}\geq p^2_{l^{-1}(2)}\geq \cdots\geq p^2_{l^{-1}(m)},\label{r2}\end{equation}
where $p_{l^{-1}(i)}$ denotes the player in $M$ whose second index is $i$.\end{definition}

Given $p_i\in M$, suppose $D_{1i}$ is non-empty, then players in $D_{1i}$ have consecutive indices. To
be exact, we need a new notation.

\begin{definition}$\forall p_i\in M$, we use $x(i)$ to denote the largest index of players in $D_{1i}$, i.e. \begin{equation}x(i)=\max\Big\{j: p_{j}\in D_{1i}\Big\}.\label{xi}\end{equation}\end{definition}

Then we have \begin{equation}D_{1i}=\Big\{p_{j}:1\leq j\leq x(i)\Big\}.\label{xd1i}\end{equation}

Therefore, $D_{1i}$ is nicely determined by $x(i)$. If $D_{1i}=\emptyset$, we simply let
$x(i)=0$.

Similarly, \begin{definition} $\forall p_i\in M$, we use let $y(i)$ to denote the largest second index of players in $D_{2i}$, i.e. \begin{equation}y(i)=\max\Big\{l(j): p_j\in D_{2i}\Big\}.\end{equation} \end{definition}

Then we have \begin{equation}D_{2i}=\Big\{p_{l^{-1}(j)}: 1\leq
j\leq y(i)\Big\}.\end{equation}

And if $D_{2i}=\emptyset$, we let $y(i)=0$.

By definitions of $D_{1i}, D_{2i}$, and the two indexings, the following inclusion relations are obvious: \begin{equation}D_{1l^{-1}(1)}\subseteq
D_{1l^{-1}(2)}\subseteq \cdots \subseteq D_{1l^{-1}(m)},\label{b1}\end{equation}
\begin{equation}D_{21}\subseteq D_{22}\subseteq \cdots \subseteq D_{2m}.\label{b2}\end{equation}

And equivalently:
\begin{equation}x(l^{-1}(1))\leq x(l^{-1}(2))\leq \cdots \leq x(l^{-1}(m)),\label{a6}\end{equation}
\begin{equation}y(1)\leq y(2)\leq \cdots\leq y(m).\label{a7}\end{equation}

Due to the above monotone relations, and the fact that they are all integers falling into the interval $[0,m]$, $x(i)$s and $y(i)$s can be computed easily in $O(m)$ time. We omit the details. And hence the $D_{1i}$s, $D_{2i}$s, $C_{1i}$s and $C_{2i}$s can be computed in $O(m)$ time.

The remaining problem is to show that checking which of such coalitions are  MWCs and which are not can be done in $O(m\log m)$. Notice that checking them one by one independently cannot be efficient enough.  The main idea is also to do some pre-treatments. To be precise, we need several more notations.

\begin{definition}$\forall p_i\in M$, we use $R_1(i)$ and $r_1(i)$ to denote the smallest second index and second smallest second index of the players $\{p_1,p_2,\cdots,p_i\}$, respectively, i.e. \begin{equation}R_1(i)=\min\Big\{l(j):1\leq j\leq i\Big\},\end{equation} \begin{equation}r_1(i)=\min\Big\{l(j):1\leq j\leq i,l(j)\neq R_1(i)\Big\}.\end{equation}\end{definition}

 \begin{definition}$\forall p_i\in M$, we use $\mu_{1i}$ to denote the indicator of whether $p_i\in D_{1i}$ and has the largest second dimension, i.e. \begin{equation}\mu_{1i}=\left\{\begin{array}{ll}1&if~l(i)=R_1(x(i))\\
0&otherwise\end{array}\right..\end{equation}\end{definition}

Then it can be checked that\begin{equation}q_2\Big(D_{1i}\setminus \{p_i\}\Big)=(1-\mu_{1i})p_{l^{-1}(R_1(x(i)))}^2+\mu_{1i}p_{l^{-1}(r_1(x(i)))}^2.\label{g1}\end{equation}

The advantage of the $R_1(i)$s and $r_1(i)$s is that they can be computed in $O(m)$ time, using simple algorithmic techniques.

Similarly, $\forall p_i\in M$, let $R_2(i)$ and $r_2(i)$ be the smallest index and second smallest index of the players $\Big\{p_{l^{-1}(1)},p_{l^{-1}(2)},\cdots,p_{l^{-1}(i)}\Big\}$, respectively, i.e. \begin{equation}R_2(i)=\min\Big\{j:1\leq l(j)\leq i\Big\},\end{equation} \begin{equation}r_2(i)=\min\Big\{j:1\leq l(j)\leq i,j\neq R_2(i)\Big\}.\end{equation}

$R_2(i)$s and $r_2(i)$s can also be computed in $O(m)$ time. $\forall p_i\in M$, let $\mu_{2i}$ be the indicator of whether $p_i\in D_{2i}$ and has the largest first dimension, i.e. \begin{equation}\mu_{2i}=\left\{\begin{array}{ll}1&if~i=R_2(y(i))\\
0&otherwise\end{array}\right.,\end{equation}
then it can be checked that\begin{equation}q_1\Big(D_{2i}\setminus \{p_i\}\Big)=(1-\mu_{2i})p_{R_2(y(i))}^1+\mu_{2i}p_{r_2(y(i))}^1.\end{equation}

Analogous to Lemma \ref{l3}, we have the following more concrete
result.

\begin{lemma}In the two dimensional C-WMMG, $\forall p_i\in M$, $C_{1i}\in $ ${\mathcal M}{\mathcal W}{\mathcal C}1_i$ iff the following three
conditions hold simultaneously: \begin{equation}p_i^2\geq \max\Big\{p_{l^{-1}(R_1(x(i)))}^2, q_2(A_1(N))\Big\},\label{a9}\end{equation}
\begin{equation}q_1(N)+p_i^2>q_2(N)+\max\Big\{q_1(A_2(N)), p^1_{x(i)+1}\Big\},\label{a8}\end{equation} \begin{equation}q_1(N)+\max\Big\{q_2\big(D_{1i}\setminus \{p_i\}\big), q_2(A_1(N))\Big\}\leq
q_2(N)+\max\Big\{q_1(A_2(N)), p^1_{x(i)+1}, p_i^1\Big\},\label{a10}\end{equation} where $p_{m+1}^1$ is
defined as 0.\label{l4}\end{lemma}
\begin{proof}Condition (\ref{a9}) means that $p_i$ is busy-2 in $C_{1i}$; Condition (\ref{a8}) guarantees  that
$C_{1i}$ is winning; Condition (\ref{a10}) says $p_i$ is decisive in $C_{1i}$.\qed
\end{proof}

The above lemma tells us that to determine whether $C_{1i}\in {\mathcal M}{\mathcal W}{\mathcal
C}1_i$ or not takes constant time (recall (\ref{g1})). Similar result holds for $C_{2i}$.

A valuable notice is that, for any $p_i,p_j\in M, p_i\neq p_j$,
even if $C_{1i}\in {\mathcal M}{\mathcal W}{\mathcal C}1_i$ and $C_{1j}\in {\mathcal M}{\mathcal
W}{\mathcal C}1_j$, it is still possible that $C_{1i}=C_{1j}$. If this situation
occurs, it must hold that $p^2_i=p^2_j.$ The following lemma shows that eliminating the redundant coalitions can be done efficiently.

\begin{lemma}In the two dimensional C-WMMG, suppose that $p_{l^{-1}(i)}^2=p_{l^{-1}(i+1)}^2=\cdots =p_{l^{-1}(j)}^2, 1\leq i<j\leq
m$, and $C_{1l^{-1}(t)}\in {\mathcal M}{\mathcal W}{\mathcal C}1_{l^{-1}(t)}$ for all $i\leq t\leq
j$. By definition (\ref{d}), it can be observed that all the $D_{1l^{-1}(t)}$s are the same. We denote them as $D$.

(a) If $p_{l^{-1}(i)}\in D$, then $C_{1l^{-1}(i)}=C_{1l^{-1}(i+1)}=\cdots =C_{1l^{-1}(j)}$.

(b) If $p_{l^{-1}(i)}\notin D$, then $C_{1l^{-1}(s)}\neq C_{1l^{-1}(t)}$ holds for all $i\leq s, t\leq j, s\neq t$.\label{l5}\end{lemma}
\begin{proof}

(a) $\forall i<s\leq j$, we know by hypothesis that $C_{1l^{-1}(i)}\subseteq C_{1l^{-1}(s)}$. Since the two coalitions are both MWCs, they can only be identical.

(b) Suppose on the contrary that
$C_{1l^{-1}(s)}=C_{1l^{-1}(t)}$ for some $s\neq t$. This implies that $p_{l^{-1}(t)}\in
D$. Using a similar argument as in (a),  we have $C_{1l^{-1}(i)}=C_{1l^{-1}(i+1)}=\cdots
=C_{1l^{-1}(j)}$, and hence $p_{l^{-1}(i)}\in D$, a contradiction with the hypothesis. \qed\end{proof}

According to the above discussions, ${\mathcal
M}{\mathcal W}{\mathcal C}1$ and ${\mathcal M}{\mathcal W}{\mathcal C}2$ can be computed in $O(m)$
time. Since computing ${\mathcal M}{\mathcal W}{\mathcal C}3$ is easy (Lemma \ref{l05}), the algorithm for computing all the
MWCs can be easily designed. The algorithm is described as follows.
\begin{center}
\line(1,0){330}\\

{\it Algorithm MWC-2}\end{center}
{\it step 1.} {\bf Input} $N$;

 ~~~~~~Compute $A_1(N)$, $A_2(N)$, as well as $q_1(N)$, $q_2(N)$, $q_1(A_2(N))$ and $q_2(A_1(N))$;

~~~~~~{\bf if} $A_1(N)\cap A_2(N)\neq \emptyset$

~~~~~~~~~~{\bf if} $A_1(N)\subseteq A_2(N)$

~~~~~~~~~~~~~~{\bf output} $A_1(N)$

~~~~~~~~~~{\bf elseif} $A_2(N)\subset A_1(N)$

~~~~~~~~~~~~~~{\bf output} $A_2(N)$

~~~~~~~~~~{\bf else}

~~~~~~~~~~{\bf output} $\Big\{A_1(N), A_2(N)\Big\}$ and {\bf stop};

~~~~~~{\bf endif} \\
{\it step 2.} Re-index all the players in $M$ and calculate all the $l(j)$s as in (\ref{r1}) and (\ref{r2});\\
{\it step 3.} Compute all the $x(i)$s and $y(i)$s for all $p_i\in M$;\\
{\it step 4.} Compute all the $R_1(i)$s, $r_1(i)$s,  $R_2(i)$s and $r_2(i)$s for all $p_i\in M$;\\
{\it step 5.} For each $p_i\in M$, check if $C_{1i}$ and $C_{2i}$ are MWCs;\\
{\it step 6.} Eliminate the redundant coalitions in ${\mathcal M}{\mathcal W}{\mathcal C}1$ and ${\mathcal M}{\mathcal W}{\mathcal C}2$;\\
{\it step 7.} {\bf Output} ${\mathcal M}{\mathcal W}{\mathcal C}$.\\
\line(1,0){330}\\
\begin{theorem}${\mathcal M}{\mathcal W}{\mathcal C}$ can be computed in $O(n\log n)$ time. \label{t2}\end{theorem}
\begin{proof}Step 1 can be done in $O(n)$ time. According to the sorting theory which can be found in any algorithm design book, step 2 can be done in $O(m\log m)$. According to our previous discussions, step 3 and step 4 can be done in $O(m)$ time. According to Lemma \ref{l4}, step 5 can be done in $O(m)$ time. Lemma \ref{l5} tells us that step 6 can be done in $O(m)$ time. Since ${\mathcal M}{\mathcal W}{\mathcal C}$ is the union of ${\mathcal M}{\mathcal W}{\mathcal C}1$, ${\mathcal M}{\mathcal W}{\mathcal C}2$ and ${\mathcal M}{\mathcal W}{\mathcal C}3$, according to (\ref{mwc3}) and the above discussion, step 7 can be done in $O(n)$ time. In sum, the time complexity of  Algorithm MWC-2 is $O(n\log n)$ time (note $m<n$).   \qed\end{proof}

\section{Computing Holler-Packel and Deegan-Packel indices\label{s6}}
As the Holler-Packel and Deegan-Packel indices are directly determined by ${\mathcal
M}{\mathcal W}{\mathcal C}$, computing them is routine. The main concern is still a careful treatment of the data structure, because computing the indices one by one independently is quite inefficient.

\subsection{Holler-Packel indices}

As in the preceding section, we discuss the easier cases first. Due to Lemma \ref{l03}, the following theorem is obvious.
\begin{theorem}In the two dimensional C-WMMG, suppose $A_1(N)\cap A_2(N)\neq\emptyset$.

(a) When $A_1(N)=A_2(N)$:

\begin{equation}hp_i=\left\{\begin{array}{ll}1&if~p_i\in A_1(N)\\
0&if~p_i\in M\end{array}\right. ;\end{equation}

(b1) When $A_1(N)\subset A_2(N)$:

\begin{equation}hp_i=\left\{\begin{array}{ll}1&if~p_i\in A_1(N)\\
0&if~p_i\in N\setminus A_1(N)\end{array}\right. ;\end{equation}

(b2) When $A_2(N)\subset A_1(N)$:

\begin{equation}hp_i=\left\{\begin{array}{ll}1&if~p_i\in A_2(N)\\
0&if~p_i\in N\setminus A_2(N)\end{array}\right. ;\end{equation}

(c) Otherwise:

\begin{equation}hp_i=\left\{\begin{array}{ll}2&if~p_i\in A_1(N)\cap A_2(N)\\
1&if~p_i\in B(N)\setminus \big(A_1(N)\cap A_2(N)\big)\\
0&if~ p_i\in M\end{array}\right. .\end{equation}\label{thp1}
\end{theorem}

Due to Lemma \ref{l04}, the following theorem is also easy.
\begin{theorem}In the two dimensional C-WMMG, suppose $A_1(N)\cap A_2(N)=\emptyset$.

(a1) When $A_1(N)$ is a winning coalition and $m_1=1$:
 \begin{equation}hp_i=\left\{\begin{array}{ll}1&if~p_i\in A_1(N)\\
0&if~p_i\in A_2(N)\cup M\end{array}\right. ;\end{equation}

 (b1) When $A_1(N)$ is a winning coalition and $m_1>1$:
 \begin{equation}hp_i=\left\{\begin{array}{ll}2&if~p_i\in A_1(N)\\
m_1&if~p_i\in A_2(N)\\
0&if~p_i\in M\end{array}\right. ;\end{equation}

(a2) When $A_2(N)$ is a winning coalition and $m_2=1$:
 \begin{equation}hp_i=\left\{\begin{array}{ll}1&if~p_i\in A_2(N)\\
0&if~p_i\in A_1(N)\cup M\end{array}\right.;\end{equation}

(b2) When $A_2(N)$ is a winning coalition and $m_2>1$:
 \begin{equation}hp_i=\left\{\begin{array}{ll}m_2&if~p_i\in A_1(N)\\
2&if~p_i\in A_2(N)\\
0&if~p_i\in M\end{array}\right..\end{equation}\label{thp2}\end{theorem}

 Now, let's discuss the most complicated case. We need two additional symbols.

\begin{definition}We use $n_1$ and $n_2$ to denote the number of coalitions in  ${\mathcal
M}{\mathcal W}{\mathcal C}1$ and ${\mathcal
M}{\mathcal W}{\mathcal C}2$, respectively, i.e. \begin{equation}|{\mathcal
M}{\mathcal W}{\mathcal C}1|=n_1,\end{equation}
\begin{equation} |{\mathcal
M}{\mathcal W}{\mathcal C}2|=n_2.\end{equation}\end{definition}

We suppose that
\begin{equation}{\mathcal
M}{\mathcal W}{\mathcal C}1=\Big\{C_{1l^{-1}(u_1)}, C_{1l^{-1}(u_2)}, \cdots, C_{1l^{-1}(u_{n_1})}\Big\},\label{mwc1}\end{equation}
\begin{equation}{\mathcal M}{\mathcal W}{\mathcal C}2=\Big\{C_{2v_1}, C_{2v_2},\cdots, C_{2v_{n_2}}\Big\},\label{mwc2}\end{equation}
where $p_{l^{-1}(u_s)}\in M, p_{v_t}\in M, \forall 1\leq s\leq n_1, 1\leq t\leq n_2$, and \begin{equation} u_1<u_2<\cdots<u_{n_1},\label{u}\end{equation} \begin{equation}v_1<v_2<\cdots<v_{n_2}.\label{v}\end{equation}

\begin{definition}$\forall p_i\in M$, we use $\tau_{1i}$ to denote the index of the first coalition in (\ref{mwc1}) that the corresponding $D$ contains $p_i$, i.e.
\begin{equation}\tau_{1i}=\min\Big\{j:1\leq j\leq n_1, i\leq x\big(l^{-1}(u_j)\big)
\Big\}.\label{tau1}\end{equation}\end{definition}

Similarly, we define \begin{equation}\tau_{2i}=\min\Big\{j:1\leq j\leq n_2, i\leq y(v_j)\Big\}.\end{equation}

\begin{definition}$\forall p_i\in M$, we use $\varrho_{1i}$ to denote the indicator of whether $p_i\notin D_{1i}$, i.e.  \begin{equation}\varrho_{1i}=\left\{\begin{array}{ll}1&if~i>x(i)\\
0&otherwise\end{array}\right..\end{equation}\end{definition}

\begin{definition}$\forall p_i\in M$, we use $\sigma_{1i}$ to denote the indicator of whether $i\in \Big\{l^{-1}(u_1), l^{-1}(u_2), \cdots,
l^{-1}(u_{n_1})\Big\} $ and $p_i\notin D_{1i}$, i.e. \begin{equation}\sigma_{1i}=\left\{\begin{array}{ll}1&if~i\in\Big\{l^{-1}(u_1), l^{-1}(u_2), \cdots,
l^{-1}(u_{n_1})\Big\}~\&~\varrho_{1i}=1\\
0&otherwise\end{array}\right..\label{sigma1}\end{equation}\end{definition}

Similarly, for all $p_i\in M$, we  define \begin{equation}\varrho_{2i}=\left\{\begin{array}{ll}1&if~l(i)>y(i)\\
0&otherwise\end{array}\right.,\end{equation}
and \begin{equation}\sigma_{2i}=\left\{\begin{array}{ll}1&if~i\in\{v_1,v_2 \cdots,
v_{n_2}\}~\&~\varrho_{2i}=1\\
0&otherwise\end{array}\right..\end{equation}

Recall  that $m_1=|A_1(N)|,m_2=|A_2(N)|$ as defined in Section \ref{s3}. The formulas for calculating Holler-Packel indices in the most complicated case are as follows.

 \begin{theorem}In the two dimensional C-WMMG, suppose $A_1(N)\cap A_2(N)=\emptyset$ and neither $A_1(N)$ nor $A_2(N)$ is winning.

 (a1) For $p_i\in A_1(N)$: \begin{equation}hp_i=\left\{\begin{array}{ll}n_1+1&if~m_2=1\\
 n_1+m_2&if~m_1=1\\
  n_1+m_2+1&if~m_2>1,m_1>1\end{array}\right.;\end{equation}

(a2) For $p_i\in A_2(N)$: \begin{equation}hp_i=\left\{\begin{array}{ll}n_2+1&if~m_1=1\\
 n_2+m_1&if~m_2=1\\
  n_2+m_1+1&if~m_1>1,m_2>1\end{array}\right.;\end{equation}

(b) For $p_i\in M$:
\begin{equation}hp_i=\sum_{k=1,2}\Big(n_k-\tau_{ki}+1+\sigma_{ki}\Big).\end{equation}
\label{thp3}
\end{theorem}
\begin{proof}$\forall p_i\in A_1(N)$. Obviously, all the $n_1$ MWCs in ${\mathcal
M}{\mathcal W}{\mathcal C}1$ contain $p_i$. (i) If $m_2=1$, then by Lemma \ref{l05}, $\{p_i\}\cup A_2(N)$ is the only MWC in ${\mathcal
M}{\mathcal W}{\mathcal C}3$ that contains $p_i$.  The hypothesis $A_1(N)\cap A_2(N)=\emptyset$ tells us that none of the coalitions in ${\mathcal
M}{\mathcal W}{\mathcal C}2$ contains $p_i$. Therefore in this case $hp_i=n_1+1$. (ii) If $m_1=1$, then by Lemma \ref{l05}, the set of
MWCs in ${\mathcal
M}{\mathcal W}{\mathcal C}3$ that contains $p_i$ is $\Big\{A_1(N)\cup \{p_j\}:p_j\in A_2(N)\Big\}$, whose cardinality is $m_2$. Hence in this case $hp_i=n_1+m_2$. (iii) If $m_1>1$ and $m_2>1$, then by Lemma \ref{l05}, the set of
MWCs in ${\mathcal
M}{\mathcal W}{\mathcal C}3$ that contains $p_i$ is $\Big\{A_1(N)\cup \{p_j\}:p_j\in A_2(N)\Big\}\cup \Big\{\{p_i\}\cup \{A_2(N)\}\Big\}$, whose cardinality is $m_2+1$. Hence in this case $hp_i=n_1+m_2+1$.
 Based on the above discussion, (a1) is valid. (a2) can be shown to be true in the same way. We are left to show (b).

 From (\ref{b1})(\ref{b2})(\ref{u})(\ref{v}), we observe that \begin{equation}C_{1l^{-1}(u_{s})}\setminus C_{1l^{-1}(u_{t})}=\Big\{p_{l^{-1}(u_{s})}\Big\},\forall 1\leq s<t\leq n_1.\label{o1}\end{equation}

 By (\ref{b1})(\ref{u}) and the definition (\ref{tau1}), we know that $n_1-\tau_{1i}+1$ is the number of coalitions in ${\mathcal
M}{\mathcal W}{\mathcal C}1$ whose corresponding $D$ contains $p_i$. Recall the definition (\ref{c1i}), we know by definition (\ref{sigma1}) and observation (\ref{o1}) that there is at most one coalition in ${\mathcal
M}{\mathcal W}{\mathcal C}1$ that contain $p_i$ as the busy-2 player. This coalition has an extra contribution to the power indices iff $i\in \Big\{l^{-1}(u_1), l^{-1}(u_2), \cdots,
l^{-1}(u_{n_1})\Big\} $ and $p_i\notin D_{1i}$. By definition, $\sigma_{1i}$ is the indicator. Therefore, $n_1-\tau_{1i}+1+\sigma_{1i}$ is the total number of coalitions in ${\mathcal
M}{\mathcal W}{\mathcal C}1$ that contain $p_i$. Parallely, $n_2-\tau_{2i}+1+\sigma_{2i}$ is the total number of coalitions in ${\mathcal
M}{\mathcal W}{\mathcal C}2$ that contain $p_i$. Since there is no coalition in ${\mathcal
M}{\mathcal W}{\mathcal C}3$ that contains $p_i$, (b) is true.
\qed\end{proof}

\subsection{Deegan-Packel indices}
Similar to Theorem \ref{thp1}, Theorem \ref{thp2} and Theorem \ref{thp3}, we have the following formulas for calculating Deegan-Packel indices, whose proofs are quite similar to those of Holler-Packel indices and thus are  omitted. We only need to notice that \begin{equation}|C_{1i}|=m_1+x(i)+\varrho_{1i},\end{equation}
and\begin{equation}|C_{2i}|=m_2+y(i)+\varrho_{2i}.\end{equation}

\begin{theorem}In the two dimensional C-WMMG, suppose $A_1(N)\cap A_2(N)\neq\emptyset$.

(a) When $A_1(N)=A_2(N)$:

\begin{equation}dp_i=\left\{\begin{array}{ll}\frac{1}{m_1}&if~p_i\in A_1(N)\\
0&if~p_i\in M\end{array}\right. ;\end{equation}

(b1) When $A_1(N)\subset A_2(N)$:

\begin{equation}dp_i=\left\{\begin{array}{ll}\frac{1}{m_1}&if~p_i\in A_1(N)\\
0&if~p_i\in N\setminus A_1(N)\end{array}\right. ;\end{equation}

(b2) When $A_2(N)\subset A_1(N)$:

\begin{equation}dp_i=\left\{\begin{array}{ll}\frac{1}{m_2}&if~p_i\in A_2(N)\\
0&if~p_i\in N\setminus A_2(N)\end{array}\right. ;\end{equation}

(c) Otherwise:

\begin{equation}hp_i=\left\{\begin{array}{ll}\frac{1}{m_1}+\frac{1}{m_2}&if~p_i\in A_1(N)\cap A_2(N)\\
\frac{1}{m_1}&if~p_i\in A_1(N)\setminus A_2(N)\\
\frac{1}{m_2}&if~p_i\in A_2(N)\setminus A_1(N)\\
0&if~ p_i\in M\end{array}\right. .\end{equation}\label{tdp1}
\end{theorem}

\begin{theorem}In the two dimensional C-WMMG, suppose $A_1(N)\cap A_2(N)=\emptyset$.

 (a1) When $A_1(N)$ is a winning coalition and $m_1=1$:
 \begin{equation}dp_i=\left\{\begin{array}{ll}1&if~p_i\in A_1(N)\\
0&if~p_i\in A_2(N)\cup M\end{array}\right. ;\end{equation}

(b1) When $A_1(N)$ is a winning coalition and $m_1>1$:
 \begin{equation}dp_i=\left\{\begin{array}{ll}\frac{1}{m_1}+\frac{1}{m_2+1}&if~p_i\in A_1(N)\\
\frac{m_1}{m_2+1}&if~p_i\in A_2(N)\\
0&if~p_i\in M\end{array}\right. ;\end{equation}

(a2) When $A_2(N)$ is a winning coalition and $m_2=1$:
 \begin{equation}dp_i=\left\{\begin{array}{ll}1&if~p_i\in A_2(N)\\
0&if~p_i\in A_1(N)\cup M\end{array}\right.;\end{equation}

(b2) When $A_2(N)$ is a winning coalition and $m_2>1$:
 \begin{equation}dp_i=\left\{\begin{array}{ll}\frac{m_2}{m_1+1}&if~p_i\in A_1(N)\\
\frac{1}{m_2}+\frac{1}{m_1+1}&if~p_i\in A_2(N)\\
0&if~p_i\in M\end{array}\right..\end{equation}\label{tdp2}\end{theorem}

\begin{theorem}In the two dimensional C-WMMG, suppose $A_1(N)\cap A_2(N)=\emptyset$ and neither $A_1(N)$ nor $A_2(N)$ is winning.

(a1) For $p_i\in A_1(N)$: \begin{equation}dp_i=\left\{\begin{array}{ll}\frac{1}{2}
+\sum_{s=1}^{n_1}\frac{1}{m_1+x\big(l^{-1}(u_s)\big)+\varrho_{1l^{-1}(u_s)}}&if~m_2=1\\
 \frac{m_2}{m_1+1}
+\sum_{s=1}^{n_1}\frac{1}{m_1+x\big(l^{-1}(u_s)\big)+\varrho_{1l^{-1}(u_s)}}&if~m_1=1\\
  \frac{m_2}{m_1+1}+\frac{1}{m_2+1}
+\sum_{s=1}^{n_1}\frac{1}{m_1+x\big(l^{-1}(u_s)\big)+\varrho_{1l^{-1}(u_s)}}&if~m_2>1,m_1>1\end{array}\right.;\end{equation}

 (a2) For $p_i\in A_2(N)$:
\begin{equation}dp_i=\left\{\begin{array}{ll}\frac{1}{2}+\sum_{s=1}^{n_2}\frac{1}{m_2+y(v_s)+\varrho_{2v_s}}&if~m_1=1\\
 \frac{m_1}{m_2+1}+\sum_{s=1}^{n_2}\frac{1}{m_2+y(v_s)+\varrho_{2v_s}}&if~m_2=1\\
 \frac{m_1}{m_2+1}+\frac{1}{m_1+1}+\sum_{s=1}^{n_2}\frac{1}{m_2+y(v_s)+\varrho_{2v_s}}&if~m_1>1,m_2>1\end{array}\right.;\end{equation}

(b) For $p_i\in M$:
\begin{eqnarray}dp_i&=&\sum_{s=\tau_{1i}}^{n_1}\frac{1}{m_1+x\big(l^{-1}(u_s)\big)+\varrho_{1l^{-1}(u_s)}}+
\sum_{s=\tau_{2i}}^{n_2}\frac{1}{m_2+y(v_s)+\varrho_{2v_s}}\\
&&+\frac{\sigma_{1i}}{m_1+x(i)+1}+\frac{\sigma_{2i}}{m_2+y(i)+1}.\end{eqnarray}\label{th6}\end{theorem}

\subsection{Summary result}
\begin{theorem}After computing ${\mathcal M}{\mathcal W}{\mathcal C}$, computing  all the Holler-Packel
 indices can be done in $O(n)$ time, and computing all the Deegan-Packel indices can be done in $O(n^2)$ time.
\end{theorem}
\begin{proof}We only need to discuss the most complicated case. Due to (\ref{a6})(\ref{a7}) and (\ref{u})(\ref{v}), computing $\tau_{1i}$s and $\tau_{2i}$s can be done in $O(m)$
time. All the other necessary parameters can obviously be computed in $O(m)$ time. By Theorem \ref{thp3}, computing  all the Holler-Packel indices can be done in $O(n)$ time. By Theorem \ref{t1} and Theorem \ref{th6}, we know that each Deegan-Packel index can be computed in $O(n)$ time, and so computing all the Deegan-Packel indices can be done in $O(n^2)$ time.\qed\end{proof}

\section{Computing WCs\label{wc}}
The role of this section is three-folded. First, we want to show that, in the two dimensional C-WMMG,  although the number of winning coalitions may be huge, the structure is relatively simple, i.e. we can ``describe" ${\mathcal W}{\mathcal C}$ in polynomial time. This result obviously has its own interest. Second, we want to do some analysis about ${\mathcal W}{\mathcal C}(i)$s, which will be used in the next section.
Unlike Holler-Packel and Deegan-Packel indices, which are based on ${\mathcal M}{\mathcal
W}{\mathcal C}$, Penrose-Banzhaf and Shapley-Shubik indices rely on ${\mathcal W}{\mathcal C}(i)$s. To
compute the latter two power indices, we should first analyze the structure of ${\mathcal W}{\mathcal C}(i)$s. The analysis of  ${\mathcal W}{\mathcal C}(i)$s, however, is not complete, because we shall only consider the cases that are by-products of computing ${\mathcal W}{\mathcal C}$.
The relatively separable cases will be postponed to the next section. Third, the main ideas and notations for analyzing ${\mathcal W}{\mathcal C}$ will also be used in the next section.

\begin{definition} Similar to the way we
dealt with ${\mathcal M}{\mathcal W}{\mathcal C}$, we divide  ${\mathcal W}{\mathcal C}$ into
three sub-collections: \begin{equation}{\mathcal W}{\mathcal C}1=\Big\{C\in {\mathcal W}{\mathcal C}: A_1(N)\subseteq
C, A_2(N)\cap C=\emptyset\Big\},\end{equation}\begin{equation} {\mathcal W}{\mathcal C}2=\Big\{C\in {\mathcal W}{\mathcal C}:
A_2(N)\subseteq C, A_1(N)\cap C=\emptyset\Big\},\end{equation}\begin{equation} {\mathcal W}{\mathcal C}3=\Big\{C\in {\mathcal W}{\mathcal
C}: A_1(N)\cap C\neq \emptyset, A_2(N)\cap C\neq \emptyset\Big\}.\end{equation}\end{definition}

\begin{definition} $\forall p_i\in N$, we use ${\mathcal W}{\mathcal
C}3(i)$ to denote the set of winning coalitions in ${\mathcal W}{\mathcal
C}3$ that contain $p_i$ as a decisive player, i.e. \begin{equation}{\mathcal W}{\mathcal
C}3(i)=\Big\{C\in {\mathcal W}{\mathcal
C}3: C\ni p_i, C\setminus \{p_i\}\notin {\mathcal W}{\mathcal
C}\Big\}.\end{equation}\label{dwc3i}\end{definition}

\begin{lemma}In the two dimensional C-WMMG, $\forall p_i\in M$, we have ${\mathcal W}{\mathcal
C}3(i)=\emptyset$.\label{l012}\end{lemma}
\begin{proof}For any $C\in {\mathcal W}{\mathcal
C}3$ that contains $p_i\in M$, we have $q(C)=q(N)$ and the leaving of $p_i$ will not change $q(C)$ at all. At the same time, we also have $q(C^-\cup \{p_i\})<q(N)$. Therefore, $p_i$ is not decisive in any  $C\in {\mathcal W}{\mathcal
C}3$ and hence the lemma.\qed\end{proof}

\subsection{The simple cases}

\begin{lemma}In the two dimensional C-WMMG, if  $A_1(N)\cap A_2(N)\neq \emptyset$, then

(a) ${\mathcal W}{\mathcal
C}1={\mathcal W}{\mathcal
C}2=\emptyset$;

(b) $\forall p_i\in A_1(N)\cup A_2(N)$:{\small \begin{equation}{\mathcal W}{\mathcal
C}3(i)=\left\{\begin{array}{ll}\Big\{A_1(N)\cup C\cup D: C\subset A_2(N)\setminus A_1(N), D\subseteq M\Big\}&~\\
\bigcup \Big\{A_2(N)\cup C\cup D: C\subset A_1(N)\setminus A_2(N), D\subseteq M\Big\}&~\\
\bigcup \Big\{A_1(N)\cup A_2(N)\cup D: D\subseteq M\Big\}&if~p_i\in A_1(N)\cap A_2(N)\\
\Big\{A_1(N)\cup C\cup D: C\subset A_2(N)\setminus A_1(N), D\subseteq M\Big\}&if~p_i\in A_1(N)\setminus A_2(N)\\
\Big\{A_2(N)\cup C\cup D: C\subset A_1(N)\setminus A_2(N), D\subseteq M\Big\}&if~p_i\in A_2(N)\setminus A_1(N)\end{array}\right.,\label{wc3i}\end{equation}} and the three families in the first case are disjoint;

(c) \begin{eqnarray}{\mathcal W}{\mathcal
C}3&=&\Big\{A_1(N)\cup C\cup D: C\subset A_2(N)\setminus A_1(N), D\subseteq M\Big\}\\
&&\cup \Big\{A_2(N)\cup C\cup D: C\subset A_1(N)\setminus A_2(N), D\subseteq M\Big\}\\
&&\cup \Big\{A_1(N)\cup A_2(N)\cup D: D\subseteq M\Big\}.\end{eqnarray}\label{l013}\end{lemma}

\begin{proof}Since $A_1(N)\cap A_2(N)\neq \emptyset$, ${\mathcal W}{\mathcal
C}1={\mathcal W}{\mathcal
C}2=\emptyset$ is straightforward. (c) is an immediate result of (b) and Lemma \ref{l012}.
We are left to show that (b) is correct. First of all, coalitions that either include $A_1(N)$ or $A_2(N)$ are all in ${\mathcal W}{\mathcal C}3$.

In the case that $p_i\in A_1(N)\cap A_2(N)$, $p_i$ is contained in each of these coalitions and  is always decisive. The union of the three families in formula (\ref{wc3i}), which are evidently disjoint,   is exactly the whole set of these coalitions.

Let's consider now the case $p_i\in A_1(N)\setminus A_2(N)$. First of all, $A_1(N)\setminus A_2(N)\neq \emptyset$ implies that $A_1(N)\nsubseteq A_2(N)$. Suppose $C(i)$ is a coalition in ${\mathcal W}{\mathcal
C}3(i)$, then by Lemma \ref{l2}, we know that either $A_1(N)\subseteq C(i)$ or $A_2(N)\subseteq C(i)$. We claim that $A_2(N)\nsubseteq C(i)$. In fact, if $A_2(N)\subseteq C(i)$, we would have $q((N\setminus C(i))\cup \{p_i\})<q(N)=q(C(i)\setminus \{p_i\})$ and hence $p_i$ would be indecisive in $C(i)$, contradicting our hypothesis. Therefore we have  \begin{equation}A_1(N)\subseteq C(i).\label{g2}\end{equation} Further, we must have \begin{equation}\Big(N\setminus C(i)\Big)\cap A_2(N)\neq \emptyset,\label{g3}\end{equation}because otherwise $p_i$ would be indecisive in $C(i)$.
Combining (\ref{g2})(\ref{g3}),we can easily check that $\Big\{A_1(N)\cup C\cup D: C\subset A_2(N)\setminus A_1(N), D\subseteq M\Big\}$ is exactly ${\mathcal W}{\mathcal
C}3(i)$. So the second case is valid, and hence the third case by symmetry.\qed\end{proof}

Notice that when $A_1(N)\subseteq A_2(N)$, ${\mathcal W}{\mathcal
C}3(i)=\emptyset$ holds for all $p_i\in A_2(N)\setminus A_1(N)$. This is embodied is Lemma \ref{l013}. Similar result holds for $p_i\in A_1(N)\setminus A_2(N)$ when $A_2(N)\subseteq A_1(N)$.


\begin{lemma}In the two dimensional C-WMMG, suppose $A_1(N)\cap A_2(N)=\emptyset$.

(a1) When $A_1(N)$ is a winning coalition and $m_1=1$:
\begin{equation}{\mathcal W}{\mathcal
C}(i)=\left\{\begin{array}{ll}\Big\{A_1(N)\cup C: C\subseteq A_2(N)\cup M\Big\} &if~p_i\in A_1(N)\\
\emptyset &if~p_i\in A_2(N)\cup M\end{array}\right.,\label{l114}\end{equation}
\begin{equation}{\mathcal W}{\mathcal
C}=\Big\{A_1(N)\cup C: C\subseteq A_2(N)\cup M\Big\};\label{l115}\end{equation}

(b1) When $A_1(N)$ is a winning coalition and $m_1>1$:
\begin{equation}{\mathcal W}{\mathcal
C}(i)=\left\{\begin{array}{ll}\Big\{A_1(N)\cup C\cup D: C\subset A_2(N), D\subseteq M\Big\}&~\\
\bigcup \Big\{A_2(N)\cup \{p_i\}\cup D: D\subseteq M\Big\} &if~p_i\in A_1(N)\\
\Big\{A_2(N)\cup C\cup D: \emptyset \subset C\subset A_1(N), D\subseteq M\Big\} &if~p_i\in A_2(N)\\
\emptyset&if~p_i\in M\end{array}\right.,\label{g116}\end{equation}
\begin{eqnarray}{\mathcal W}{\mathcal C}&=&\Big\{A_1(N)\cup C\cup D: C\subset A_2(N), D\subseteq M\Big\}\\
&&\cup \Big\{A_2(N)\cup C\cup D: \emptyset \subset C\subset A_1(N), D\subseteq M\Big\};\end{eqnarray}

(a2) When $A_2(N)$ is a winning coalition and $m_2=1$:
\begin{equation}{\mathcal W}{\mathcal
C}(i)=\left\{\begin{array}{ll}\Big\{A_2(N)\cup C: C\subseteq A_1(N)\cup M\Big\} &if~p_i\in A_2(N)\\
\emptyset &if~p_i\in A_1(N)\cup M\end{array}\right.,\end{equation}

\begin{equation}{\mathcal W}{\mathcal
C}=\Big\{A_2(N)\cup C: C\subseteq A_1(N)\cup M\Big\};\end{equation}

(b2) When $A_2(N)$ is a winning coalition and $m_2>1$:

\begin{equation}{\mathcal W}{\mathcal
C}(i)=\left\{\begin{array}{ll}\Big\{A_1(N)\cup C\cup D: \emptyset \subset C\subset A_2(N), D\subseteq M\Big\} &if~p_i\in A_1(N)\\
\Big\{A_2(N)\cup C\cup D: C\subset A_1(N), D\subseteq M\Big\}&~\\
\bigcup \Big\{A_1(N)\cup \{p_i\}\cup D: D\subseteq M\Big\} &if~p_i\in A_2(N)\\
\emptyset & if~p_i\in M\end{array}\right.,\end{equation}

\begin{eqnarray}{\mathcal W}{\mathcal C}&=&\Big\{A_2(N)\cup C\cup D: C\subset A_1(N), D\subseteq M\Big\}\\
&&\cup \Big\{A_1(N)\cup C\cup D: \emptyset \subset C\subset A_2(N), D\subseteq M\Big\}.\end{eqnarray}\label{l014}\end{lemma}

\begin{proof}Since (a2) is symmetric to (a1), and (b2) is symmetric to (b1). It is sufficient to show that (a1) and (b1) are correct.

(a1) Since $A_1(N)$ is a winning coalition and has only one member, we know that any coalition that contains her is winning, and any coalition that fails to contain her is losing, and no other player in a winning coalition is decisive. So  (\ref{l114}) and (\ref{l115}) are correct;

(b1) $\forall X\in {\mathcal W}{\mathcal C}(i)$, due to Lemma \ref{l2}, either $A_1(N)\subseteq X$ or $A_2(N)\subseteq X$. By hypothesis, any coalition include $A_1(N)$ is winning.

(i) Suppose first $p_i\in A_1(N)$.  When $A_1(N)\subseteq X$, if we still have $A_2(N)\subseteq X$, then $p_i$ would not be decisive in $X$, because we would have \begin{equation}q(X\setminus \{p_i\})=q(N)>q\big((N\setminus X)\cup \{p_i\}\big),\end{equation} where the equality is true because there are at least two players in $A_1(N)$, and the inequality is true because $q_2\big((N\setminus X)\cup \{p_i\}\big)<q_2(N)$ (notice that the hypothesis $A_1(N)\cap A_2(N)$ is used). Therefore it must be true that $A_2(N)\nsubseteq X$. When $A_2(N)\subseteq X$, we must have $X\cap A_1(N)=\{p_i\}$, because if there is a second player $p_j\in X\cap A_1(N)$, we would have\begin{equation}q(X\setminus \{p_i\})=q(N)>q\big((N\setminus X)\cup \{p_i\}\big),\end{equation} where the equality is true because the existence of $p_j$, and the inequality is true because $q_1\big((N\setminus X)\cup \{p_i\}\big)<q_1(N)$ (notice again that the hypothesis $A_1(N)\cap A_2(N)$ is used). It can be checked trivially that all the coalitions in the formula are in ${\mathcal W}{\mathcal C}(i)$, and hence the formula in case is correct.

(ii) Suppose now $p_i\in A_2(N)$. First, it cannot be true that $A_1(N)\subseteq X$, because otherwise $p_i$ would not be decisive (notice that $A_1(N)$ is winning). Consequently, it must hold that $A_2(N)\subseteq X$. To be winning, $X$ must contain at least one member of $A_1(N)$. Due to the above analysis, it can be checked that the formula in this case is also true.

(iii) Suppose now $p_i\in M$. If $A_1(N)\subseteq X$, then $p_i$ would not be decisive, because $A_1(N)$ alone is winning. So we must have $A_2(N)\subseteq X$. To be winning, $X$ must also contain at least one member of $A_1(N)$. Consequently we have  \begin{equation}q(X\setminus \{p_i\})=q(N)>q\big((N\setminus X)\cup \{p_i\}\big),\end{equation} a contradiction with the assumption that $p_i$ is decisive in $X$.
It can only be that such $X$ does not exist at all, i.e. ${\mathcal W}{\mathcal C}(i)=\emptyset$.

Noticing that the set in the second line of  (\ref{g116}) is a subset of that in the third line, the formula for ${\mathcal W}{\mathcal C}$ can be easily checked to be true.\qed\end{proof}

\subsection{The complex case}
We are left to discuss the most complicated case, i.e. the case that $A_1(N)\cap A_2(N)=\emptyset$ and neither $A_1(N)$ nor $A_2(N)$ is winning.

\begin{definition} $\forall p_i\in M$, we denote ${\mathcal W}{\mathcal
C}1_{i}$ as the set of winning coalitions in ${\mathcal W}{\mathcal
C}1$ such that player $p_i$ is busy-2 (but not necessarily decisive), i.e. \begin{equation}{\mathcal W}{\mathcal
C}1_{i}=\Big\{C\in {\mathcal W}{\mathcal C}: A_1(C)=A_1(N), p_i\in A_2(C)\Big\}.\end{equation}\end{definition}

Obviously, when $A_1(N)\cap A_2(N)=\emptyset$ and $A_1(N)$ is not winning:\begin{equation}{\mathcal W}{\mathcal C}1=\bigcup_{p_i\in M} {\mathcal W}{\mathcal C}1_i.\label{wc1i}\end{equation}

\begin{definition}$\forall p_i\in M_1$, let ${\mathcal W}1_i$ be the subset of winning coalitions in ${\mathcal W}{\mathcal C}1_i$  that are minimal, i.e.\begin{equation}{\mathcal W}1_i=\Big\{C\in {\mathcal W}{\mathcal C}1_i:  C\setminus \{p_j\}\notin {\mathcal W}{\mathcal C}, \forall p_j\in C\setminus \{p_i\}\Big\}.\label{g117}\end{equation}\label{d14}\end{definition}

Recall the definitions of $C_{1i}$ and $D_{1i}$ in (\ref{c1i}) and (\ref{d}) of Section \ref{s3}.  Then using the same argument as in (\ref{c1i}) we know that ${\mathcal W}1_i$  has at most one member, in fact:
\begin{equation}{\mathcal W}1_i=\left\{\begin{array}{ll}\Big\{C_{1i}=A_1(N)\cup \{p_i\}\cup D_{1i}\Big\}&if~ p_i\in A_2(C_{1i})\\
\emptyset&otherwise\end{array}\right..\label{g118}\end{equation}

It is obvious that  \begin{equation}{\mathcal
M}{\mathcal W}{\mathcal C}1_i\subseteq {\mathcal W}1_i\subseteq {\mathcal W}{\mathcal C}1_i.\end{equation}

\begin{definition}For all $p_i\in M$, we use $E_{1i}$ to denote the set of ``optional" players in $M$ in the sense that (i) they are not indispensable, i.e they are not in $D_{1i}$, (ii) they cannot prevent $p_i$ being busy-2, i.e. their second dimensions are not larger than that of $p_i$. To be precise,
\begin{equation}E_{1i}=\Big\{p_j\in M\setminus \big(D_{1i}\cup \{p_i\}\big):p_j^2\leq p_i^2\Big\}.\label{e1i}\end{equation}\end{definition}

When ${\mathcal W}{\mathcal C}1_{i}$ is not empty, $C_{1i}$ and $E_{1i}$ determine the structure of ${\mathcal W}{\mathcal C}1_{i}$ completely. To be exact, we have the following result.
\begin{lemma}In the two dimensional C-WMMG, suppose $A_1(N)\cap A_2(N)=\emptyset$  and neither $A_1(N)$ nor $A_2(N)$ is a winning coalition.

(a) ${\mathcal W}{\mathcal C}1_{i}=\emptyset$ iff ${\mathcal W}1_{i}=\emptyset$;

(b) If  ${\mathcal W}1_{i}\neq \emptyset$, then ${\mathcal W}{\mathcal C}1_{i}=\Big\{C_{1i}\cup E:
E\subseteq E_{1i}\Big\}$.\label{l9}
\end{lemma}
\begin{proof}(a) According to definition, ${\mathcal W}1_{i}=\emptyset$ is equivalent to $p_i\notin A_2(C_{1i})$, which implies that ${\mathcal W}{\mathcal C}1_{i}=\emptyset$, because there is no way for $p_i$ to be busy-2. On the other hand, if we have ${\mathcal W}{\mathcal C}1_{i}=\emptyset$, then it must be true that  $p_i\notin A_2(C_{1i})$, because otherwise $C_{1i}$ would be a valid candidate.

(b) When ${\mathcal W}1_{i}\neq \emptyset$, then $\forall X\in {\mathcal W}1_{i}$, we have $C_{1i}\subseteq X$. By definition, $E_{1i}$ is the set of all optional players, and hence the formula. \qed\end{proof}

\begin{lemma}The family of $E_{1i}$s can be computed in $O(m)$ time.\label{l10}\end{lemma}
\begin{proof}For all $p_i\in M$, by (\ref{e1i}) and (\ref{xd1i}), we have \begin{equation}E_{1i}=\Big\{p_j\in M: j\neq i, j\geq x(i)+1, p_j^2\leq p_i^2\Big\}.\label{l121}\end{equation} Due to (\ref{r2}) and (\ref{a6}), it is easy to prove
that \begin{equation}E_{1l^{-1}(m)}\setminus \{p_{l^{-1}(m)}\}\subseteq E_{1l^{-1}(m-1)}\setminus \{p_{l^{-1}(m-1)}\}\subseteq\cdots \subseteq
E_{1l^{-1}(1)}\setminus \{p_{l^{-1}(1)}\}.\label{l122}\end{equation} Since checking whether $p_i\in E_{1i}$ takes constant time, computing all the $E_{1i}$s can be done in $O(m)$ time.\qed\end{proof}

The
following result is also easy but useful.

\begin{lemma}In the two dimensional C-WMMG, suppose $A_1(N)\cap A_2(N)=\emptyset$ and neither $A_1(N)$ nor $A_2(N)$ is a winning coalition. $\forall p_{i_1}, p_{i_2}\in M$.  If $p_{i_1}^2\neq p_{i_2}^2$, then ${\mathcal W}{\mathcal C}1_{i_1}\cap {\mathcal W}{\mathcal
C}1_{i_2}=\emptyset$.\label{l8}\end{lemma}
\begin{proof}The lemma is true because coalitions in ${\mathcal W}{\mathcal C}1_{i_1}$ have completely different busy-2 players with those in ${\mathcal W}{\mathcal C}1_{i_2}$.\qed\end{proof}

 However, when $p_{i_1}^2=p_{i_2}^2$, we cannot get ${\mathcal W}{\mathcal C}1_{i_1}={\mathcal W}{\mathcal
C}1_{i_2}$. This brings some more complication to computing ${\mathcal W}{\mathcal C}1$.

\begin{definition}We use $\Big\{M_{21}, M_{22},
\cdots, M_{2n^0_2}\Big\}$ to denote the partition of $M$, such that players in the same subset have the same
second dimensions and \begin{equation}q_2(M_{21})>q_2(M_{22})>\cdots>q_2(M_{2n_2^0}).\label{g121}\end{equation}\end{definition}



\begin{definition} $\forall 1\leq s\leq n_2^0$, for $p_{i_1},p_{i_2}\in M_{2s}$, it is straightforward that  $D_{1i_1}=D_{1i_2}$. We denote this  identical set  as  $D_1(M_{2s})$, i.e. \begin{equation}D_1(M_{2s})=\Big\{p_j\in M: p_j^1+q_2(N)\geq q_1(N)+q_2(M_{2s})\Big\}.\end{equation}\end{definition}

 Notice again it may be true that $D_1(M_{2s})\cap M_{2s}\neq\emptyset$.

 \begin{definition}$E_1(M_{2s})$, parallel to $E_{1i}$, is defined as the set of ``optional" players, i.e. \begin{equation}E_1(M_{2s})=\Big\{p_j\in M\setminus \big(D_{1}(M_{2s})\cup M_{2s}\big):p_j^2<q_2(M_{2s})\Big\}.\end{equation}\end{definition}

 \begin{definition}$\forall 1\leq s\leq n_2^0$, we define $\alpha_{1}^s$ as an indicator of whether $M_{2s}\cap D_1(M_{2s})=\emptyset$, i.e.
\begin{equation}\alpha_{1}^s=\left\{\begin{array}{ll}1&if~M_{2s}\cap D_1(M_{2s})=\emptyset\\
0&otherwise\end{array}\right..\end{equation}\end{definition}

Parallel to Lemma \ref{l8}, Lemma \ref{l9}(b), and Lemma \ref{l10}, we have

\begin{lemma}In the two dimensional C-WMMG, suppose $A_1(N)\cap A_2(N)=\emptyset$  and neither $A_1(N)$ nor $A_2(N)$ is a winning coalition.

(a) $\forall 1\leq s_1\neq s_2\leq n_2^0$, \begin{equation}\left(\bigcup_{p_i\in M_{2s_1}} {\mathcal W}{\mathcal C}1_i\right)\cap \left(\bigcup_{p_i\in M_{2s_2}} {\mathcal W}{\mathcal C}1_i\right)=\emptyset;\end{equation}

(b) $\forall 1\leq s\leq n_2^0$, $\bigcup_{p_i\in M_{2s}} {\mathcal W}{\mathcal C}1_i\neq \emptyset$ iff $q_2(M_{2s})\geq q_2\big(A_1(N)\cup D_1(M_{2s})\big)$;

(c) $\forall 1\leq s\leq n_2^0$, when $q_2\big(M_{2s})\geq q_2(A_1(N)\cup D_1(M_{2s})\big)$, $\bigcup_{p_i\in M_{2s}} {\mathcal W}{\mathcal C}1_i$ is equal to

\begin{equation} \left\{\begin{array}{ll}\Big\{{\small A_1(N)\cup D_1(M_{2s})\cup X\cup Y: \emptyset \subset X\subseteq M_{2s}, Y\subseteq E_1(M_{2s})}\Big\}& if M_{2s}\cap D_1(M_{2s})=\emptyset\\
{\small \Big\{A_1(N)\cup D_1(M_{2s})\cup X\cup Y: X\subseteq M_{2s}\setminus D_1(M_{2s}), Y\subseteq E_1(M_{2s})}\Big\}&otherwise\end{array}\right.;\end{equation}

(d) $\forall 1\leq s\leq n_2^0$, when $q_2\big(M_{2s})\geq q_2(A_1(N)\cup D_1(M_{2s})\big)$, we have \begin{equation}\left|\bigcup_{p_i\in M_{2s}} {\mathcal W}{\mathcal C}1_i\right|=2^{\big|M_{2s}\setminus D_1(M_{2s})\big|+\big|E_1(M_{2s})\big|}-\alpha_{1}^s\cdot 2^{\big|E_1(M_{2s})\big|};\label{140}\end{equation}

(e) The family of $E_{1}(M_{2s})$s can be computed in $O(m^2)$ time.\label{l18}\end{lemma}

\begin{proof}(a) $s_1\neq s_2$ means that $M_{2s_1}\cap M_{2s_2}=\emptyset$. While coalitions in $\bigcup_{p_i\in M_{2s_1}} {\mathcal W}{\mathcal C}1_i$ have at least one player in $M_{2s_1}$ as its busy-2 player, and $\bigcup_{p_i\in M_{2s_2}} {\mathcal W}{\mathcal C}1_i$ have at least one player in $M_{2s_2}$ as its busy-2 player, we know that they must be disjoint.

(b) Necessity is obvious, because otherwise no player in $M_{2s}$ would be busy. Sufficiency is also easy because $A_1(N)\cup D_1(M_{2s})\cup M_{2s}$ is clearly a member.

(c) Due to the discussion in (b), (c) is easy. We only need to notice that if $M_{2s}\cap D_1(M_{2s})=\emptyset$, we should be careful to include at least one player in $M_{2s}$.

(d) Owing to (c), \begin{equation}\left|\bigcup_{p_i\in M_{2s}} {\mathcal W}{\mathcal C}1_i\right|=\left\{\begin{array}{ll}2^{|E_1(M_{2s})|}\cdot(2^{|M_{2s}|}-1)&if~~M_{2s}\cap D_1(M_{2s})=\emptyset\\
2^{|M_{2s}\setminus D_1(M_{2s})|+|E_1(M_{2s})|}&otherwise\end{array}\right..\label{142}\end{equation}

Due to the definition of $\alpha_{1}^s$, (\ref{142}) is equivalent to the formula (\ref{140}).

(e) Note that we do not have any nice structure that is similar to (\ref{l122}). Simply by enumeration, however, $O(m^2)$ time is enough, because computing each $E_{1}(M_{2s})$ can be done in $O(m)$ time.\qed\end{proof}

We need a further notation that is parallel to (\ref{g118}).

\begin{definition}$\forall 1\leq s\leq n_2^0$, we denote\begin{equation}{\mathcal W}1(M_{2s})=\left\{\begin{array}{ll}\Big\{A_1(N)\cup D_{1}(M_{2s})\Big\}&if~ q_2(M_{2s})\geq q_2\big(A_1(N)\cup D_1(M_{2s})\big)\\
\emptyset&otherwise\end{array}\right..\end{equation}\end{definition}

Denote \begin{equation}{\mathcal W}1=\bigcup_{1\leq s\leq n_2^0} {\mathcal W}1(M_{2s}),\end{equation} then by Lemma \ref{l18} we know that it determines the structure of ${\mathcal W}{\mathcal C}1$.

\begin{lemma} ${\mathcal W}1$ can be computed in $O(m^2)$ time.\label{l19}\end{lemma}
\begin{proof}$\forall 1\leq s\leq n_2^0, \forall p_{i_1}, p_{i_2}\in M_{2s}$, we know that $x(i_1)=x(i_2)$. Denote it as $x(M_{2s})$. Similar to (\ref{xd1i}), $D_1(M_{2s})$ is nicely determined by $x(M_{2s})$. It takes $O(m)$ time, for all $1\leq s\leq n_2^0$, to determine whether $q_2(M_{2s})\geq q_2\big(A_1(N)\cup D_1(M_{2s})$. Combining Lemma \ref{l18}(d), we know that computing the set of ${\mathcal W}1$ can be done in $O(m^2)$ time.\qed\end{proof}

Similarly, we can define ${\mathcal W}{\mathcal C}2_i, {\mathcal W}2_i,
E_{2i}, M_{1i},  n_1^0, M_{1s},D_2(M_{1s}), E_2(M_{1s}), \alpha_2^s, {\mathcal W}2$, and parallel results hold.

When $A_1(N)\cap A_2(N)=\emptyset$, and neither $A_1(N)$ nor $A_2(N)$ is a winning coalition, similar to the structure of  ${\mathcal M}{\mathcal W}{\mathcal
C}3$ in Lemma \ref{l05}, the structure of ${\mathcal W}{\mathcal
C}3$ is also simple.

\begin{lemma}In the two dimensional C-WMMG, suppose $A_1(N)\cap A_2(N)=\emptyset$, and neither $A_1(N)$ nor $A_2(N)$ is a winning coalition, then
\begin{equation}{\mathcal W}{\mathcal
C}3=
\left\{\begin{array}{ll}\Big\{A_1(N)\cup C\cup D: \emptyset\subset C\subseteq A_2(N), D\subseteq M\Big\}& if~m_1=1\\
\Big\{A_2(N)\cup C\cup D: \emptyset \subset C\subseteq A_1(N), D\subseteq M\Big\}&if~m_2=1\\
\Big\{A_1(N)\cup C\cup D: \emptyset\subset C\subset A_2(N), D\subseteq M\Big\}&~\\
\bigcup \Big\{A_2(N)\cup C\cup D: \emptyset \subset C\subset A_1(N), D\subseteq M\Big\}&~\\
\bigcup\Big\{A_1(N)\cup A_2(N)\cup D: D\subseteq M\Big\}&if~m_1>1,m_2>1\end{array}\right..\end{equation}\label{l20a}\end{lemma}
\begin{proof}$\forall X\in{\mathcal W}{\mathcal C}3$. If $m_1=1$, we know by definition of ${\mathcal W}{\mathcal
C}3$ that $A_1(N)\subseteq X$ and $X$ contains at least one member of $A_2(N)$, hence the formula. The case that $m_2=1$ is symmetric to the case above. If $m_1>1$ and $m_2>1$, the formula is also true and the three sets are disjoint.\qed\end{proof}

\subsection{Summary result}

Now, we are ready to present the main result of this section.
\begin{theorem}In the two dimensional C-WMMG, the structure of ${\mathcal W}{\mathcal
C}$  can be computed in $O(m^2)$ time.\end{theorem}
\begin{proof}It is valuable to notice first that we don't need to ``enumerate" all the WCs, but ``describe" them, which can be done by enumerating  the {\it critical} ones which determine completely
the whole structure. We discuss in three cases.


(i) When $A_1(N)\cap A_2(N)\neq \emptyset$, by Lemma 13(a), ${\mathcal W}{\mathcal
C}1={\mathcal W}{\mathcal C}2=\emptyset$. From Lemma 13(c), we know that ${\mathcal W}{\mathcal C}3$ can be described in constant time;

(ii) When $A_1(N)\cap A_2(N)=\emptyset$ and either $A_1(N)$ or $A_2(N)$ is a winning coalition, the result hold due to Lemma \ref{l014};

(iii) When $A_1(N)\cap A_2(N)=\emptyset$ and neither $A_1(N)$ nor $A_2(N)$ is a winning coalition, we only need to show that each of ${\mathcal W}{\mathcal
C}1$, ${\mathcal W}{\mathcal
C}2$, and ${\mathcal W}{\mathcal
C}3$ can be computed in $O(m^2)$ time, because they are disjoint. Lemma \ref{l18} and Lemma \ref{l19} tell us that  ${\mathcal W}{\mathcal
C}1$ can be described in $O(m^2)$ time. And hence  ${\mathcal W}{\mathcal C}2$ by symmetry. Due to Lemma \ref{l20a},  ${\mathcal W}{\mathcal
C}3$ can be described in constant time. \qed\end{proof}

\subsection{Structure of ${\mathcal W}{\mathcal
C}3(i)$}

When $A_1(N)\cap A_2(N)=\emptyset$, and neither $A_1(N)$ nor $A_2(N)$ is a winning coalition, discussion about the structure of ${\mathcal W}{\mathcal
C}3(i)$  will be done in  this subsection, and that of all the other cases will be postponed to the next section.

Remember first by Lemma \ref{l012} that when $p_i\in M$, we have ${\mathcal W}{\mathcal C}3(i)=\emptyset$.

\begin{lemma}In the two dimensional C-WMMG, suppose $A_1(N)\cap A_2(N)=\emptyset$, and neither $A_1(N)$ nor $A_2(N)$ is a winning coalition.

(a) When $m_1>1$ and $m_2>1$, \begin{equation}{\mathcal W}{\mathcal
C}3(i)=\left\{\begin{array}{ll}\Big\{A_1(N)\cup C\cup D: \emptyset\subset C\subset A_2(N), D\subseteq M\Big\}&~\\
\bigcup \Big\{A_2(N)\cup \{p_i\}\cup D: D\subseteq M, p_i~is~decisive\Big\}&if~p_i\in A_1(N)\\
\Big\{A_2(N)\cup C\cup D: \emptyset\subset C\subset A_1(N), D\subseteq M\Big\}&~\\
\bigcup \Big\{A_1(N)\cup \{p_i\}\cup D: D\subseteq M, p_i~is~decisive\Big\}&if~p_i\in A_2(N)\end{array}\right.;\end{equation}

(b1) When $m_1=1$ and $m_2>1$,  \begin{equation}{\mathcal W}{\mathcal
C}3(i)=\left\{\begin{array}{ll}\Big\{A_1(N)\cup C\cup D: \emptyset\subset C\subset A_2(N), D\subseteq M\Big\}&~\\
\bigcup \Big\{A_2(N)\cup \{p_i\}\cup D: D\subseteq M,p_i~is~decisive\Big\}&if~p_i\in A_1(N)\\
\Big\{A_1(N)\cup \{p_i\}\cup D: D\subseteq M, p_i~is~decisive\Big\}&if~p_i\in A_2(N)\end{array}\right.;\end{equation}

(b2) When $m_1>1$ and $m_2=1$,  \begin{equation}{\mathcal W}{\mathcal
C}3(i)=\left\{\begin{array}{ll}\Big\{A_2(N)\cup \{p_i\}\cup D: D\subseteq M, p_i~is~decisive\Big\}&if~p_i\in A_1(N)\\
\Big\{A_2(N)\cup C\cup D: \emptyset\subset C\subset A_1(N), D\subseteq M\Big\}&~\\
\bigcup \Big\{A_1(N)\cup \{p_i\}\cup D: D\subseteq M,p_i~is~decisive\Big\}&if~p_i\in A_2(N)
\end{array}\right.;\end{equation}

(c) When $m_1=1$ and $m_2=1$, $\forall p_i\in A_1(N)\cup A_2(N)$: \begin{equation}{\mathcal W}{\mathcal
C}3(i)=\big
\{A_1(N)\cup A_2(N)\cup D: D\subseteq M, p_i~is~decisive\Big\}. \end{equation}\label{l20}\end{lemma}
\begin{proof}(a) We only need to consider the case $p_i\in A_1(N)$, because the other case is symmetric. $\forall X\in{\mathcal W}{\mathcal
C}3(i)$, due to Lemma \ref{l2}, either $A_1(N)\subseteq X$ or $A_2(N)\subseteq X$. (i) If $A_1(N)\subseteq X$, then $X$ must contain at least one member of $A_2(N)$ (the definition of ${\mathcal W}{\mathcal C}3(i)$), and it cannot be true that $A_2(N)\subseteq X$ also holds, i.e. $X\cap A_2(N)$ must be a nonempty and proper subset of $A_2(N)$(since $m_2>1$, this can be done), because otherwise $p_i$ would be indecisive (note the hypothesis that $m_1>1$). At the same time, all the players in $M$ are optional. It can be checked that all coalitions satisfying the above requirements are members of ${\mathcal W}{\mathcal C}3(i)$. (ii) If $A_2(N)\subseteq X$, then it must be true that $p_i\in X$, and $X$ can only contain this single member from $A_1(N)$, because otherwise $p_i$ would be indecisive. However, even so, we cannot guarantee that $p_i$ is decisive, and players in $M$ are not completely optional either. So we added a requirement that ``$p_i$ is decisive", and hence the formula.

(b1) $\forall X\in{\mathcal W}{\mathcal
C}3(i)$. The case that $p_i\in A_1(N)$ is identical to that in (a). When $p_i\in A_2(N)$, then it cannot be true that $A_2(N)\subseteq X$, because combining with $A_1(N)\cap X\neq \emptyset$ and $m_2>1$ would give that $p_i$ is indecisive. Therefore we have $A_1(N)\subseteq X$, and $p_i$ must be the only member from $A_2(N)$ contained by $X$. Still, we cannot guarantee that $p_i$ is always decisive. Hence the formula.

(b2) Symmetric to (b1).

(c) $\forall X\in{\mathcal W}{\mathcal
C}3(i)$. Since there is only one member in each of $A_1(N)$ and $A_2(N)$, the two players must both be contained by $X$. Still, $p_i$ cannot be guaranteed to be decisive. Hence the formula.\qed\end{proof}

The above lemma gives us a very rough whole picture of the structure of ${\mathcal W}{\mathcal C}3(i)$. It helps, but solves only a very little part of  the problem.

\begin{definition}$\forall p_i\in A_1(N)\cup A_2(N)$, we denote \begin{equation}{\mathcal W}{\mathcal
C}3^*(i)=\left\{\begin{array}{ll}\Big\{\{p_i\}\cup A_2(N)\cup D: D\subseteq M, p_i~is~decisive\Big\}&if~p_i\in A_1(N)\\
\Big\{\{p_i\}\cup A_1(N)\cup D: D\subseteq M, p_i~is~decisive\Big\}&if~p_i\in A_2(N)\end{array} \right..\end{equation}\end{definition}

To completely understand the structure of  ${\mathcal W}{\mathcal
C}3(i)$, we know from Lemma \ref{l20} that we need to find an efficient way to describe ${\mathcal W}{\mathcal
C}3^*(i)$.

\begin{definition}Suppose $p_i\in A_1(N)$. We denote $\Delta_1$ as the set of all possible busy-2 players when $p_i$ moves to the complementary coalition from a coalition in ${\mathcal W}{\mathcal
C}3^*(i)$, i.e. \begin{equation}\Delta_1=A_2(A_1(N))\cup \Big\{p_j\in M: p_j^2\geq q_2\big(A_1(N)\big)\Big\}.\end{equation}\end{definition}

\begin{definition}For all $p_i\in A_1(N)$ and $p_j\in \Delta_1$, we use ${\mathcal W}{\mathcal
C}3^{*}(1,i,j)$ to denote the set of coalitions in ${\mathcal W}{\mathcal
C}3^*(i)$ where $p_j$ is busy-2 if $p_i$ moves to the complementary coalition, i.e.
\begin{equation}{\mathcal W}{\mathcal
C}3^{*}(1,i,j)=\begin{array}{ll}\Big\{\{p_i\}\cup A_2(N)\cup D: D\subseteq M\setminus \{p_j\},&~\\
 ~~q_1(N)+p_j^2=q\big(A_1(N)\cup \{p_i\}\cup (M\setminus D)\big)\geq q\big(A_2(N)\cup D\big)\Big\}&~\end{array}.\end{equation}\end{definition}

\begin{definition}For all  $p_j\in \Delta_1$, we also let \begin{equation}M^+_{2j}=\Big\{p_k\in M: p_k^2>p_j^2\Big\},\end{equation} \begin{equation}{\Bar D}_{1j}=\Big\{p_k\in M: p_k^1+q_2(N)\leq q_1(N)+q^2_j\Big\}.\end{equation}\end{definition}

\begin{lemma}$\forall p_i\in A_1(N)$ and $p_j\in \Delta_1$,\begin{equation}{\mathcal W}{\mathcal
C}3^{*}(1,i,j)=\Big\{\{p_i\}\cup A_2(N)\cup D: M^+_{2j}\subseteq D\subseteq {\Bar D}_{1j}\Big\}.\end{equation}\label{l021}
 \end{lemma}
\begin{proof}$M^+_{2j}$ is the set of players in $M$ whose second dimensions are bigger than that of $p_j$. In order to guarantee that $p_j$ is busy-2 in the new coalition, none member of $M^+_{2j}$ should be in the complementary coalition, and therefore it must hold that $D\supseteq M^+_{2j}$.  ${\Bar D}_{1j}$ is the set of players in $M$ whose first dimensions are not so big that without $p_j$ the  coalition is still winning, and therefore  $D\subseteq {\Bar D}_{1j}$. Hence the lemma.\qed\end{proof}

We remark that Lemma \ref{l021} also implies that ${\mathcal W}{\mathcal
C}3^{*}(1,i,j)=\emptyset$ if $M^+_{2j}\nsubseteq {\Bar D}_{1j}$.

\begin{lemma}For all $p_i\in A_1(N)$ and $p_{j_1}, p_{j_2}\in \Delta_1$.

(a) If $p_{j_1}^2=p_{j_2}^2$, then ${\mathcal W}{\mathcal
C}3^{*}(1,i,j_1)={\mathcal W}{\mathcal
C}3^{*}(1,i,j_2)$;

(b) If $p_{j_1}^2\neq p_{j_2}^2$, then ${\mathcal W}{\mathcal
C}3^{*}(1,i,j_1)\cap {\mathcal W}{\mathcal
C}3^{*}(1,i,j_2)=\emptyset.$ \label{l022}\end{lemma}
\begin{proof}Immediate from Lemma \ref{l021}.\qed\end{proof}

\begin{definition}We use $\Delta_1^*$ to denote an arbitrary maximum subset of $\Delta_1$ such that no pair of players have the same second dimension.\end{definition}

\begin{lemma}In the two dimensional C-WMMG, suppose $A_1(N)\cap A_2(N)=\emptyset$, and neither $A_1(N)$ nor $A_2(N)$ is a winning coalition. Then for all $p_i\in A_1(N)$,  \begin{equation}{\mathcal W}{\mathcal
C}3^{*}(i)=\bigcup\limits_{p_j\in\Delta_1^*}{\mathcal W}{\mathcal
C}3^{*}(1,i,j)\end{equation}.\label{l023}\end{lemma}
\begin{proof}Immediate from Lemma \ref{l022}.\qed\end{proof}

We can similarly define $\Delta_2$ and  $\Delta_2^*$. For all $p_i\in A_2(N), p_j\in \Delta_2$, we define ${\mathcal W}{\mathcal
C}3^{*}(2,i,j)$, $M^+_{1j}, {\Bar D}_{2j}$, and results parallel to Lemma \ref{l021}, Lemma \ref{l022} and Lemma \ref{l023} all hold.


\begin{definition}For all $p_i\in A_1(N)$ and $p_j\in \Delta_1$, define\begin{equation}d_{1j}=\left\{\begin{array}{ll}1&if~M^+_{2j}\subseteq {\Bar D}_{1j}\\
0&otherwise\end{array}\right..\end{equation}\end{definition}

For all $p_i\in A_2(N)$ and $p_j\in \Delta_2$, $d_{2j}$ is similarly defined.

\begin{lemma}In the two dimensional C-WMMG, suppose $A_1(N)\cap A_2(N)=\emptyset$ and neither $A_1(N)$ nor $A_2(N)$ is winning. Then
\begin{equation}|{\mathcal W}{\mathcal C}3(i)|=\left\{\begin{array}{ll}2^m(2^{m_2}-2)+\sum_{p_j\in \Delta_1^*}d_{1j}\cdot 2^{|{\Bar D}_{1j}\setminus M^+_{2j}|}&if~p_i\in A_1(N)\\
2^m(2^{m_1}-2)+\sum_{p_j\in \Delta_2^*}d_{2j}\cdot 2^{|{\Bar D}_{2j}\setminus M^+_{1j}|}&if~p_i\in A_2(N)\end{array}\right..\end{equation}
\label{l14}\end{lemma}
\begin{proof}Suppose now $m_1>1$ and $m_2>1$. $\forall p_i\in A_1(N)$. According to Lemma \ref{l20}, ${\mathcal W}{\mathcal C}3(i)$ consists of two parts, and the first part has a cardinality of $2^m(2^{m_2}-2)$. Due to Lemma \ref{l021}, Lemma \ref{l022}, and Lemma \ref{l023}, $\sum_{p_j\in \Delta_1^*}d_{1j}\cdot 2^{|{\Bar D}_{1j}\setminus M^+_{2j}|}$ is the cardinality of the second part. Hence the formula. It is symmetric for  $p_i\in A_2(N)$. For all the other cases, it can be checked that the formula still holds.\qed\end{proof}

\section {Computing Penrose-Banzhaf and Shapley-Shubik indices}
Computing Penrose-Banzhaf and Shapley-Shubik indices is even more tedious than computing Holler-Packel and Deegan-Packel indices. This is quite expected because they rely on ${\mathcal W}{\mathcal C}$ (to be precise, ${\mathcal W}{\mathcal C}(i)$s) instead of ${\mathcal M}{\mathcal W}{\mathcal C}$, and we have already seen that the structure of  ${\mathcal W}{\mathcal C}$ is much more complicated than that of ${\mathcal M}{\mathcal W}{\mathcal C}$.

One more trouble is that the analysis of ${\mathcal W}{\mathcal
C}(i)$ is incomplete in the previous sections. To be precise, when $A_1(N)\cap A_2(N)\neq \emptyset$, the discussion is done by  Lemma \ref{l012} and Lemma \ref{l013}; When $A_1(N)\cap A_2(N)=\emptyset$ and either $A_1(N)$ or $A_2(N)$ is a winning coalition, the analysis is done by Lemma \ref{l014}. However, analysis in the most complicated case, i.e. $A_1(N)\cap A_2(N)=\emptyset$ and neither $A_1(N)$ nor $A_2(N)$ is a winning coalition, is not enough, though we have already had a clear understanding of the structure of ${\mathcal W}{\mathcal C}3(i)$ from the final part of the previous section.

\subsection{The simple cases}
Based on Lemma \ref{l012}, Lemma \ref{l013} and Lemma \ref{l014}, computing Penrose-Banzhaf and Shapley-Shubik indices in the case $A_1(N)\cap A_2(N)\neq \emptyset$ and the cases where $A_1(N)\cap A_2(N)=\emptyset$ and either $A_1(N)$ is winning or $A_2(N)$ is winning is straightforward. We omit the exact formulas to save space.

\subsection{The complex case}



We are left to discuss  the most complicated case, i.e. the case $A_1(N)\cap A_2(N)=\emptyset$
and neither $A_1(N)$ nor $A_2(N)$ is winning.

Above all, we need a concept similar to Definition \ref{dwc3i}.

\begin{definition} $\forall p_i\in N$, we use ${\mathcal W}{\mathcal
C}1(i)$ to denote the set of winning coalitions in ${\mathcal W}{\mathcal
C}1$ that contain $p_i$ as a decisive player, i.e. \begin{equation}{\mathcal W}{\mathcal
C}1(i)=\Big\{C\in {\mathcal W}{\mathcal
C}1: C\ni p_i, C\setminus \{p_i\}\notin {\mathcal W}{\mathcal
C}\Big\}.\end{equation}\end{definition}

\begin{definition}$\forall p_i, p_j\in N$, we use ${\mathcal W}{\mathcal C}1(i,j)$  to denote
the set of winning coalitions in ${\mathcal W}{\mathcal C}1_j$ such that $p_i$ is decisive, i.e.\begin{equation}{\mathcal W}{\mathcal C}1(i,j)={\mathcal W}{\mathcal C}1(i)\cap {\mathcal W}{\mathcal C}1_j.\end{equation} \end{definition}


\begin{lemma}In the two dimensional C-WMMG, suppose $A_1(N)\cap A_2(N)=\emptyset$ and neither $A_1(N)$ nor $A_2(N)$ is winning. $\forall p_{j_1}, p_{j_2}\in M$. If $p_{j_1}^2\neq p_{j_1}^2$, then ${\mathcal W}{\mathcal C}1(i,j_1)\cap {\mathcal W}{\mathcal C}1(i,j_2)=\emptyset$.
\end{lemma}
\begin{proof}Directly from Lemma \ref{l8}.\qed\end{proof}

However, even if $p_{j_1}^2=p_{j_1}^2$, we may not get that ${\mathcal W}{\mathcal C}1(i,j_1)={\mathcal W}{\mathcal C}1(i,j_2)$, because ${\mathcal W}{\mathcal C}1_{j_1}={\mathcal W}{\mathcal C}1_{j_2}$ is not guaranteed.

\begin{definition}$\forall p_i\in N$ and $1\leq s\leq n_2^0$, we define ${\mathcal W}{\mathcal C}1(i,M_{2s})$ as the set of coalitions in ${\mathcal W}{\mathcal C}1$ such that $p_i$ is decisive and busy-2 players are taken from $M_{2s}$, i.e. \begin{equation}{\mathcal W}{\mathcal C}1(i,M_{2s})=\bigcup_{p_j\in M_{2s}} {\mathcal W}{\mathcal C}1(i,j).\end{equation}\end{definition}

\begin{lemma}In the two dimensional C-WMMG, suppose $A_1(N)\cap A_2(N)=\emptyset$ and neither $A_1(N)$ nor $A_2(N)$ is winning. $\forall 1\leq s\leq n_2^0$, and suppose also that $q_2\big(M_{2s})\geq q_2\big(A_1(N)\cup D_1(M_{2s})\big)$.

(a) If $p_i\in A_1(N)\cup D_{1}(M_{2s})$,
then \begin{equation}\left|{\mathcal W}{\mathcal C}1(i,M_{2s})\right|=2^{\big|M_{2s}\setminus D_1(M_{2s})\big|+\big|E_1(M_{2s})\big|}-\alpha_{1}^s\cdot 2^{\big|E_1(M_{2s})\big|};\end{equation}

(b) If $p_i\in E_{1}(M_{2s})$, then $|{\mathcal W}{\mathcal C}1(i,M_{2s})|=0$.\label{l12}\end{lemma}
\begin{proof}Due to Lemma \ref{l18}(b), $q_2\big(M_{2s})\geq q_2\big(A_1(N)\cup D_1(M_{2s})\big)$ implies that $\bigcup_{p_j\in M_{2s}} {\mathcal W}{\mathcal C}1_j\neq \emptyset$. To verify this lemma, it is sufficient to notice from Lemma \ref{l18}(c) that all the players in $A_1(N)\cup D_{1}(M_{2s})$ are always decisive and players in $E_{1}(M_{2s})$ are always bench-warmers.\qed\end{proof}

\begin{definition}Let $s_2(i)$ be the index $s$ in $\{1,2,\cdots,n_2^0\}$ such that $p_i\in M_{2s}$.\end{definition}

Lemma \ref{l12}, however, didn't give us a complete description of all the ${\mathcal W}{\mathcal C}1(i,M_{2s})$s. To be precise,  ${\mathcal W}{\mathcal C}1(i,M_{2s_2(i)})$ may not be covered. For each $p_i\in M$, if $p_i\in D_{1}(M_{2s_2(i)})$, then Lemma \ref{l12} is applicable. In the case that $p_i\notin D_{1}(M_{2s_2(i)})$, calculating $|{\mathcal W}{\mathcal C}1(i,M_{2s_2(i)})|$ needs more
delicate analysis.

First of all, we show that the problem of computing $|{\mathcal W}{\mathcal C}1(i,M_{2s_2(i)})|$ can be reduced to computing $|{\mathcal W}{\mathcal C}1(i,i)|$. In fact, the two sets are identical.

\begin{lemma}$\forall p_i\in M$, we have\begin{equation}{\mathcal W}{\mathcal C}1(i,M_{2s_2(i)})={\mathcal W}{\mathcal C}1(i,i).\end{equation}\label{l28}\end{lemma}
\begin{proof}By definition, it is trivial that ${\mathcal W}{\mathcal C}1(i,i)\subseteq {\mathcal W}{\mathcal C}1(i,M_{2s_2(i)})$. On the other hand,  $\forall X\in {\mathcal W}{\mathcal C}1(i,M_{2s_2(i)})$,$p_i$ must be a busy-2 player of $X$, because busy-2 players of $X$ are drawn from $M_{2s_2(i)}$, and $p_i$ has the same second dimension as any one in $M_{2s_2(i)}$. Hence it also holds that ${\mathcal W}{\mathcal C}1(i,M_{2s_2(i)})\subseteq {\mathcal W}{\mathcal C}1(i,i)$ and the lemma is valid.\qed\end{proof}

Second of all, let us show that checking whether ${\mathcal W}{\mathcal C}1(i,i)=\emptyset$ or not is easy.

\begin{lemma}${\mathcal W}{\mathcal C}1(i,i)\neq \emptyset$ if and only if $C_{1i}\in {\mathcal M}{\mathcal W}{\mathcal C}1_i$, i.e.  $C_{1i}$ is an MWC with $p_i$ being a busy-2 player.\label{l27}\end{lemma}
\begin{proof}Sufficiency is self-evident. Suppose now ${\mathcal W}{\mathcal C}1(i,i)\neq \emptyset$, then there exists $C_{1i}\cup E\in {\mathcal W}{\mathcal C}1(i,i)$, where $E\subseteq E_{1i}$. By definition of $E_{1i}$, we know that $q(C_{1i}\cup E)=q(C_{1i})$ and $q(C_{1i}^-)=q((C_{1i}\cup E)^-)$. Therefore, $C_{1i}$ is still winning, and $p_i$ is still busy-2 and decisive in $C_{1i}$. Hence $C_{1i}\in {\mathcal M}{\mathcal W}{\mathcal C}1_i$. \qed \end{proof}

Notice that when $C_{1i}\in {\mathcal M}{\mathcal W}{\mathcal C}1_i$, we still cannot guarantee that $p_i$ is decisive in each $C_{1i}\cup E$.
In fact, the main complication in the following analysis is that $\forall X\in{\mathcal W}{\mathcal C}1_i$, $p_i$ may not be decisive in $X$ (the only information we know is that she is busy). A naive idea is to check all the coalitions in ${\mathcal W}{\mathcal C}1_i$, one by one, to see whether $p_i$ is decisive. However, this set may have an exponential cardinality, so the naive idea does not work.

 We find that analyzing the complementary coalitions (added by $p_i$) may be more convenient than directly analyzing ${\mathcal W}{\mathcal C}1(i,i)$.

\begin{definition}$\forall p_i\in M$ such that $p_i\notin D_{1i}$, let ${\mathcal H}1_i$ be
the set of coalitions constructed by moving $p_i$ from coalitions in ${\mathcal W}{\mathcal C}1(i,i)$ to the corresponding complementary coalitions, i.e.
\begin{equation}{\mathcal H}1_i=\Big\{X^-\cup \{p_i\}: X\in {\mathcal W}{\mathcal C}1(i,i)\Big\}.\end{equation}\end{definition}

\begin{lemma}When ${\mathcal W}{\mathcal C}1(i,i)\neq \emptyset$, we can rewrite ${\mathcal H}1_i$ as \begin{equation}{\mathcal H}1_i=\Big\{X=(C_{1i}\cup E_{1i})^-\cup\{p_i\}\cup E:E\subseteq
E_{1i}, q(X)\geq q(X^-)\Big\}.\end{equation}\label{l30}\end{lemma}
\begin{proof}$\forall X_0\in {\mathcal W}{\mathcal C}1(i,i)$. Since ${\mathcal W}{\mathcal C}1(i,i)\subseteq {\mathcal W}{\mathcal C}1_i=\{C_{1i}\cup E: E\subseteq E_{1i}\}$, we know there is an $E_0\subseteq E_{1i}$ such that $X_0=C_{1i}\cup E_0$. Therefore,  $X_0^-\cup \{p_i\}=(C_{1i}\cup E_{1i})^-\cup \{p_i\}\cup (E_{1i}\setminus E_0)$. Let $X=X_0^-\cup \{p_i\}$. Because $E_{1i}\setminus E_0\subseteq E_{1i}$ and $q(X)\geq q(X^-)$ ($p_i$ is decisive in $X_0$), we know that $X$ belongs to the right hand side set of this lemma. The other direction of inclusion is also easy.\qed\end{proof}

Then \begin{equation}|{\mathcal W}{\mathcal C}1(i,i)|=|{\mathcal H}1_i|.\label{91}\end{equation}

Still, the definition of ${\mathcal H}1_i$, as well as the equivalent formula, is but an expression. To calculate its cardinality, we need to do some further decomposition such that all the elements are efficiently computable.

\begin{definition}We denote ${\mathcal B}1_i$ as the subset of ${\mathcal H}1_i$ where no busy-1 player is in the corresponding $E$, i.e.
 \begin{equation}{\mathcal B}1_i=\Big\{X=(C_{1i}\cup E_{1i})^-\cup \{p_i\}\cup
E: E\subseteq
E_{1i},X\in {\mathcal H}1_i, A_1(X)\cap E=\emptyset\Big\},\end{equation} and for all $p_j\in E_{1i}$, let ${\mathcal B}1_{ij}$ be the subset of ${\mathcal H}1_i$ where $p_j$ is a busy-1 player, i.e.
\begin{equation}{\mathcal B}1_{ij}=\Big\{X=(C_{1i}\cup E_{1i})^-\cup \{p_i\}\cup E: E\subseteq
E_{1i}, X\in {\mathcal H}1_i, p_j\in A_1(X)\cap E\Big\}.\end{equation}\end{definition}

To ensure that ${\mathcal B}1_{ij}$ is nonempty, $p_j^1$ should be big enough.

\begin{definition}$\forall p_i\in M$ such that $p_i\notin D_{1i}$. Let $E^0_{1i}$ be the set of  players in $E_{1i}$ whose first dimensions are large enough such that they may be busy-1 players in ${\mathcal B}1_{ij}$:
\begin{equation}E^0_{1i}=\Big\{p_j\in E_{1i}: p_j^1\geq q_1\big((C_{1i}\cup E_{1i})^-\cup
\{p_i\}\big)\Big\}.\end{equation}\end{definition}

Then $\forall p_j\in E_{1i}\setminus E^0_{1i}$, we have ${\mathcal B}1_{ij}=\emptyset$. Therefore, we can rewrite ${\mathcal H}1_i$ as
\begin{equation}{\mathcal H}1_i={\mathcal B}1_i\cup \bigcup_{p_j\in E_{1i}^0}{\mathcal
B}1_{ij}.\label{95}\end{equation}

It is obvious  that for all $p_j\in E_{1i}^0$,
\begin{equation}{\mathcal B}1_i\cap{\mathcal B}1_{ij}=\emptyset,\label{96}\end{equation} because any pair of coalitions from  the two families, respectively, have different busy-1 player sets.

\begin{definition}$\forall p_i\in M$ such that $p_i\notin D_{1i}$. We use ${\hat E}_{1i}$ and  ${\hat D}_{1i}$ to denote the set of potentially optional players, and the set of blocking players, respectively: \begin{equation}{\hat E}_{1i}=\Big\{p_j\in E_{1i}: q\big((C_{1i}\setminus\{p_i\})\cup \{p_j\}\big)\leq
q\big((C_{1i}\cup E_{1i})^-\cup \{p_i\}\big)\Big\},\end{equation}
\begin{equation}{\hat D}_{1i}=\Big\{p_j\in E_{1i}: p_j^1\geq q_1\big((C_{1i}\cup E_{1i})^-\cup \{p_i\}\big)\Big\}.\end{equation}

We also use $\delta_{1i}$ to denote the indicator of whether ${\hat D}_{1i}\setminus {\hat E}_{1i}=\emptyset$, i.e.
\begin{equation}\delta_{1i}=\left\{\begin{array}{cc}1&if~{\hat D}_{1i}\setminus {\hat E}_{1i}=\emptyset\\
0&otherwise\end{array}\right..\end{equation}\end{definition}

\begin{lemma}In the two dimensional C-WMMG, suppose $A_1(N)\cap A_2(N)=\emptyset$ and neither $A_1(N)$ nor $A_2(N)$ is winning. $\forall p_i\in M$ such that $p_i\notin D_{1i}$. Suppose  ${\mathcal W}{\mathcal C}1(i,i)\neq \emptyset$. Then
\begin{equation}{\mathcal B}1_i=\left\{\begin{array}{cc}\Big\{(C_{1i}\cup E_{1i})^-\cup \{p_i\}\cup E: E\subseteq {\Hat E}_{1i}\setminus {\Hat D}_{1i}\Big\}&if~{\Hat D}_{1i}\setminus {\Hat E}_{1i}=\emptyset\\
\emptyset&otherwise\end{array}\right.,\label{175}\end{equation}
and therefore
 \begin{equation}|{\mathcal B}1_i|=\delta_{1i}\cdot 2^{|{\hat E}_{1i}\setminus {\hat D}_{1i}|}.\end{equation}\label{l31}\end{lemma}
\begin{proof}$\forall E\subseteq E_{1i}$, let $X(E)=(C_{1i}\cup E_{1i})^-\cup \{p_i\}\cup
E$, then  the corresponding coalition in ${\mathcal W}{\mathcal C}1(i,i)$ is \begin{equation}Y(E)=C_{1i}\cup (E_{1i}\setminus E)=A_1(N)\cup D_{1i}\cup \{p_i\}\cup (E_{1i}\setminus E).\end{equation}

By definition, $X(E)\in {\mathcal B}1_{i}$ if and only if the following four conditions hold: (i) $p_i$ is busy-2 in $Y(E)$, (ii)  $Y(E)$ is winning, (iii) $p_i$ is decisive in $Y(E)$, (iv) $A_1(X(E))\cap E=\emptyset$.


When ${\Hat D}_{1i}\setminus {\Hat E}_{1i}\neq\emptyset$, let $p_j\in A_1\big({\Hat D}_{1i}\setminus {\Hat E}_{1i}\big)$. Obviously, $p_j$ should not be contained in $E$, because otherwise, from $p_j\in {\Hat D}_{1i}$, we would have $p_j\in A_1(X(E))\cap E$. Hence it must be true that $p_j\in Y(E)$.  However, $p_j\in E_{1i}\setminus {\Hat E}_{1i}$ implies that  $q(X(E)^-)>q(X(E))$, a contradiction with condition (iii).

Suppose now ${\Hat D}_{1i}\setminus {\Hat E}_{1i}=\emptyset$. On the one hand, $\forall X(E)\in {\mathcal B}1_i$, condition (iv) implies that \begin{equation}q(X(E))=q\big((C_{1i}\cup E_{1i})^-\cup \{p_i\}\big).\end{equation} Therefore, we must have \begin{equation}E\subseteq {\hat E}_{1i},\label{g171}\end{equation} because any player in  $E_{1i}\setminus{\hat E}_{1i}$ belonging to $E$ can guarantee that $Y(E)\setminus \{p_i\}$ is still winning, a contradiction with condition (iii). Also, \begin{equation}E\cap {\Hat D}_{1i}=\emptyset\label{g172}\end{equation} must hold because otherwise condition (iv) would be violated. Combining (\ref{g171}) and (\ref{g172}) we arrive at $E\subseteq {\hat E}_{1i}\setminus {\hat D}_{1i}$

On the other hand, for any subset $E$ of  ${\Hat E}_{1i}\setminus {\Hat D}_{1i}$, we prove that $X(E)$ and $Y(E)$ satisfy the four conditions. Due to Lemma \ref{l9}(b), Lemma \ref{l27}, and the hypothesis that ${\mathcal W}{\mathcal C}1(i,i)\neq \emptyset$, we know that condition (i) and condition (ii) are both true. Condition (iv) is true because $E\cap {\Hat D}_{1i}=\emptyset$. Condition (iii) is valid due to condition (iv) and $E\subseteq {\Hat E}_{1i}$.

To sum up, (\ref{175}) is valid, and hence the lemma.\qed \end{proof}

\begin{definition}$\forall p_i\in M$ such that $p_i\notin D_{1i}$. We define $E^{00}_{1i}$ as a refinement of $E^{0}_{1i}$: \begin{equation}E^{00}_{1i}=\Big\{p_j\in E_{1i}^0: {\Hat D}_{1ij}\subseteq {\Hat E}_{1ij}\Big\},\end{equation}
where
\begin{equation}{\Hat E}_{1ij}=\Big\{p_k\in E_{1i}: k\neq j, q\big((C_{1i}\setminus\{p_i\})\cup \{p_k\}\big)\leq
p^1_{j}+q_2(N)\Big\},\end{equation}
\begin{equation}{\Hat D}_{1ij}=\Big\{p_k\in E_{1i}: p_k^1>p_j^1\Big\},\end{equation}
 for all $p_{j}\in E_{1i}^0$.
\end{definition}

\begin{lemma}In the two dimensional C-WMMG, suppose $A_1(N)\cap A_2(N)=\emptyset$ and neither $A_1(N)$ nor $A_2(N)$ is winning. $\forall p_i\in M$ such that $p_i\notin D_{1i}$. Suppose  also ${\mathcal W}{\mathcal C}1(i,i)\neq \emptyset$.

(a) If $p_j\in E^0_{1i}\setminus E^{00}_{1i}$, then ${\mathcal B}_{1ij}=\emptyset$.

(b) For any $p_j,
p_k\in E_{1i}^{00}$, if $p_j^1\neq p_k^1$, then ${\mathcal B}1_{ij}\cap {\mathcal B}1_{ik}=\emptyset.$
\label{l32b} \end{lemma}
\begin{proof}
(a) By definition, $p_j\in E^0_{1i}\setminus E^{00}_{1i}$ means that ${\Hat D}_{1ij}\nsubseteq{\Hat E}_{1ij}$. Therefore, there exists a player $w$ who is in ${\Hat D}_{1ij}$ but not in ${\Hat E}_{1ij}$. $\forall X\in {\mathcal B}_{1ij}$, using an argument similar to that  in the proof to Lemma \ref{l31}, we can show that there is always a contradiction, regardless of whether $w\in X$ or not. Hence such $X$ does not exist, i.e. ${\mathcal B}_{1ij}=\emptyset$.

(b) Self-evident by definition.\qed\end{proof}

However, when $p_j^1=p_k^1$, we cannot get ${\mathcal B}1_{ij}={\mathcal B}1_{ik}.$

 Still as in Section \ref{wc}, we let $\{M_{11}, M_{12},
\cdots, M_{1n^1_0}\}$ be a partition of $M$, such that players in the same subset have the same
first dimensions and players from different subsets have different first  dimensions.

\begin{definition}$\forall p_i\in M$ such that $p_i\notin D_{1i}$, and $1\leq s\leq n^1_0$, let \begin{equation}{\Hat E}_{1i}^s=\Big\{p_k\in E_{1i}: q((C_{1i}\setminus
\{p_i\})\cup \{p_k\})\leq q_1(M_{1s})+q_2(N)\Big\},\end{equation}
\begin{equation}{\Hat D}^s_{1i}=\Big\{p_j\in
E_{1i}: p_j>q_1(M_{1s})\Big\},\end{equation}
\begin{equation}\delta_{1i}^s=\left\{\begin{array}{cc}1&if~{\hat D}^s_{1i}\setminus {\hat E}^s_{1i}=\emptyset\\
0&otherwise\end{array}\right..\end{equation}
\end{definition}

\begin{lemma}In the two dimensional C-WMMG, suppose $A_1(N)\cap A_2(N)=\emptyset$ and neither $A_1(N)$ nor $A_2(N)$ is winning. $\forall p_i\in M$ such that $p_i\notin D_{1i}$.

(a) If $\delta_{1i}^s=0$, then $\bigcup_{p_j\in M_{1s}\cap E_{1i}^{00}}{\mathcal B}1_{ij}=\emptyset$.

(b) If $\delta_{1i}^s=1$, then \begin{equation}\bigcup_{p_j\in M_{1s}\cap E_{1i}^{00}}{\mathcal B}1_{ij}=\Big\{(C_{1i}\cup E_{1i})^-\cup
\{p_i\}\cup E: E\subseteq {\Hat E}_{1i}^s\setminus {\Hat D}^s_{1i},E\cap (E_{1i}^{00}\cap
M_{1s})\neq \emptyset\Big\}.\end{equation}

(c)\begin{equation}\left|\bigcup_{p_j\in M_{1s}\cap E_{1i}^{00}}{\mathcal B}1_{ij}\right|=
\delta_{1i}^s\left(2^{\left|{\Hat E}_{1i}^s\setminus {\Hat D}^s_{1i}\right|}
-2^{\left|\left({\Hat E}_{1i}^s\setminus {\Hat D}^s_{1i}\right)\setminus
\left(E_{1i}^{00}\cap M_{1s}\right)
\right|}\right).\end{equation}\label{l32}
\end{lemma}
\begin{proof}(a) $\delta_{1i}^s=0$ means that there exists a player in ${\Hat E}_{1i}^s\setminus {\Hat D}^s_{1i}$. $\forall p_j\in M_{1s}\cap E_{1i}^{00}$, $\forall X\in {\mathcal B}1_{ij}$, using a similar argument as in the proof of Lemma \ref{l31}, we can derive a contradiction. Hence such $X$ does not exist, i.e. ${\mathcal B}1_{ij}=\emptyset$.

 (b) Similar to the proof of Lemma \ref{l31}. We only need to be careful to let the coalitions in $\bigcup_{p_j\in M_{1s}\cap E_{1i}^{00}}{\mathcal B}1_{ij}$ contain at least one member of $E_{1i}^{00}\cap M_{1s}$.

 (c) Straightforward from (a) and (b). Notice that $2^{\left|\left({\Hat E}_{1i}^s\setminus {\Hat D}^s_{1i}\right)\setminus
\left(E_{1i}^{00}\cap M_{1s}\right) \right|}$ is the number of coalitions where no player in $E_{1i}^{00}\cap M_{1s}$ is included. Such coalitions are not qualified and we need to exclude them.\qed\end{proof}

\begin{lemma}In the two dimensional C-WMMG, suppose $A_1(N)\cap A_2(N)=\emptyset$ and neither $A_1(N)$ nor $A_2(N)$ is winning. $\forall p_i\in M$ such that $p_i\notin D_{1i}$: \begin{equation}|{\mathcal W}{\mathcal C}1(i,i)|
=\delta_{1i}\cdot 2^{|{\Hat E}_{1i}\setminus {\Hat D}_{1i}|}+\sum_{1\leq s\leq
n_1^0}\delta_{1i}^s\left(2^{\left|{\Hat E}_{1i}^s\setminus {\Hat D}^s_{1i}\right|}
-2^{\left|\left({\Hat E}_{1i}^s\setminus {\Hat D}^s_{1i}\right)\setminus
\left(E_{1i}^{00}\cap M_{1s}\right)
\right|}\right).\end{equation}\label{l13}\end{lemma}
\begin{proof}This lemma is a combination of (\ref{91}), (\ref{95}), (\ref{96}), Lemma \ref{l31}, Lemma \ref{l32b} and Lemma \ref{l32}(c).\qed\end{proof}

We remark that when $p_i\in D_{1i}$, then the above formula should be consistent with the one in part (a) of Lemma \ref{l12}, because the condition that $p_i\notin D_{1i}$ is never used.

Similarly, we define ${\mathcal W}{\mathcal C}2(i,j), {\mathcal H}2_i, E_{2i}^0, {\Hat E}_{2i},
{\Hat D}_{2i}, {\mathcal B}_{2i}, {\mathcal B}_{2ij},$ ${\Hat E}_{2ij}, {\Hat D}_{2ij}, E_{2i}^{00},
{\Hat E}_{2i}^s,$  $\delta_{2i}, \delta_{2i}^s$, then
 results parallel to  Lemma \ref{l12},  Lemma \ref{l28}, Lemma \ref{l27},  Lemma \ref{l30},  Lemma \ref{l31}, Lemma \ref{l32b}, Lemma \ref{l32} and Lemma \ref{l13} hold.

\subsubsection {Penrose-Banzhaf indices}




\begin{definition} We use $M^{1*}$ to denote the set of $s$ such that $\bigcup_{p_j\in M_{2s}} {\mathcal W}{\mathcal C}1_j\neq \emptyset$. According to Lemma \ref{l18}(b), \begin{equation}M^{1*}=\Big\{1\leq s\leq n_2^0: q_2\big(M_{2s})\geq q_2\big(A_1(N)\cup D_1(M_{2s})\big)\Big\}.\end{equation}\end{definition}

 Similarly, we define:
  \begin{equation}M^{2*}=\Big\{1\leq s\leq n_1^0: q_1\big(M_{1s})\geq q_1\big(A_1(N)\cup D_2(M_{1s})\big)\Big\}.\end{equation}

  While \begin{equation}{\mathcal W}{\mathcal C}1(i)\subseteq {\mathcal W}{\mathcal C}1\end{equation} and \begin{equation}{\mathcal W}{\mathcal C}1=\bigcup_{1\leq s\leq n_2^0}\left(\cup_{p_j\in M_2^s}{\mathcal W}{\mathcal C}1_j\right),\end{equation}
we have the following useful decomposition:
\begin{eqnarray}{\mathcal W}{\mathcal C}1(i)&=&{\mathcal W}{\mathcal C}1(i)\cap {\mathcal W}{\mathcal C}1\\
&=&{\mathcal W}{\mathcal C}1(i)\bigcap \left(\bigcup_{1\leq s\leq n_2^0}\left(\cup_{p_j\in M_2^s}{\mathcal W}{\mathcal C}1_j\right)\right)\\
&=&\bigcup_{1\leq s\leq n_2^0}\left({\mathcal W}{\mathcal C}1(i)\bigcap\left(\cup_{p_j\in M_2^s}{\mathcal W}{\mathcal C}1_j\right)\right)\\
&=&\bigcup_{1\leq s\leq n_2^0}{\mathcal W}{\mathcal C}1(i,M_{2s})\\
&=&\bigcup_{s\in M^{1*}}{\mathcal W}{\mathcal C}1(i,M_{2s}).\label{205}\end{eqnarray}

 Parallel decomposition holds for ${\mathcal W}{\mathcal C}2(i)$.

We need two final notations.
\begin{definition} $\forall p_i\in M$, let $\lambda_{1i}$ be the indicator of whether $|{\mathcal W}{\mathcal C}1(i,i)|$ is nonzero and should be counted in calculating  $|{\mathcal W}{\mathcal C}1(i)|$. According to Lemma \ref{l27}, \begin{equation}\lambda_{1i}=\left\{\begin{array}{ll}1& if p_i\notin D_{1}(M_{2s_2(i)})~ and~ C_{1i}\in {\mathcal M}{\mathcal W}{\mathcal C}1_i \\
0& otherwise\end{array}\right.,\end{equation} \end{definition}

Similarly, we define\begin{equation}\lambda_{2i}=\left\{\begin{array}{ll}1& if p_i\notin D_{2}(M_{1s_1(i)})~ and~ C_{2i}\in {\mathcal M}{\mathcal W}{\mathcal C}2_i \\
0& otherwise\end{array}\right..\end{equation}

\begin{theorem}In the two dimensional C-WMMG, suppose $A_1(N)\cap A_2(N)=\emptyset$ and neither $A_1(N)$ nor $A_2(N)$ is winning.

(a1) If $p_i\in A_1(N)$, then \begin{eqnarray}pb_i&=&\frac{1}{2^{n-1}}\left(2^m(2^{m_2}-2)+\sum\limits_{p_j\in \Delta_1^*}d_{1j}\cdot 2^{|{\Bar D}_{1j}\setminus M^+_{2j}|}+\right.\\
&&\left.\sum\limits_{s\in M^{1*}}\left(2^{\big|M_{2s}\setminus D_1(M_{2s})\big|+\big|E_1(M_{2s})\big|}-\alpha_{1}^s\cdot 2^{\big|E_1(M_{2s})\big|}\right)\right);\end{eqnarray}

(a2) If $p_i\in A_2(N)$, then \begin{eqnarray}pb_i&=&\frac{1}{2^{n-1}}\left(2^m(2^{m_1}-2)+\sum\limits_{p_j\in \Delta_2^*}d_{2j}\cdot 2^{|{\Bar D}_{2j}\setminus M^+_{1j}|}+\right.\\
&&\left.\sum\limits_{s\in M^{2*}}\left(2^{\big|M_{1s}\setminus D_2(M_{1s})\big|+\big|E_2(M_{1s})\big|}-\alpha_{2}^s\cdot 2^{\big|E_2(M_{1s})\big|}\right)\right);\end{eqnarray}

(b) If $p_i\in M$, then \begin{eqnarray}pb_i
&=&\frac{1}{2^{n-1}}\left\{\sum_{s\in M^{1*}:p_i\in D_1(M_{2s})}\left(2^{\big|M_{2s}\setminus D_1(M_{2s})\big|+\big|E_1(M_{2s})\big|}-\alpha_{1}^s\cdot 2^{\big|E_1(M_{2s})\big|}\right)\right.\\
&&+\sum_{s\in M^{2*}:p_i\in D_2(M_{1s})}\left(2^{\big|M_{1s}\setminus D_2(M_{1s})\big|+\big|E_2(M_{1s})\big|}-\alpha_{2}^s\cdot 2^{\big|E_2(M_{1s})\big|}\right)\\
&&\left.+\sum_{t=1,2}\lambda_{ti}\left[\delta_{ti}\cdot 2^{\big|{\Hat E}_{ti}\setminus {\Hat D}_{ti}\big|}+\sum_{1\leq
s\leq n_0^t}\delta_{ti}^s\cdot\left(2^{\big|{\Hat E}_{ti}^s\setminus {\Hat D}_{ti}^{s}\big|}- 2^{\big|{\Hat E}_{ti}^s\setminus
{\Hat D}_{ti}^{s}\big|-\big|E_{ti}^{00}\cap M_{ts}\big|}\right)\right]\right\}.\end{eqnarray}\label{t12}
\end{theorem}

\begin{proof}By definition of Penrose-Banzhaf index, $\forall p_i\in N$, we have:\begin{equation}pb_i=\frac{1}{2^{n-1}}\Big(\left|{\mathcal W}{\mathcal C}3(i)\right|+\left|{\mathcal W}{\mathcal C}1(i)\right|+\left|{\mathcal W}{\mathcal C}2(i)\right|\Big).\label{197}\end{equation}

(a1) First of all, $\forall p_i\in A_1(N)$, since ${\mathcal W}{\mathcal C}2$ is the set of winning coalitions where no player in $A_1(N)$ is included,  we know that \begin{equation}{\mathcal W}{\mathcal C}2(i)=\emptyset.\label{198}\end{equation}

Second of all, \begin{equation}{\mathcal W}{\mathcal C}1(i,M_{2s})=\bigcup_{p_j\in M_{2s}}{{\mathcal W}{\mathcal C}1_j},\label{210}\end{equation}
because $p_i\in A_1(N)$ implies that $p_i$ is decisive in each coalition of $\cup_{p_j\in M_{2s}}{{\mathcal W}{\mathcal C}1_j}$.

Combining  (\ref{205}) (\ref{197}) (\ref{198}) (\ref{210}), Lemma \ref{l18}(a)(d), and Lemma \ref{l14}, we get part (a1) of this theorem.

(a2) Symmetric to (a1).

(b) For $p_i\in M$, since ${\mathcal W}{\mathcal C}3(i)=\emptyset$ (Lemma \ref{l012}), we have:
\begin{equation}pb_i=\left|{\mathcal W}{\mathcal C}1(i)\right|+\left|{\mathcal W}{\mathcal C}2(i)\right|.\end{equation}

$|{\mathcal W}{\mathcal C}1(i)|$ can be decomposed further as: \begin{equation}|{\mathcal W}{\mathcal C}1(i)|=\sum_{s\in M^{1*}:p_i\in D_1(M_{2s})}|{\mathcal W}{\mathcal C}1(i,M_{2s})|+\sum_{s\in M^{1*}:p_i\notin D_1(M_{2s})}|{\mathcal W}{\mathcal C}1(i,M_{2s})|.\label{212}\end{equation}

For each $s\neq s_2(i)$, $p_i\notin D_1(M_{2s})$ implies that $p_i\in E_1(M_{2s})$. By Lemma \ref{l12}(b), \begin{equation}\sum_{s\in M^{1*}\setminus \{s_2(i)\}:p_i\notin D_1(M_{2s})}|{\mathcal W}{\mathcal C}1(i,M_{2s})|=0.\label{213}\end{equation}

Combining (\ref{212})(\ref{213}), we have \begin{eqnarray}|{\mathcal W}{\mathcal C}1(i)|&=&\sum_{s\in M^{1*}:p_i\in D_1(M_{2s})}|{\mathcal W}{\mathcal C}1(i,M_{2s})|+\lambda_{1i}|{\mathcal W}{\mathcal C}1(i,M_{2s_2(i)})|\\
&=&\sum_{s\in M^{1*}:p_i\in D_1(M_{2s})}|{\mathcal W}{\mathcal C}1(i,M_{2s})|+\lambda_{1i}|{\mathcal W}{\mathcal C}1(i,i)|,\end{eqnarray}
where the second equality is valid because of Lemma \ref{l28}.

We have a similar formula for  $|{\mathcal W}{\mathcal C}2(i)|$.

 When $p_i\in D_1(M_{2s})$, Lemma \ref{l12}(a) is applicable to compute $|{\mathcal W}{\mathcal C}1(i,M_{2s})|$. Lemma \ref{l13} is used to compute $|{\mathcal W}{\mathcal C}1(i,i)|$.
Based on the above discussion, part (b) of this theorem is valid.\qed
\end{proof}

\subsubsection {Shapley-Shubik indices}

For computing Shapley-Shubik indices, we have the following formulas.
\begin{theorem}In the two dimensional C-WMMG, suppose $A_1(N)\cap A_2(N)=\emptyset$ and neither $A_1(N)$ nor $A_2(N)$ is winning.

(a1) If $p_i\in A_1(N)$, then \begin{eqnarray}&&ss_i=\sum\limits_{c=1}^{m_2-1}\sum\limits_{d=0}^{m}C_{m_2}^c\cdot C_m^d \cdot \frac{(m_1+c+d-1)!(n-m_1-c-d)!}{n!}\label{216}\\
&&+\sum\limits_{p_j\in \Delta_1^*}\sum\limits_{d=0}^{|{\Bar D}_{1j}\setminus M^+_{2j}|} d_{1j}\cdot C_{|{\Bar D}_{1j}\setminus M^+_{2j}|}^{d}\cdot \frac{(m_2+d)!(n-m_2-d-1)!}{n!}\label{217}\\
&&+\sum_{s\in M^{1*}}\left(\sum\limits_{c=0}^{|M_{2s}\setminus D_1(M_{2s})|}\sum\limits_{d=0}^{|E_1(M_{2s})|}C_{|M_{2s}\setminus D_1(M_{2s})|}^c\cdot C_{|E_1(M_{2s})|}^d\right.\label{218}\\
&&\cdot \frac{(m_1+|D_1(M_{2s})|+c+d-1)!(n-m_1-|D_1(M_{2s})|-c-d)!}{n!}\label{219}\\
&&\left.-\alpha_1^s\sum_{d=0}^{|E_1(M_{2s})|}C_{|E_1(M_{2s})|}^d\cdot \frac{(m_1+|D_1(M_{2s})|+d-1)!(n-m_1-|D_1(M_{2s})|-d)!}{n!}\right),\label{220}\end{eqnarray}
where  the right hand side of (\ref{216}) is defined as zero when $m_2=1$;

(a2) If $p_i\in A_1(N)$, then \begin{eqnarray}&&ss_i=\sum\limits_{c=1}^{m_1-1}\sum\limits_{d=0}^{m}C_{m_1}^c\cdot C_m^d \cdot \frac{(m_2+c+d-1)!(n-m_2-c-d)!}{n!}\label{221}\\
&&+\sum\limits_{p_j\in \Delta_2^*}\sum\limits_{d=0}^{|{\Bar D}_{2j}\setminus M^+_{1j}|} d_{2j}\cdot C_{|{\Bar D}_{2j}\setminus M^+_{1j}|}^{d}\cdot \frac{(m_1+d)!(n-m_1-d-1)!}{n!}\\
&&+\sum_{s\in M^{2*}}\left(\sum\limits_{c=0}^{|M_{1s}\setminus D_2(M_{1s})|}\sum\limits_{d=0}^{|E_2(M_{1s})|}C_{|M_{1s}\setminus D_2(M_{1s})|}^c\cdot C_{|E_2(M_{1s})|}^d\right.\\
&&\cdot \frac{(m_2+|D_2(M_{1s})|+c+d-1)!(n-m_2-|D_2(M_{1s})|-c-d)!}{n!}\\
&&\left.-\alpha_2^s\sum_{d=0}^{|E_2(M_{1s})|}C_{|E_2(M_{1s})|}^d\cdot \frac{(m_2+|D_2(M_{1s})|+d-1)!(n-m_2-|D_2(M_{1s})|-d)!}{n!}\right),\end{eqnarray}
where the right hand side of (\ref{221}) is defined as zero when $m_1=1$;

(b) If $p_i\in M$, then \begin{eqnarray}&&ss_i=\sum_{s\in M^{1*}:p_i\in D_1(M_{2s})}\left(\sum\limits_{c=0}^{|M_{2s}\setminus D_1(M_{2s})|}\sum\limits_{d=0}^{|E_1(M_{2s})|}C_{|M_{2s}\setminus D_1(M_{2s})|}^c\cdot C_{|E_1(M_{2s})|}^d\right.\label{226}\\
&&\cdot \frac{(m_1+|D_1(M_{2s})|+c+d-1)!(n-m_1-|D_1(M_{2s})|-c-d)!}{n!}\label{227}\\
&&\left.-\alpha_1^s\sum_{d=0}^{|E_1(M_{2s})|}C_{|E_1(M_{2s})|}^d\cdot \frac{(m_1+|D_1(M_{2s})|+d-1)!(n-m_1-|D_1(M_{2s})|-d)!}{n!}\right)\label{228}\\
&&+\sum_{s\in M^{2*}: p_i\in D_2(M_{1s})}\left(\sum\limits_{c=0}^{|M_{1s}\setminus D_2(M_{1s})|}\sum\limits_{d=0}^{|E_2(M_{1s})|}C_{|M_{1s}\setminus D_2(M_{1s})|}^c\cdot C_{|E_2(M_{1s})|}^d\right.\label{229}\\
&&\cdot \frac{(m_2+|D_2(M_{1s})|+c+d-1)!(n-m_2-|D_2(M_{1s})|-c-d)!}{n!}\label{230}\\
&&\left.-\alpha_2^s\sum_{d=0}^{|E_2(M_{1s})|}C_{|E_2(M_{1s})|}^d\cdot \frac{(m_2+|D_2(M_{1s})|+d-1)!(n-m_2-|D_2(M_{1s})|-d)!}{n!}\right)\label{231}\\
&&+\sum_{t=1,2}\lambda_{ti}\left[\delta_{ti}\cdot
\sum_{l=0}^{|{\Hat E}_{ti}\setminus {\Hat D}_{ti}|}C_{|{\Hat E}_{ti}\setminus {\Hat D}_{ti}|}^l\frac{(|C_{ti}|+|E_{ti}|-l-1)!(n-|C_{ti}|-|E_{ti}|+l)!}{n!}\right.\label{232}\\
&&+\sum_{1\leq s\leq n_0^t}\delta_{ti}^s\cdot\left(\sum_{l=0}^{|{\Hat E}_{ti}^s\setminus {\Hat D}_{ti}^{s}|}
C_{|{\Hat E}_{ti}^s\setminus {\Hat D}_{ti}^{s}|}^l\frac{(|C_{ti}|+|E_{ti}|-l-1)!(n-|C_{ti}|-|E_{ti}|+l)!}{n!}\right.\label{233}\\
&&\left.-\sum_{l=0}^{|({\Hat E}_{ti}^s\setminus {\Hat D}_{ti}^{s})\setminus (E_{ti}^{00}\cap M_{ts})|}
\left. C_{|({\Hat E}_{ti}^s\setminus {\Hat D}_{ti}^{s})\setminus (E_{ti}^{00}\cap M_{ts})|}^l\frac{(|C_{ti}|+|E_{ti}|-l-1)!(n-|C_{ti}|-|E_{ti}|+l)!}{n!}\right)\right].\label{234}\end{eqnarray}\label{t13}\end{theorem}
\begin{proof}Above all, due to the definition of the Shapley-Shubik index in (\ref{ss}) and the decomposition of ${\mathcal W}{\mathcal C}(i)$, we have
\begin{eqnarray}ss_i&=&\sum_{C\in{\mathcal W}{\mathcal C}1(i)
}\frac{(|C|-1)!(n-|C|)!}{n!}+\sum_{C\in{\mathcal W}{\mathcal C}2(i)
}\frac{(|C|-1)!(n-|C|)!}{n!}\\
&&+\sum_{C\in{\mathcal W}{\mathcal C}3(i)
}\frac{(|C|-1)!(n-|C|)!}{n!}.\label{228}\end{eqnarray}

(a1) Note first $p_i\in A_2(N)$ implies that \begin{equation}{\mathcal W}{\mathcal C}2(i)=\emptyset,\end{equation}
therefore the middle part of $ss_i$ is zero.

Let us consider the third part of $ss_i$, i.e. formula (\ref{228}). Suppose for the moment that $m_1>1$ and $m_2>1$.
 Due to Lemma \ref{l20}(a),  ${\mathcal W}{\mathcal C}3(i)$ consists of two parts. For the first part, a general coalition can be expressed as $A_1(N)\cup C\cup D$, where $\emptyset\subset C\subset A_2(N)$ and $D\subseteq M$. Let $c$ denote the cardinality of $C$ and $d$ the cardinality of $D$, then there are a total of $C_{m_2}^cC_{m}^d$ such combinations of $C$ and $D$ with these fixed cardinalities. The requirements that $1\leq c\leq m_2-1$ and $0\leq d\leq m$ are trivial. Therefore, the right hand side of (\ref{216}) is exactly the first part of ${\mathcal W}{\mathcal C}3(i)$. Due to Lemma \ref{l021}, Lemma \ref{l022}, and Lemma \ref{l023}, formula (\ref{217}) corresponds to the second part. It can be checked that formulas (\ref{216}) and (\ref{217}) are also valid for the other cases, i.e. the cases where at least one of $m_1$ and $m_2$ is 1.

 We are left the final part, i.e. $\sum_{C\in{\mathcal W}{\mathcal C}1(i)
}\frac{(|C|-1)!(n-|C|)!}{n!}$. As in the proof to Theorem \ref{t12}(a1), \begin{equation}{\mathcal W}{\mathcal C}1(i,M_{2s})=\bigcup_{p_j\in M_{2s}}{{\mathcal W}{\mathcal C}1_j},\label{210}\end{equation}
because $p_i\in A_1(N)$ implies that $p_i$ is decisive in each coalition of $\cup_{p_j\in M_{2s}}{{\mathcal W}{\mathcal C}1_j}$.
Due to Lemma \ref{l18}(a)(c), we know that (\ref{218})(\ref{219})(\ref{220}) correspond to the this final part.

Based on the above  discussions, part (a1) of this theorem is correct.

(a2) Symmetric to (a1).

(b) We prove by using similar arguments as in the proof to Theorem \ref{t12}(b). First of all, $p_i\in M$ implies that ${\mathcal W}{\mathcal C}3(i)=\emptyset$ (Lemma \ref{l012}). Hence $ss_i$ consists of two parts: contribution from  ${\mathcal W}{\mathcal C}1(i)$ and contribution from  ${\mathcal W}{\mathcal C}2(i)$. We shall prove that (the right hand side of) (\ref{226}), (\ref{227}), (\ref{228}), and the $t=1$ half of (\ref{232})(\ref{233})(\ref{234}) correspond to ${\mathcal W}{\mathcal C}1(i)$ (the rest formulas correspond to ${\mathcal W}{\mathcal C}2(i)$ by symmetry).

Due to Lemma \ref{l18}(c), (the right hand side of) (\ref{226}), (\ref{227}), (\ref{228}) correspond to all the contribution of  ${\mathcal W}{\mathcal C}1(i)$, except possibly for  ${\mathcal W}{\mathcal C}1(i, M_{1s_2(i)})$. This possibility is characterized by indicator $\lambda_{1i}$. When $\lambda_{1i}=1$, it can be checked that the (\ref{232}) part within the square brackets is the contribution from \begin{equation}\Big\{C^-\cup \{p_i\}: C\in {\mathcal B}1_i\Big\},\end{equation}(owing to Lemma \ref{l31}), while the (\ref{233}) and (\ref{234}) part is the contribution from  \begin{equation}\bigcup_{p_j\in M_{2s}}\Big\{C^-\cup \{p_i\}: C\in {\mathcal B}1_{ij}\Big\},\end{equation}owing to Lemma \ref{l32}(a)(b).
\qed\end{proof}

\subsection{Summary result}

\begin{theorem}Computing all the Penrose-Banzhaf indices and computing all the Shapley-Shubik indices can both be done in $O(n^3)$ time.\end{theorem}
\begin{proof}Using similar arguments as in the proof to Lemma \ref{l18}(e), computing all the necessary parameters can be done in $O(m^2)$ time. We omit the detailed analysis.

In the formulas of Theorem \ref{t12}, each elementary operation (addition, subtraction, multiplication, division, and the exponentiation with base 2), can be done in $O(m)$ time. For each player $p_i$, there are $O(m)$ such operations, so computing the Penrose-Banzhaf index of each player can be done in $O(m^2)$ time. Since there are $n$ players in total, and $m=O(n)$, we know that computing all the Penrose-Banzhaf indices can be done in $O(n^3)$ time.

In the formulas of Theorem \ref{t13}, the complication is that factorials are involved, and usually we cannot compute the exact values (this is not the case in Theorem \ref{t13}, though division is also involved there). We assume that the precision length is given as a constant, and thus computing a factorial, as well as  division and other elementary operations, can be done in constant time. For each player, to calculate the Shapley-Shubik  index we need to do $O(m^2)$ such operations. Since there are $n$ players in total, computing all the Penrose-Banzhaf indices can be done in $O(n^3)$ time.  Hence the theorem.\qed
\end{proof}
\section{List of Notations}
Notations listed below have already been defined the first time they appear, and the aim of this section is to help the reader find their meanings  more conveniently. The list, nevertheless, is not complete. Notations that are used only once are excluded. In particular, indicators, as a class, are not included, because most of them are used only once, right after their definitions. A more self-evident way to denote the indicators is to write them in the form of ``$Is(X)$", where $X$ is a statement, meaning that the value is 1 if statement $X$ is true and 0 otherwise. However, that would make the main formulas overly long. So we finally took the current more concise approach. Note also that for several notations that would be not clear enough and even confusing if an index $k\in \{1,2\}$ was introduced to  indicate which dimension we are concentrating on, we only listed a half of them. The other half is symmetric.

\begin{center}{\bf Basic Ones}\end{center}
\begin{itemize}
\item $\subseteq$: set inclusion
\item $\subset$: strict set inclusion, i.e. $X\subset Y$ means that $X\subseteq Y$ and $Y\subseteq X$ is not true
\item $N$: total set of players
\item $n$: total number of players, i.e. $n=|N|$
\item $C^-$: the complement of $C$, i.e. $C^-=N\setminus C$
\item $|\cdot|$: the cardinality of a set
\item $A_k(N)$: set of busy-$k$ players (of the grand coalition $N$), $k=1,2$
\item $m_k$: number of busy-$k$ players (of the grand coalition $N$),i.e. $m_k=|A_k(N)|$, $k=1,2$
\item $q_k(C)$: largest $k$-th dimension of coalition $C$, $k=1,2$
\item $q(C)=q_1(C)+q_2(C)$
\item $A_k(C)$: set of busy-$k$ players of coalition $C$, $k=1,2$
\item $B(C)$: set of busy players of  coalition $C$, i.e. $B(C)=A_1(C)\cup A_2(C)$
\item $M$: set of bench-warmers (of the grand coalition $N$), i.e. $M=N\setminus B(N)$
\item $m$: number of bench-warmers (of the grand coalition $N$), i.e. $m=|M|$
\item $D_{1i}=\Big\{p_j\in M: p_j^1+q_2(N)\geq p_i^2+q_1(N)\Big\}$
\item  $C_{1i}=A_1(N)\cup \{p_i\}\cup D_{1i}$
\item $n^0_2$: total number of players with distinct second dimensions
\item $\{M_{21}, M_{22},\cdots, M_{2n_2^0}\}$: partition of $M$ according to second dimensions
\item $C_c^d$: binomial coefficient (c choose d)\\
~~\\

\begin{center}{\bf Player Sets}\end{center}
\item $E_{1i}=\Big\{p_j\in M\setminus \big(D_{1i}\cup \{p_i\}\big):p_j^2\leq p_i^2\Big\}$

\item $D_1(M_{2s})=\Big\{p_j\in M: p_j^1+q_2(N)\geq q_1(N)+q_2(M_{2s})\Big\}$
\item $E_1(M_{2s})=\Big\{p_j\in M\setminus \big(D_{1}(M_{2s})\cup M_{2s}\big):p_j^2<q_2(M_{2s})\Big\}$
\item $\Delta_1=A_2(A_1(N))\cup \Big\{p_j\in M: p_j^2\geq q_2\big(A_1(N)\big)\Big\}$
\item $\Delta_1^*$: a maximum subset of $\Delta_1$ s.t. no pair of players have the same second dimension.
\item $M^+_{2j}=\Big\{p_k\in M: p_k^2>p_j^2\Big\}$
\item ${\Bar D}_{1j}=\Big\{p_k\in M: p_k^1+q_2(N)\leq q_1(N)+q^2_j\Big\}$

\item $E^0_{1i}=\Big\{p_j\in E_{1i}: p_j^1\geq q_1\big((C_{1i}\cup E_{1i})^-\cup
\{p_i\}\big)\Big\}$
\item ${\Hat E}_{1i}=\Big\{p_j\in E_{1i}: q\big((C_{1i}\setminus\{p_i\})\cup \{p_j\}\big)\leq
q\big((C_{1i}\cup E_{1i})^-\cup \{p_i\}\big)\Big\}$
\item ${\Hat D}_{1i}=\Big\{p_j\in E_{1i}: p_j^1>q_1\big((C_{1i}\cup E_{1i})^-\cup \{p_i\}\big)\Big\}$
\item ${\Hat E}_{1ij}=\Big\{p_k\in E_{1i}: k\neq j, q\big((C_{1i}\setminus\{p_i\})\cup \{p_k\}\big)\leq
p^1_{j}+q_2(N)\Big\}$
\item ${\Hat D}_{1ij}=\Big\{p_k\in E_{1i}: p_k^1>p_j^1\Big\}$
\item $E^{00}_{1i}=\Big\{p_j\in E_{1i}^0: {\Hat D}_{1ij}\subseteq {\Hat E}_{1ij}\Big\}$
\item ${\Hat D}^s_{1i}=\Big\{p_j\in E_{1i}: p_j>q_1(M_{1s})\Big\}$
\item ${\Hat E}_{1i}^s=\Big\{p_k\in E_{1i}: q((C_{1i}\setminus \{p_i\})\cup \{p_k\})\leq q_1(M_{1s})+q_2(N)\Big\}$
\item $M^{1*}=\Big\{1\leq s\leq n_2^0: q_2\big(M_{2s})\geq q_2\big(A_1(N)\cup D_1(M_{2s})\big)\Big\}$

\begin{center}{\bf Player Families}\end{center}
\item ${\mathcal M}{\mathcal W}{\mathcal C}$: set of minimal winning coalitions
\item ${\mathcal M}{\mathcal W}{\mathcal C}(j)$: set of minimal winning coalitions where $p_j$ is a member
\item ${\mathcal M}{\mathcal W}{\mathcal C}_j$: set of minimal winning coalitions where $p_j$ is busy
\item ${\mathcal W}{\mathcal C}$: set of winning coalitions
\item ${\mathcal W}{\mathcal C}(j)$: set of winning coalitions where $p_j$ is decisive
\item ${\mathcal W}{\mathcal C}_j$: set of winning coalitions where $p_j$ is busy
\item ${\mathcal M}{\mathcal W}{\mathcal C}1=\Big\{C\in
{\mathcal M}{\mathcal W}{\mathcal C}: A_1(N)\subseteq C, A_2(N)\cap C=\emptyset\Big\}$
\item ${\mathcal M}{\mathcal W}{\mathcal C}2=\Big\{C\in {\mathcal M}{\mathcal W}{\mathcal C}: A_2(N)\subseteq C, A_1(N)\cap C=\emptyset\Big\}$
\item ${\mathcal M}{\mathcal W}{\mathcal C}3=\Big\{C\in {\mathcal M}{\mathcal W}{\mathcal C}: A_1(N)\cap C\neq \emptyset, A_2(N)\cap C\neq
\emptyset\Big\}$
\item ${\mathcal M}{\mathcal W}{\mathcal C}1_{i}=\{C\in {\mathcal M}{\mathcal W}{\mathcal C}:
A_1(C)=A_1(N), p_i\in A_2(C)\}$

\item ${\mathcal W}{\mathcal C}1=\Big\{C\in {\mathcal W}{\mathcal C}: A_1(N)\subseteq
C, A_2(N)\cap C=\emptyset\Big\}$
\item ${\mathcal W}{\mathcal C}2=\Big\{C\in {\mathcal W}{\mathcal C}:
A_2(N)\subseteq C, A_1(N)\cap C=\emptyset\Big\}$
\item ${\mathcal W}{\mathcal C}3=\Big\{C\in {\mathcal W}{\mathcal
C}: A_1(N)\cap C\neq \emptyset, A_2(N)\cap C\neq \emptyset\Big\}$
\item ${\mathcal W}{\mathcal C}3(i)={\mathcal W}{\mathcal
C}3\cap {\mathcal W}{\mathcal C}(i)$
\item ${\mathcal W}{\mathcal C}1(i)={\mathcal W}{\mathcal
C}1\cap {\mathcal W}{\mathcal C}(i)$
\item ${\mathcal W}{\mathcal C}1(i,j)={\mathcal W}{\mathcal C}1(i)\cap {\mathcal W}{\mathcal C}1_j$
\item ${\mathcal W}{\mathcal C}1(i,M_{2s})=\bigcup_{p_j\in M_{2s}} {\mathcal W}{\mathcal C}1(i,j)$
\item ${\mathcal H}1_i=\Big\{X^-\cup \{p_i\}: X\in {\mathcal W}{\mathcal C}1(i,i)\Big\}$
\item ${\mathcal B}1_i=\Big\{X=(C_{1i}\cup E_{1i})^-\cup \{p_i\}\cup
E: E\subseteq E_{1i},X\in {\mathcal H}1_i, A_1(X)\cap E=\emptyset\Big\}$
\item ${\mathcal B}1_{ij}=\Big\{X=(C_{1i}\cup E_{1i})^-\cup \{p_i\}\cup E: E\subseteq
E_{1i}, X\in {\mathcal H}1_i, p_j\in A_1(X)\cap E\Big\}$

\begin{center}{\bf Others}\end{center}
\item $l(i)$: the second index of $p_i\in M$
\item $p_{l^{-1}(i)}$: the player in $M$ whose second index is $i$
\item $x(i)=\max\Big\{j: p_{j}\in D_{1i}\Big\}$
\item $y(i)=\max\Big\{l(j): p_j\in D_{2i}\Big\}$
\item $R_1(i)=\min\Big\{l(j):1\leq j\leq i\Big\}$
\item $r_1(i)=\min\Big\{l(j):1\leq j\leq i,l(j)\neq R_1(i)\Big\}$
\item $n_1$: number of coalitions in ${\mathcal M}{\mathcal W}{\mathcal C}1$, i.e. $n_1=|{\mathcal M}{\mathcal W}{\mathcal C}1|$
\item $s_2(i)$: the index $s$ in $\{1,2,\cdots,n_2^0\}$ such that $p_i\in M_{2s}$

\end{itemize}
\end{appendix}
\end{document}